\documentclass[english,11pt,a4paper]{article}

\usepackage{ifdraft}
\usepackage{pgf}
\usepackage{xspace}

\usepackage[english]{babel}
\usepackage{geometry}
\usepackage{makeidx}
\usepackage{tabto}
\usepackage{fancyhdr}
\usepackage{multicol}

\usepackage{marvosym}
\usepackage{amsfonts}
\usepackage{amssymb}
\usepackage{stmaryrd}
\usepackage{braket}
\usepackage{bbold}
\usepackage{bbm}

\usepackage{braket}
\usepackage{extarrows}
\usepackage{colonequals}

\usepackage{mathtools}
\usepackage{amsmath}
\usepackage{amsthm}
\usepackage{cases}
\usepackage{multirow, bigdelim}
\usepackage{nicefrac}
\usepackage{relsize}

\usepackage{algorithm}
\usepackage{algorithmic}

\usepackage{array}
\usepackage{enumerate}
\usepackage[inline]{enumitem}
\usepackage{graphicx}
\usepackage{caption}
\usepackage{subcaption}
\usepackage{tikz}
\usepackage{tkz-graph}
\usepackage{float}
\usepackage{pdfpages}
\usepackage[nopgf,nographicx,nobookmark]{incgraph}
\usepackage{longtable}
\usepackage{array}
\usepackage{multirow}
\usepackage{makecell}

\usepackage{hyperref}

\usepackage{tocbasic}
\usepackage{cite}
\usepackage{url}

\usepackage{xcolor,color, colortbl}

\usepackage{marginnote} 
\usepackage{textcomp}

\allowdisplaybreaks[1]
\graphicspath{ {images/} }

\usetikzlibrary{fit}
\usetikzlibrary{arrows}
\usetikzlibrary{petri}
\usetikzlibrary{topaths}
\usetikzlibrary{positioning}
\usetikzlibrary{patterns}
\usetikzlibrary{decorations.pathreplacing}
\usetikzlibrary{calc}
\usetikzlibrary{shapes}
\usetikzlibrary{trees}
\usetikzlibrary{intersections}
\usetikzlibrary{backgrounds}
\tikzset{>=latex'}

\theoremstyle{plain}
\newtheorem{theorem}             {Theorem}[section]
\newtheorem{lemma}      [theorem]{Lemma}
\newtheorem{corollary}  [theorem]{Corollary}

\theoremstyle{definition}
\newtheorem{definition} [theorem]{Definition}

\makeatletter
\newtheorem*{rep@theorem}{\rep@title}
\newcommand{\newreptheorem}[2]{%
\newenvironment{rep#1}[1]{%
 \def\rep@title{#2 \ref{##1}}%
 \begin{rep@theorem}}%
 {\end{rep@theorem}}}
\makeatother

\newreptheorem{theorem}{Theorem}

\newcommand{\Nat}{\ensuremath{\mathbb{N}}}
\newcommand{\Rel}{\ensuremath{\mathbb{R}}}

\newcommand{\abs}[1]{| #1 |}

\DeclareMathOperator{\lcm}{lcm}
\newcommand{\reverseFunction}[1]{\ensuremath{#1^{-1}}}

\newcommand{\restr}[2]{
	\ensuremath{
		\left.\kern-\nulldelimiterspace
		#1
		\vphantom{\big|}
		\right|_{#2}
	}
}

\newcommand{\disjointUnion}{\ensuremath{\mathbin{\dot{\cup}}}}

\title{
	An Upper Bound on the Weisfeiler-Leman Dimension\footnote{The research leading to these results has received funding from the European Research Council (ERC) under the European Union’s Horizon 2020 research and innovation programme (EngageS: grant agreement No.\ 820148).}
}
\author{%
	Thomas Schneider and Pascal Schweitzer\\
	TU Darmstadt
}

\makeatletter
\def\@maketitle{%
  \newpage
  \null
  \vskip 2em%
  \begin{center}%
  \let \footnote \thanks
    {\LARGE \@title \par}%
    \vskip 1.5em%
    {\large \@author}%
    \vskip 1.5em%
    {\large \@date}%
  \end{center}%
  \par
  \vskip 1.5em}
\makeatother

\definecolor{darkgreen}{rgb}{0,0.6,0}
\definecolor{darkred}{RGB}{128, 0, 0}
\definecolor{darkblue}{RGB}{51, 204, 204}
\definecolor{darkyellow}{RGB}{204,204,0}
\definecolor{fuchsia}{RGB}{255,0,255}

\definecolor{lightblue}{RGB}{173,216,230}
\definecolor{lightred}{RGB}{233,150,122}
\definecolor{lightyellow}{RGB}{250,250,210}
\definecolor{lightgray}{RGB}{198,198,198}

\definecolor{hrefblue}{rgb}{0.5,0.5,1.0}
\definecolor{hrefred}{rgb}{0.5,0,0}
\definecolor{hrefgreen}{rgb}{0,0.5,0}
\definecolor{hrefblue}{rgb}{0,0,0.5}
\definecolor{labelkey}{RGB}{51, 204, 204}

\definecolor{lavender}{RGB}{181,126,220}
\definecolor{sage}{RGB}{188,184,138}

\hypersetup{breaklinks=true}
\hypersetup{colorlinks,linkcolor=hrefred,filecolor=hrefgreen,urlcolor=hrefred,citecolor=hrefblue}

\setlist[enumerate]{label=(\arabic*)}

\def\env@matrix{\hskip -\arraycolsep
	\let\@ifnextchar\new@ifnextchar
	\array{*\c@MaxMatrixCols c}}

\makeatletter
\renewcommand*\env@matrix[1][c]{\hskip -\arraycolsep
	\let\@ifnextchar\new@ifnextchar
	\array{*\c@MaxMatrixCols #1}}
\makeatother

\DeclareMathOperator{\Fibers}{F}
\newcommand{\coherentConfig}{\ensuremath{\mathfrak{X}}}
\newcommand{\fibers}[1]{\ensuremath{\Fibers \left( #1 \right)}}
\newcommand{\interspace}[2]{\ensuremath{\coherentConfig[#1,#2]}}
\newcommand{\inducedCC}[1]{\ensuremath{\coherentConfig[#1]}}
\newcommand{\intDegree}[1]{\ensuremath{\Deg(#1)}}
\newcommand{\minimalDegree}[2]{\Deg(#1,#2)}
\DeclareMathOperator*{\ul}{ul}

\DeclareMathOperator*{\Type}{T}
\newcommand{\type}[1]{\ensuremath{\Type\left(#1\right)}}

\newcommand{\vertices}{\ensuremath{\Omega}}
\newcommand{\relations}{\ensuremath{\mathcal{A}}}
\newcommand{\arcs}{\ensuremath{A}}

\DeclareMathOperator*{\WLdim}{WLdim}
\newcommand{\wldim}[1]{\ensuremath{\WLdim\left(#1\right)}}

\newcommand{\finer}{\preccurlyeq}
\newcommand{\wld}[1]{\texttt{WL#1}}
\newcommand{\wltwo}{\wld{2}\xspace}
\newcommand{\wldstable}[2]{\ensuremath{#2^{\wld{#1}}_{\infty}}}

\DeclareMathOperator*{\Quotient}{Q}
\newcommand{\quotientGraph}[1]{\ensuremath{\Quotient(#1)}}
\newcommand{\quotientGraphLarge}[1]{\ensuremath{\Quotient^L(#1)}}
\newcommand{\quotientGraphSmall}[1]{\ensuremath{\Quotient^S(#1)}}
\DeclareMathOperator{\ColorDeg}{qdeg}
\DeclareMathOperator{\ColorDegLarge}{qdeg_L}
\DeclareMathOperator{\ColorDegSmall}{qdeg_S}
\DeclareMathOperator{\ColorDegRelevantSmall}{qdeg_{rS}}
\newcommand{\colorDeg}[1]{\ensuremath{\ColorDeg(#1)}}
\newcommand{\colorDegLarge}[1]{\ensuremath{\ColorDegLarge(#1)}}
\newcommand{\colorDegSmall}[1]{\ensuremath{\ColorDegSmall(#1)}}
\newcommand{\colorDegRelevantSmall}[1]{\ensuremath{\ColorDegRelevantSmall(#1)}}

\newcommand{\coloring}{\ensuremath{\chi}}
\newcommand{\coloredGraph}{\ensuremath{(G,\coloring)}}
\DeclareMathOperator{\Deg}{d}

\newcommand{\degree}[3]{\Deg_{#1}^{#3}\left(#2\right)}

\DeclareMathOperator{\CFI}{CFI}
\newcommand{\cfi}[1]{\CFI \left(#1\right)}

\newcommand{\partition}[1]{\ensuremath{\mathcal{P}^1_1(#1)}}
\newcommand{\partitionRel}[3]{\ensuremath{\mathcal{P}^{#1}_{#2}(#3)}}
\newcommand{\equivalenceClasses}[1]{\ensuremath{\mathcal{P}(#1)}}
\newcommand{\partitionStructure}[1]{\ensuremath{\mathfrak{S}(#1)}}

\DeclareMathOperator{\treewidth}{tw}
\DeclareMathOperator{\pathwidth}{pw}

\newcommand{\f}{f}
\DeclareMathOperator{\hfunction}{h}
\DeclareMathOperator{\parameters}{Par}

\newcommand{\interspacePattern}[1]{\ensuremath{\llbracket#1\rrbracket}}
\newcommand{\ipfourClique}  {\ensuremath{\interspacePattern{\clique{4},2}}}
\newcommand{\ipfourMatching}{\ensuremath{\interspacePattern{\disjointCliques{2}{2},2}}}
\newcommand{\ipfourCycle}   {\ensuremath{\interspacePattern{\cycle{4},2}}}

\newcommand{\ipsixCliqueTwo}     {\ensuremath{\interspacePattern{\clique{6},2}}}
\newcommand{\ipsixCliqueTwoTwice}{\ensuremath{\interspacePattern{\clique{6},2,2}}}
\newcommand{\ipsixCliqueThree}   {\ensuremath{\interspacePattern{\clique{6},3^\dag}}}
\newcommand{\ipsixCliqueThreeD}	 {\ensuremath{\interspacePattern{\clique{6},3^\ddag}}}

\newcommand{\ipsixMatching}             {\ensuremath{\interspacePattern{\disjointCliques{3}{2},2}}}
\newcommand{\ipsixMatchingTwice}        {\ensuremath{\interspacePattern{\disjointCliques{3}{2},2,2}}}
\newcommand{\ipsixMatchingMatching}     {\ensuremath{\interspacePattern{\disjointCliques{3}{2},2;\disjointCliques{3}{2},2}}}
\newcommand{\ipsixMatchingAndCycle}     {\ensuremath{\interspacePattern{\cycle{6},2;\disjointCliques{3}{2},2}}}
\newcommand{\ipsixMatchingAndComplement}{\ensuremath{\interspacePattern{\clique{2,2,2},2;\disjointCliques{3}{2},2}}}

\newcommand{\ipsixTriangle}               {\ensuremath{\interspacePattern{\disjointCliques{2}{3},3^\dag}}}
\newcommand{\ipsixTriangleComplement}     {\ensuremath{\interspacePattern{\clique{3,3},2}}}
\newcommand{\ipsixTriangleComplementTwice}{\ensuremath{\interspacePattern{\clique{3,3},2,2}}}

\newcommand{\ipsixMatchingComplement} {\ensuremath{\interspacePattern{\clique{2,2,2},3^\dag}}}
\newcommand{\ipsixMatchingComplementD}{\ensuremath{\interspacePattern{\clique{2,2,2},3^\ddag}}}

\newcommand{\matchingCC}[1]{\ensuremath{\disjointCliques{#1}{2}}}
\newcommand{\clique}[1]{\ensuremath{K_{#1}}}
\newcommand{\cycle}[1]{\ensuremath{C_{#1}}}
\newcommand{\disjointCliques}[2]{\ensuremath{#1 \clique{#2}}}
\newcommand{\disjointCycles}[2]{\ensuremath{#1 \cycle{#2}}}
\DeclareMathOperator{\fanoPlane}{FP}
\DeclareMathOperator{\LeviGraph}{L}
\newcommand{\leviGraph}[1]{\ensuremath{\LeviGraph\!\left(#1\right)}}
\newcommand{\leviFano}{\leviGraph{\fanoPlane}}
\newcommand{\rookGraph}[1]{\ensuremath{R_{#1}}}

\newcommand{\matching}[1]{\ensuremath{#1 K_{1,1}}}
\newcommand{\interspaceFourSix}{\ensuremath{Sp_{4,6}}}

\tikzstyle{vertex}=[circle,draw,minimum size=.2mm]
\tikzstyle{empty}=[]
\tikzstyle{edge}=[draw,very thick]
\tikzstyle{arrow}=[draw,very thick,->]
\tikzstyle{rectangle}=[thick]

\makeatletter
\newcommand{\thickhline}{%
    \noalign {\ifnum 0=`}\fi \hrule height 1pt
    \futurelet \reserved@a \@xhline
}
\newcolumntype{"}{@{\hskip\tabcolsep\vrule width 1pt\hskip\tabcolsep}}
\makeatother

\begin{document}
    \maketitle
    \thispagestyle{empty}
    \begin{abstract}
        The Weisfeiler-Leman (WL) algorithms form a family of incomplete approaches to the graph isomorphism problem. They recently found various applications in algorithmic group theory and machine learning.
        In fact, the algorithms form a parameterized family: for each~$k \in \Nat$ there is a corresponding $k$-dimensional algorithm~$\wld{k}$.
        The algorithms become increasingly powerful with increasing dimension, but at the same time the running time increases.
        The WL-dimension of a graph~$G$ is the smallest $k \in \Nat$ for which~$\wld{k}$ correctly decides isomorphism between~$G$ and every other graph.
        In some sense, the WL-dimension measures how difficult it is to test isomorphism of one graph to others using a fairly general class of combinatorial algorithms. Nowadays, it is a standard measure in descriptive complexity theory for the structural complexity of a graph.

        We prove that the WL-dimension of a graph on~$n$ vertices is at most $3/20 \cdot  n + o(n) = 0.15 \cdot n + o(n)$.

        Reducing the question to coherent configurations, the proof develops various techniques to analyze their structure. This includes sufficient conditions under which a fiber can be restored uniquely up to isomorphism if it is removed, a recursive proof exploiting a degree reduction and treewidth bounds, as well as an exhaustive analysis of interspaces involving small fibers.

        As a base case, we also analyze the dimension of coherent configurations with small fiber size and thereby graphs with small color class size.
    \end{abstract}

    \addtocounter{page}{-1}
    \newpage


\section{Introduction}
\label{sec:intro}

In recent years, the Weisfeiler-Leman (WL) dimension has evolved to become a standard measure for the structural complexity of a graph. Initially coined in the context of isomorphism questions\cite{Ba79b,MR0543783,MR1060782}, a plethora of equivalent reformulations in seemingly unrelated areas has surfaced (e.g.~\cite{DBLP:journals/combinatorica/CaiFI92,DBLP:journals/jgt/Dvorak10,DBLP:journals/siamcomp/AtseriasM13,DBLP:journals/jsyml/GroheO15,DBLP:conf/icalp/DellGR18,DBLP:journals/jct/AtseriasMRSSV19}). The concept in particular has applications in machine learning on graphs (see~\cite{Morrisetal2023} for a survey).

In its initial formulation, the WL-dimension of a graph~$G$ is  characterized as the minimum~$k$ required so that the~$k$-dimensional WL-algorithm distinguishes~$G$ from every non-isomorphic graph.
By a central result of Cai, F\"{u}rer, and Immerman~\cite[Theorem~5.2]{DBLP:journals/combinatorica/CaiFI92}, the dimension plus one is also the least number of variables required to identify~$G$ in a particular logic (fixed-point logic with counting) and also the number of pebbles required in a particular combinatorial pebble game (the bijective pebble game).

In some sense, the WL-dimension measures how difficult it is to test isomorphism of one graph to others using a fairly general class of combinatorial algorithms.
Crucially, isomorphism between graphs of bounded WL-dimension can be decided in polynomial time. More precisely, if the WL-dimension is at most~$k$, then the problem can be solved in time~$O(n^{k+1}\log n)$.
While group theoretic techniques can circumvent the structural complexity given by high WL-dimension, the currently fastest theoretical algorithm, which runs in quasi-polynomial time~\cite{DBLP:conf/stoc/Babai16}, nevertheless uses a~$O(\log(n))$-dimensional WL-algorithm as a subroutine.

Many graph classes are known to have bounded WL-dimension (e.g., all classes with an excluded minor~\cite{DBLP:books/cu/G2017} and bounded clique width graphs~\cite{DBLP:journals/tocl/GroheN23}), giving a polynomial time isomorphism algorithm for these classes.

Initial hopes to find a general bound on the WL-dimension
of graphs, however, were dispelled by the seminal construction of  Cai, F\"{u}rer, and Immerman~\cite{DBLP:journals/combinatorica/CaiFI92} which constructs graphs of order~$n$ and WL-dimension~$\Omega(n)$.

In this paper we investigate explicit bounds for the WL-dimension. A calculation of the precise constant in~\cite{DBLP:journals/combinatorica/CaiFI92} yields a lower bound of $0.00465\cdot n$ as demonstrated in~\cite{DBLP:journals/dam/PikhurkoVV06,DBLP:conf/asl/PikhurkoV09}.
In a publication independent from ours, Kiefer and Neuen~\cite{DBLP:journals/corr/abs-2402-03274} observe a lower bound of~$\frac{n}{96} - o(n)$.
This constitutes the current best lower bound in the literature and is slightly below the one we state in this paper using essentially the same observations.
In~\cite{DBLP:journals/dam/PikhurkoVV06,DBLP:conf/asl/PikhurkoV09}, an explicit upper bound of~$0.5 n + 1.5$ follows from upper bounds that apply in a more general context (more specifically from bounds on the ``non-counting version'' of the WL-algorithm).
We previously supervised a Bachelor's thesis by Simon Lutz at TU Kaiserslautern that shows an upper bound of~$\lceil\frac{n}{3}\rceil+2$.
The best currently known upper bound on the WL-dimension of~$\lceil\frac{n}{4}\rceil+o(n)$ is proven by Kiefer and Neuen~\cite{DBLP:journals/corr/abs-2402-03274} in their mentioned independent work.

\paragraph*{Contributions.}
The main result of this paper establishes an upper bound for the WL-dimension of~$0.15 \cdot  n + o(n)$ for all~$n$-vertex graphs~$G$.
As part of the proof, we also derive an upper bound of~$0.05 \cdot n + o(n)$ on the WL-dimension of colored graphs whose color classes all have size at most~$7$.
Possibly more important than the precise value of our bounds are the techniques we develop in the proof, allowing for the analysis of the descriptive complexity of graphs.
In brief, they include criteria under which color classes are removable without any effect on the WL-dimension, an analysis of the possible structures joining an arbitrarily color class with a color class of size at most~$7$, and a general framework to bound the WL-dimension with recursive reductions using potential functions.
These techniques are described in more detail below.

As a complementary result, we also observe in this paper that for all orders~$n$ there are graphs~$G$ with a WL-dimension of at least~$ 0.0105027 \cdot n - o(n)$.

As usual, these upper and lower bounds for the WL-dimension can both be recast in logical terms, bounding the number of variables required for graph identification in fixed point-logic with counting.

\paragraph*{Techniques.} On a macroscopic scale the idea for the upper bound proof is as follows. To facilitate recursion, we first generalize the problem to vertex- and edge-colored graphs. More specifically, it is sufficient for us to consider vertex- and edge-colored complete graphs that satisfy strong regularity conditions. In particular, it suffices to consider so-called coherent configurations, which naturally generalize various highly regular graph families, such as strongly regular graphs and distance-regular graphs. They can also be understood as the stable colorings under the 2-dimensional WL-algorithm. By increasing the dimension by at most 2 once, we can at any point in time, even during recursion, assume that all our objects are coherent configurations. In the language of coherent configurations, the notions of vertex color classes translate into so-called \emph{fibers} of the coherent configuration, and we will use terms interchangeably.

The proof strategy is to reduce the number of vertices, with a focus on reducing those contained in large fibers. We repeatedly use individualizations: an individualization artificially assigns a single vertex a different color. The effect on the coherent configuration, when the 2-dimensional WL-algorithm is applied, is that other fibers are split into smaller ones, which simplifies the coherent configuration. A vertex individualization decreases the dimension by at most 1, meaning we can bound the dimension of the configuration without individualization in terms of the dimension of the configuration we obtain after individualization.
We use this strategy to argue that, at a sublinear cost regarding the dimension, we can ensure that the maximum color class size is sublinear in the number of vertices.

With small fibers in mind, we investigate situations in which a color class can be removed without decreasing the WL-dimension. We call configurations \emph{critical} if no fiber can be removed this way. Fibers of size at most~$3$, which we call \emph{tiny}, can always be removed.
For fibers that are not tiny, we develop a technique to determine from the combinatorial structure that a fiber is \emph{restorable} and thus can be removed. Formally, the technique is based on extendability of automorphisms of induced subconfigurations.

Our base case of the recursion is the situation in which all fibers have size at most 7. We call such fibers \emph{small}. We analyze the structure of the configurations that small fibers can induce, as well as the possible connections between small fibers. These connections are called \emph{interspaces}. The \emph{quotient graph} captures the structural information given by fibers and interspaces. The vertices of the quotient graph are the vertex color classes (i.e.,~fibers). Two of these vertices are connected by an edge if the corresponding color classes are not homogeneously connected, that is, if the interspace is not trivial. Intuitively this means that there is some form of structural dependence between the color classes. For example, the connected components of the quotient graph can be treated separately when determining the WL-dimension, since they are structurally independent.

Several reductions lead us to a quotient graph of maximum degree at most 3. This means that every color class is non-trivially connected to at most 3 other color classes. At this point we can use bounds on the treewidth for cubic graphs to bound the WL-dimension.
Overall we show that a coherent configuration with fiber size at most~$7$ has WL-dimension at most~$0.05 \cdot n + o(n)$.

For the general recursion we also need to understand the possible interspaces between small and large fibers.
We define a specific potential function that measures the progress we make towards the base case. It gives us the possibility of trading individualizations for a reduction of the potential. We then define a sequence of~\emph{local reductions} that, for various subconfigurations and types of interspace, provide a positive trade-off. We can thus inductively assume that these subconfigurations and interspaces are not present in our configuration.

To finally reach the base case, we employ a global argument concerning the structure of configurations that avoid the subconfigurations. In more detail, we introduce the concept of a~\emph{$t$-reduced} configuration and show that configurations to which none of the local reductions are applicable are~$t$-reduced.
For reduced structures, the global argument allows us to separate the graph into pieces whose underlying structure either has small treewidth or which consist only of small fibers.
Overall this recursive approach proves the main theorem.

The lower bound simply follows straightforwardly by combining three known results on expansion, treewidth, and the CFI-construction in the evident fashion. An intermediate step in that argument is that  random cubic graphs asymptotically almost surely have treewidth at least~$0.04201\cdot n$.

\paragraph*{Structure.}
We revisit the most important definitions and concepts in Section~\ref{preliminaries/sec} and give details on the lower bound in the fairly compact Section~\ref{lower-bound/sec}.

The proof of the upper bound, in contrast, turned out to be significantly more involved.
A detailed outline of the overall proof, the involved ideas, and the new techniques is given in Section~\ref{upper-bound/sec}.
However, it omits some repetitive case distinctions, various tedious calculations, and technical aspects.
These can be found throughout Sections~\ref{critical-graph/sec} to~\ref{sec:proof:of:main:thm}:
the concepts of criticality and restorability are introduced in Section~\ref{critical-graph/sec}.
Sections~\ref{small-cc/sec} and~\ref{interspace-large-small/sec} focus on interspaces between small fibers and interspaces between large and small fibers, respectively.
Using these insights, Section~\ref{critical:restorable/sec} examines the restorability of interspaces.
Building on the previous results, Section~\ref{wldim-small/sec} proves the bound for WL-dimension of coherent configurations with fiber size at most~$7$.
Section~\ref{sec:limit:fiber:sizes} provides arguments to limit the fiber size.
The mechanisms of the local reductions are presented in Section~\ref{sec:potential:func}, and Section~\ref{recursive-argement/sec} contains the collection of necessary local reductions.
Sections~\ref{structure-reduced-cc/sec} and~\ref{global-argument/sec} define the concept of~$t$-reducedness and provide a bound on the WL-dimension of~$t$-reduced coherent configurations.
Finally, Section~\ref{sec:proof:of:main:thm} proves the upper bound by combining the previous results.

\paragraph*{Related Work.}  A concrete classification of graphs with WL-dimension 1 is known \cite{DBLP:journals/cc/ArvindKRV17,DBLP:journals/tocl/KieferSS22}.
Fuhlbr{\"{u}}ck, K{\"{o}}bler, and Verbitsky analyze the structure of
graphs with WL-dimension~2~\cite{DBLP:journals/siamdm/FuhlbruckKV21} and bounded color class size. Some of our structural lemmas can be seen as direct generalizations of results in their paper.
A recent generalization of their complexity results regarding the WL-dimension can be found in~\cite{DBLP:conf/csl/LichterRS25}. We should remark that in the two papers, just like in ours, the CFI-graphs appear innately.

A survey on descriptive complexity in particular with bounds related to the{\linebreak}WL-dimension can be found in~\cite{DBLP:conf/asl/PikhurkoV09}.
The term WL-dimension was coined by Grohe in his monograph~\cite{DBLP:books/cu/G2017}. The main result of this monograph implies that for non-trivial minor-closed graphs classes the WL-dimension is bounded. As remarked above, this is also true for graphs of bounded rank-width~\cite{DBLP:journals/tocl/GroheN23}. In recent years, for several graph classes explicit bounds on the dimension have been proven, including planar graphs~\cite{DBLP:journals/jacm/KieferPS19}, distance hereditary graphs~\cite{DBLP:journals/gc/GavrilyukNP23}, interval graphs~\cite{Ponomarenko2000interval}, permutation graphs~\cite{DBLP:journals/corr/abs-2305-15861}, and circulant graphs~\cite{wu2024weisfeiler}.

\paragraph*{Future Work.}
While our overall approach may very well be used to improve the upper bound even further, it seems that our argument involving interspaces will then produce an ever increasing number of cases. Therefore, to further improve the bounds with our approach, these arguments might need to be automated or structurally simplified.

Having developed various new techniques and a global view to deal with coherent configurations, another pressing questions for us is as follows. The Deep-Weisfeiler-Leman framework is an extension of the WL-algorithm to match choiceless polynomial time~\cite{DBLP:conf/soda/GroheSW21}. For this extension it is an open problem whether it provides a polynomial-time solution to graph isomorphism. This question in turn has consequences for the quest for a logic capturing polynomial time~\cite{DBLP:journals/jacm/LichterS24}, a central problem in finite model theory. We therefore wonder whether some of our new techniques can be used to provide better insight into Deep-Weisfeiler-Leman.


\section{Preliminaries}
\label{preliminaries/sec}

Our overall proof makes ample use of coherent configurations, for which the basic concepts and notation are introduced in this section (see also~\cite{DBLP:conf/stoc/Babai16,CC} for the broader theory).

Let~$\vertices$ be a finite set and~$A$ be a binary relation on~$\vertices$.
We set~$n \coloneqq \abs{\vertices}$.
Throughout the paper, we call the elements in~$\vertices$ \emph{vertices} and the elements in~$A$ \emph{arcs}.
We also write~$vw$ instead of~$(v,w)$ to denote arcs.

We call an arc~$vw$ a~\emph{self-loop} if~$v = w$, and interpret a self-loop~$vv$ often as its corresponding vertex~$v$.
If we want to emphasize that both~$vw, wv \in \arcs$, then we refer to them as~\emph{edge}.
We set~$A^\star \coloneqq \{wv \mid vw \in A\}$ and call it the~\emph{transpose relation} of~$A$.
For~$v \in \Omega$, the set~$vA \coloneqq \{w \in \vertices \mid vw \in A\}$ is called the \emph{neighborhood} of~$v$ under~$A$ and we define~$\degree{}{v}{A} \coloneqq \abs{vA}$.
Given a set~$\Delta \subseteq \vertices$, we also use the notation $\degree{}{\Delta}{A}$ if~$\degree{}{v}{A}$ is independent of the choice of~$v \in \Delta$.
We denote~$\{(v,v) \mid v \in \vertices\}$ by~$1_\vertices$.

A~\emph{coloring of a set~$A$} is a function~$\coloring \colon A \to C$ where~$C$ is the set of~\emph{colors}.
The~\emph{color} of~$a \in A$ is~$\coloring(a)$ and the set of all colors is~$\coloring(A)$.
The coloring induces a~\emph{color partition~$\pi(\coloring)$} on~$A$.
For a partition~$\mathcal{P}$ of a finite set~$\Omega$, we denote~$\{\bigcup_{P \in \mathcal{P}'} P \mid \mathcal{P}' \in 2^\mathcal{P} \}$ by~$\mathcal{P}^\cup$.

\paragraph{Graphs.}
Given a finite set~$\vertices$ and a binary relation~$A$, we call the pair~$(\vertices, \arcs)$ a~\emph{(directed) graph~$G$}.
We denote the set of all vertices of~$G$ by~$\vertices(G)$ and the set of all arcs of~$G$ by~$\arcs(G)$.
We call~$G$~\emph{undirected} if~$\arcs = \arcs^\star$.
Two vertices~$v,w \in \vertices$ are~\emph{adjacent} if~$vw \in \arcs$ or~$wv \in \arcs$.
Given~$\Delta \subseteq \vertices$, the~\emph{subgraph of~$G$ induced by~$\Delta \subseteq \vertices$} is~$(\Delta, \arcs \cap \Delta^2)$ and is denoted by~$G[\Delta]$.
Given the graph~$G$ and a coloring~$\coloring \colon \arcs(G) \to C$, we call~$(G,\coloring)$ a~\emph{colored graph}.

\paragraph{Isomorphisms.}
An~\emph{isomorphism} between uncolored graphs~$G$ and $H$ is a bijection~$\varphi: V(G) \to V(H)$ which preserves adjacency and non-adjacency, that is, for all~$v,w \in \vertices(G)$ we have~$vw \in \arcs(G)$ if and only if~$\varphi(v)\varphi(w) \in \arcs(H)$.
Given two colored graphs~$(G,\coloring_G)$ and~$(H,\coloring_H)$, an isomorphism~$\varphi$ between~$G$ and~$H$ is called~\emph{color-permuting} if for all~$v,w,v',w' \in \vertices(G)$ it satisfies
\begin{equation*}
\label{eq:isomorphic-colored}
    \coloring_G(vw) = \coloring_G(v'w') \iff \coloring_H(\varphi(v)\varphi(w)) = \coloring_H(\varphi(v')\varphi(w')),
\end{equation*}
and~\emph{color-preserving} if for all~$v,w \in \vertices(G)$ we have~$\coloring_G(vw) = \coloring_H(\varphi(v)\varphi(w))$.
If there exists an isomorphism between~$G$ and~$H$, then we call the graphs~\emph{isomorphic} and write~$G \cong H$.
Unless otherwise stated, we require all isomorphism between colored graphs to be color-preserving.

\paragraph{Weisfeiler-Leman Algorithm.}
Given a positive integer~$k$, the~\emph{$k$-dimensional Weis\-feiler-Leman algorithm~\wld{k}} is part of a family of incomplete deciders for the \emph{isomorphism problem}, which, given two uncolored graphs~$G$ and~$H$, asks whether~$G \cong H$ holds.

In a nutshell, the algorithm colors the~$k$-tuples of vertices of the graphs as follows. Initially it colors each~$k$-tuple depending on the isomorphism type of the graph induced by the tuple. Here the order of the vertices is taken into account. The algorithm then repeatedly refines the coloring by considering the possible tuples and colors one obtains by replacing one of the vertices in the tuple by an arbitrary other vertex. During this process, it bundles pieces of information together when possible.

In more detail, given the colored graph~$\coloredGraph$, the algorithm~\wld{k} determines for every~$(v_1, \dots, v_k) \in \vertices(G)^k$ an initial coloring
\begin{equation*}
    \coloring^{\wld{k}}_0(v_1, \dots, v_k) \coloneqq (\vartheta(v_1v_1),\vartheta(v_1v_2),  \dots, \vartheta(v_kv_k))
\end{equation*} where
\begin{equation*}
    \vartheta(vw) \coloneqq
    \begin{cases}
        (0,\coloring(vw)) & \text{if } v = w, \\
        (1,\coloring(vw)) & \text{if } v \neq w  \text{ and } vw \in \arcs(G), \\
        (2,0) & \text{if } v \neq w  \text{ and }  vw \notin \arcs(G).
    \end{cases}
\end{equation*}
Next, it iteratively computes~$\coloring^{\wld{k}}_{i+1}(v_1,\dots,v_k)$ defined by
\begin{equation*}
    (
        \coloring^{\wld{k}}_{i}(v_1,\dots,v_k),
        \{\!\!\{
            (\coloring^{\wld{k}}_{i}(w,v_2,\dots,v_k),\dots,\coloring^{\wld{k}}_{i}(v_1,\dots,v_{k-1},w)) \mid w \in \vertices(G)
        \}\!\!\}
    )
\end{equation*}
for all~$(v_1, \dots v_k) \in \vertices(G)^k$.
This process stops if~$\pi(\coloring^{\wld{k}}_{i})$ is~\emph{stable under~\wld{k}}, that is~$\pi(\coloring^{\wld{k}}_{i}) = \pi(\coloring^{\wld{k}}_{i + 1})$.
For the~$i$ at which the process stops, we define~$\wldstable{k}{\coloring}\coloneqq \coloring^{\wld{k}}_{i+1}$.

If the initially given graph is uncolored, we start the algorithm on the monochromatic version of it.
The algorithm~$\wld{k}$ \emph{distinguishes} graphs~$(G,\coloring_G)$ and~$(H,\coloring_H)$ if~$\{\!\!\{ \wldstable{k}{\coloring}(\overline{v}) \mid \overline{v} \in \vertices(G)^k \}\!\!\} \neq \{\!\!\{ \wldstable{k}{\coloring}(\overline{v}) \mid \overline{v} \in \vertices(H)^k \}\!\!\}$.
The notation~$(G,\coloring_G) \simeq_k (H,\coloring_H)$ indicates that the graphs are not distinguished by~$\wld{k}$.
The algorithm~\wld{k}~\emph{identifies}~$(G,\coloring_G)$ if it distinguishes the graph~$(G,\coloring_G)$ from all non-isomorphic graphs.
The~\emph{Weisfeiler-Leman dimension}~$\wldim{(G,\coloring_G)}$ of a graph~$(G,\coloring_G)$ is the minimal~$k \in \Nat$ such that~\wld{k} identifies~$G$.
Note that for a colored graph~$(G,\coloring_G)$, the~WL-dimension only depends on~$\pi(\coloring_G)$ and not on~$\coloring_G$ itself.

From the coloring stable under \wltwo, one can read out many invariants like distances between the vertices in~$G$.
Therefore, \wltwo is able to identify cycles and cliques.

\paragraph{Coherent configurations.}
Let~$\vertices$ be a finite set of vertices and~$\relations$ a set of binary relations partitioning~$\vertices^2$.
Then the pair~$\coherentConfig = (\vertices, \relations)$ is called~\emph{coherent configuration} if all of the following properties are met:
\begin{enumerate}[label = (CC\arabic*), leftmargin = 4em]
    \item \label{coherent-config:fibers}
    All relations of~$\relations$ contain either only self-loops $vv$ or only edges between different vertices~$vu$ with~$v\neq u$.
    \item \label{coherent-config:symmetric}
    For each relation~$A\in\relations$, the set $\relations$ contains the \emph{transposed} relation~$A^\star= \{vu \mid uv \in A \}$.
    \item \label{coherent-config:wl2}
    For all triples~$A, B, T \in \relations$, and vertices~$v,w$ with~$vw \in T$ the number of vertices~$u$ with~$vu \in A$ and~$uw \in B$ depends only on~$A$,~$B$, and~$T$, but not on the choice of~$v$ and $w$.
\end{enumerate}

We denote the set of all vertices of~$\coherentConfig$ by~$\vertices(\coherentConfig)$, and the set of all relations of~$\coherentConfig$ by~$\relations(\coherentConfig)$.
We call each~$\arcs \in \relations$ a \emph{basis relation}, $\abs{\vertices}$ the~\emph{order of~$\coherentConfig$}, and~$\abs{\relations}$ the~\emph{rank of~$\coherentConfig$}.

In literature, Condition~\ref{coherent-config:fibers} is often formulated as~$1_\vertices \in \relations^\cup$, meaning there is a union of basis relations containing all self-loops.
Condition~\ref{coherent-config:wl2} is often referred to as~\emph{coherence}.
Further, the number of vertices~$u$ with~$vu \in A$ and~$uw \in B$ in~\ref{coherent-config:wl2} is called an~\emph{intersection number} and is denoted by~$c^{AB}_T$.

\paragraph{Fibers and interspaces.}
Let~$\coherentConfig$ be a coherent configuration.
A basis relation of self-loops is called a~\emph{fiber}.
Recall that we interpret a self-loop~$vv$ often as the vertex~$v$.
Abusing notation, we call a vertex set~$\Delta \subseteq \vertices$ a fiber if~$1_\Delta \in \relations$.
We denote the set of all fibers of~$\coherentConfig$ by~$\fibers{\coherentConfig}$.
By Property~\ref{coherent-config:fibers} the vertex set~$\vertices(\coherentConfig)$ is partitioned by the collection of fibers~$\fibers{\coherentConfig}$.
A fiber~$\Delta$ is called a~\emph{singleton} if~$\abs{\Delta} = 1$, and if a coherent configuration~$\coherentConfig$ has only one fiber, we call~$\coherentConfig$~\emph{homogeneous}.

Given two unions of fibers~$\mathcal{R},\mathcal{B} \in \fibers{\coherentConfig}^\cup$, there is a unique subset~$\relations'$ of~$\relations(\coherentConfig)$ which partitions~$\mathcal{R} \times \mathcal{B}$.
We call this partition the~\emph{interspace} between $\mathcal{R}$ and~$\mathcal{B}$ and denote it by~$\interspace{\mathcal{R}}{\mathcal{B}}$.
Observe that~$\arcs \in \interspace{\mathcal{R}}{\mathcal{B}}$ if and only if~$\arcs^\star \in \interspace{\mathcal{B}}{\mathcal{R}}$.
If~$\abs{\interspace{\mathcal{R}}{\mathcal{B}}} = 1$ we call~$\interspace{\mathcal{R}}{\mathcal{B}}$ and~$\interspace{\mathcal{B}}{\mathcal{R}}$~\emph{homogeneous}.
If~$\mathcal{R} = \mathcal{B}$, we shorten~$\interspace{\mathcal{R}}{\mathcal{R}}$ to~$\inducedCC{\mathcal{R}}$.
For a union of fibers~$\mathcal{R}$, we define~$\coherentConfig - \mathcal{R} = (\vertices(\coherentConfig) - \mathcal{R}, \coherentConfig[\vertices(\coherentConfig) - \mathcal{R}])$.

Given two fibers~$R,B \in \fibers{\coherentConfig}$ and a basis relation~$\arcs \in \interspace{R}{B}$, we set $\intDegree{\arcs} \coloneqq \degree{}{R}{A}$ and
$\minimalDegree{R}{B} \coloneqq \min_{\arcs \in \interspace{R}{B}} \intDegree{A}$.

\paragraph{Constituents.}
Let~$\coherentConfig$ be a coherent configuration.
The graph~$(\vertices(\coherentConfig),\arcs)$ induced by a basis relation~$\arcs$ of~$\relations$ is called a~\emph{constituent} of~$\coherentConfig$.

For basis relation~$\arcs$ of a coherent configuration~$\coherentConfig$, we call the subgraph induced by~$\arcs$ a~\emph{constituent} of~$\coherentConfig$.
To declare that the constituent~$G$ is \emph{contained} in the interspace (or cell) between fibers~$R$ and~$B$, we will abuse notation and write~$G \in \interspace{R}{B}$ instead of~$\arcs(G) \in \interspace{R}{B}$. If clear from context, abusing notation further, by~$G\in \interspace{R}{B}$ for explicitly given graphs~$G$, we  mean that~$\interspace{R}{B}$ contains a constituent that is isomorphic to~$G$.

\paragraph{Coherent closure and individualizations.}
Given two coherent configurations~$\coherentConfig$ and~$\coherentConfig'$, we say that $\coherentConfig$ is at least as \emph{fine} as~$\coherentConfig'$, denoted by~$\coherentConfig \finer \coherentConfig'$ if for all~$\arcs \in \relations(\coherentConfig')$ we have~$\arcs \in \relations(\coherentConfig)^\cup$.
Conversely, we say that $\coherentConfig'$ is at least as \emph{coarse} as~$\coherentConfig$.

Let~$\relations$ be a collection of relations on~$\vertices$ which does not necessarily satisfy Property~\ref{coherent-config:wl2}.
The~\emph{coherent closure} of~$\relations$ is the coarsest coherent configuration whose relations are each contained in some relation in~$\relations$.
For vertices~$v_1,\dots, v_\ell \in \vertices$, we set~$\coherentConfig_{v_1,\dots,v_\ell} \coloneqq \wltwo(\relations(\coherentConfig) \cup \{1_{v_1},\dots,1_{v_\ell}\})$.
Intuitively we individualize~$v_1,\dots,v_\ell$, that is, force each~$v_i$ to form its own color class (in terms of colored graphs) or equivalently to be its own fiber (in terms of coherent configuration).

\begin{theorem}
\label{preliminaries:wldim-individualizations/thm}
    Let~$\coherentConfig$ be a coherent configuration, and let~$v_1,\dots,v_\ell \in \vertices(\coherentConfig)$.
    Then~$$\wldim{\coherentConfig} \leq \ell + \max\{2,\wldim{\coherentConfig_{v_1,\dots,v_\ell}}\}.$$
\end{theorem}

\paragraph{Coherent configurations and WL.}
The~$2$-dimensional Weisfeiler-Leman algorithm (\wltwo) computes the coarsest coherent configuration refining a given configuration.
By interpreting edges, non-edges, and self-loops as separate relations, and applying \wltwo, every graph can be transformed into a coherent configuration.

Conversely, we interpret a coherent configuration~$\coherentConfig$ as a colored complete  directed graph~$(G_\coherentConfig,\coloring_\coherentConfig)$ where
\[
    \vertices(G_\coherentConfig) = \vertices(\coherentConfig), \quad \arcs(G_\coherentConfig) = \vertices(\coherentConfig)^2,  \quad   \text{ and } \quad  \pi_{G_\coherentConfig}(\coloring) = \relations.
\]
We call this graph an~\emph{associated graph} of~$\coherentConfig$. Since we only specify the partition of the color classes and not the colors themselves, there are actually multiple associated graphs, but this will not be relevant for us.
The color partition~$\pi(\coloring_\coherentConfig)$ is stable under~$\wld{2}$ (see~\cite{DBLP:conf/stoc/Babai16}).

Overall, coherent configurations are stable under \wltwo, and in fact ``stability under \wltwo'' and ``coherence'' essentially refer to the same concept.

\paragraph{Quotient graph.}
The overarching connectivity between fibers of~$\coherentConfig$ is captured by the~\emph{quotient graph}~$\quotientGraph{\coherentConfig}$ of~$\coherentConfig$ which uses the set of fibers as vertex set and contains an edge between two fibers if they form a non-homogeneous interspace, that is, an interspace containing more than one relation.
More formally, it is the uncolored, undirected graph~$(\fibers{\coherentConfig}, \{ RB \mid \abs{\interspace{R}{B}} > 1 \})$.

We say that a fiber~$F$ is~\emph{adjacent} to a fiber~$F'$ (or~\emph{neighboring} a fiber~$F'$) if the interspace between~$F$ and $F'$ is not homogeneous ($FF' \in \arcs(\quotientGraph{\coherentConfig})$).
Further, the interspace~$\interspace{F}{F'}$ is~\emph{incident} to~$F$.
The \emph{quotient degree}~$\colorDeg{F}$ of a fiber~$F$ is its degree in~$\quotientGraph{\coherentConfig}$, that is, the number of non-homogeneous interspaces incident with~$F$.

We call a set~$\mathcal{S}$ of fibers~\emph{dominating} if every fiber of~$\coherentConfig$ is in~$\mathcal{S}$ or adjacent to at least one fiber of~$\mathcal{S}$ in~$\quotientGraph{\coherentConfig}$.

\paragraph{Special graphs.}
We denote the disjoint union of~$s$ copies of a graph~$G$ by~$sG$ and write~$\clique{n}$ for the complete, undirected graph of order~$n$.
We refer to the cycle of order~$n$ by~$\cycle{n}$, and use~$\overrightarrow{\cycle{n}}$ for the directed version.
We should remark that we use~$\disjointCliques{2}{3}$ and~$\disjointCycles{2}{3}$ both to refer to the same graph.
We refer to the rook graph on~$n\times n$ vertices by~$\rookGraph{n}$ and denote the Paley tournament on~$7$ vertices by~$PTr(7)$.

Given two vertex sets~$R$ and~$B$ of size~$n_1$ and~$n_2$ respectively, the graph~$\clique{n_1,n_2}$ is the complete, directed, bipartite graph from~$R$ to~$B$.
We call~$\disjointCliques{n_1}{1,\nicefrac{n_2}{n_1}}$ a~\emph{star}, and~$\matching{n_1}$ a~\emph{matching}.
Given disjoint fibers~$\{r_1,\dots,r_n\}$ and~$\{b_1,\dots,b_n\}$, the set of arcs~$\{r_ib_i, r_ib_{i+1} \mid i \in \{1,\dots, n-1\}\} \cup \{r_nb_n,r_nb_1\}$ introduce a direction-alternating cycle of length~$2n$ in the interspace between the fibers.
Abusing notation we denote it by~$\cycle{2n}$.

\begin{figure}[tbp]
    \centering
    \begin{minipage}{.32\textwidth}
        \centering

\begin{tikzpicture}[scale=0.755]
    \useasboundingbox (-2.25,-2.15) rectangle ++(4.5,4.7);

    \node[vertex,lightgray,fill=lightgray] (p0) at (90:1) {};
    \node[vertex,lightgray,fill=lightgray] (p1) at (90:2.1) {};
    \node[vertex,lightgray,fill=lightgray] (p2) at (210:1) {};
    \node[vertex,lightgray,fill=lightgray] (p3) at (210:2.1) {};
    \node[vertex,lightgray,fill=lightgray] (p4) at (330:1) {};
    \node[vertex,lightgray,fill=lightgray] (p5) at (330:2.1) {};
    \draw[arrow,thick,black]
        (p0) edge (p2)
        (p2) edge (p4)
        (p4) edge (p0);
    \draw[arrow,thick,black, bend right=40]
        (p1) edge (p3)
        (p3) edge (p5)
        (p5) edge (p1);
    \draw[arrow,thick,black, bend right]
        (p0) edge (p3)
        (p3) edge (p4)
        (p4) edge (p1)
        (p1) edge (p2)
        (p2) edge (p5)
        (p5) edge (p0);

\end{tikzpicture}
        \caption{The graph~${\protect\overrightarrow{C_3}[K_2]}$.}
        \label{small-cc:doubleTriangle/fig}
    \end{minipage}
    \hfil
    \begin{minipage}{.32\textwidth}
        \centering

\begin{tikzpicture}

    \useasboundingbox (-1.5,-1.5) rectangle ++(3,3.6);

    \foreach \x in {0,...,6}
    {
        \node[vertex,lightred, fill=lightred ] (v1\x) at (-1.2,-1.2 +\x*0.5) {};
        \node[vertex,lightblue,fill=lightblue] (v2\x) at ( 1.2,-1.2 +\x*0.5) {};

    }
    \draw[edge,thick,->,black]
        (v10) edge  (v20)
        (v11) edge  (v20)
        (v11) edge  (v21)
        (v12) edge  (v21)
        (v12) edge  (v22)
        (v13) edge  (v22)
        (v13) edge  (v23)
        (v14) edge  (v23)
        (v14) edge  (v24)
        (v15) edge  (v24)
        (v15) edge (v25)
        (v16) edge (v25)
        (v16) edge (v26)
        (v10) edge (v26);
    \draw[edge, thick, ->, black]
        (v10) edge  (v22)
        (v11) edge  (v23)
        (v12) edge  (v24)
        (v13) edge  (v25)
        (v14) edge  (v26)
        (v15) edge (v20)
        (v16) edge (v21);
\end{tikzpicture}
        \caption[Incidence graph of the Fano plane.]{The graph~$\leviFano$.}%
        \label{small-cc:leviFano/fig}
    \end{minipage}
    \hfil
    \begin{minipage}{.32\textwidth}
        \centering

\begin{tikzpicture}
    \useasboundingbox (-.5,-.75) rectangle ++(2.5,3.6);
    \foreach \x in {0,...,3}
    {
        \node[vertex,lightred,fill=lightred] (r\x) at (0,.25+2*\x/4) {};
    }
    \foreach \x in {0,...,5}
    {
        \node[vertex,lightblue,fill=lightblue] (b\x) at (1.5,-.25+2*\x/4) {};
    }
    \draw[edge,thick,->]
        (r0) edge (b0)
        (r1) edge (b0)
        (r1) edge (b3)
        (r2) edge (b3)
        (r2) edge (b5)
        (r3) edge (b5)
        (r3) edge (b2)
        (r0) edge (b2)
        (r0) edge (b1)
        (r2) edge (b1)
        (r1) edge (b4)
        (r3) edge (b4);
\end{tikzpicture}
        \caption[Incidence graph of the clique of size four.]{The graph~$\interspaceFourSix$.}
        \label{small-cc:graph-between-4cc-6cc/fig}
    \end{minipage}
\end{figure}

The graph~$\overrightarrow{C_3}[K_2]$ is a~$\overrightarrow{C_3}$ in which each vertex is replaced by two isomorphic copies of itself:
more precisely, it is the directed graph~$(\vertices,\arcs)$ where~$\vertices = \{v_i, w_i \mid i \in \{0,1,2\}\}$ and
\[
    \arcs = \{v_iv_{j},v_iw_{j},w_iv_{j},w_iw_{j} \mid (i,j) \in \{(0,1),(1,2),(2,0)\}\}
\]
by~$\overrightarrow{C_3}[K_2]$.
It is visualized in Figure~\ref{small-cc:doubleTriangle/fig}.

The graph~$\leviFano$ is a directed bipartite graph~$(R \cup B, \arcs)$ with~$\abs{R} = \abs{B}  = 7$ such that for all distinct~$r,r'\in R$ there is exactly one~$b \in B$ such that~$rb,r'b \in \arcs$.
The underlying undirected graph of~$\leviFano$ is isomorphic to the incidence graph (or Levi graph) of the Fano plane, which is also known as Heawood graph.
The graph~$\leviFano$ is shown in Figure~\ref{small-cc:leviFano/fig}.

The graph~$\interspaceFourSix$ is a directed version of the incidence graph of~$\clique{4}$ and is depicted in Figure~\ref{small-cc:graph-between-4cc-6cc/fig}:
more precisely, it is a directed bipartite graph~$(R \cup B, \arcs)$ with~$\abs{R} = 4$, $\abs{B} = 6$ such that for all distinct~$r,r' \in R$ there is exactly one~$b \in B$ such that~$rb,r'b \in \arcs$..


\section{Outline of the Proof for the Upper Bound}
\label{upper-bound/sec}

In this section, we outline all steps used to prove the upper bound. The details for these steps are then provided in Sections~\ref{critical-graph/sec} to~\ref{sec:proof:of:main:thm}.

\paragraph{General Strategy.}
Instead of proving our upper bound only for graphs, we generalize the claim to coherent configurations~$\coherentConfig$. The reader unfamiliar with them may think of them as edge and vertex colored graphs with strong regularity conditions for the color classes and between them. We then show the following:
    \begin{center}
    If~$\coherentConfig$ has~$n$ vertices, then~$\wldim{\coherentConfig}\leq \frac{3}{20} \cdot n + o(n)$. (See Theorem~\ref{main-theorem/thm})\end{center}

The proof of the theorem aims to iteratively reduce the number of vertices or simplify the coherent configuration as follows. For a coherent configuration~$\coherentConfig$, we want to bound $\wldim{\coherentConfig}$ in terms of~$\wldim{\coherentConfig'}$ where~$\coherentConfig'$ has fewer vertices than~$\coherentConfig$ or is structurally simpler. To simplify~$\coherentConfig$, we use several types of reductions.

In the first, we exploit knowledge of the WL-algorithm to show that certain color classes can be removed without changing the dimension at all.
In the second, called \emph{local reductions}, we individualize vertices. For example,~$\coherentConfig'$ could be obtained from~$\coherentConfig$ by an individualization followed by the application of~$\wltwo$. This means that~$\wldim{\coherentConfig}\leq \wldim{\coherentConfig'}+1$. 
Although each individualization increases our upper bound by 1, we ensure that sufficient progress is made during the reduction process. We do so by carefully choosing the vertices we individualize. (Formally a potential function argument is used to measure the trade-off.)
Eventually, no further local reduction can be applied.

At that point, we split the resulting coherent configuration using a global reduction. This splits the configuration into multiple independent components.
These components are the base cases of our reduction process, and we deal with each component separately. More formally, the overall dimension can be bounded in terms of the maximum of the dimensions of the subconfigurations induced by the components.

In all of the reductions, we heavily rely on structural analysis of fibers, interspaces, and their overall connectivity.
A criterion often influencing the selection of the reduction is the size of the fiber:
in particular, fibers~$F$ are differentiated into~\emph{tiny} fibers~($\abs{F} < 4$), \emph{small} fibers~($4 \leq \abs{F} \leq 7$), and~\emph{large} fibers~($7 < \abs{F}$).

\paragraph{Criticality and restorability (Section~\ref{critical-graph/sec}).}
For the first set of reductions, we employ the concept of a \emph{critical} configuration (see Definition~\ref{def:critical-graph}):
intuitively, a coherent configuration is critical if no fiber can be removed without decreasing the WL-dimension.
From this concept, we derive multiple structural properties, e.g. the coherent configuration does not contain tiny fibers (see Lemma~\ref{critical:tiny-CC/lem}).

A new technique we develop to deal with fibers that are not tiny is the \emph{restorability} of fibers:
intuitively, when a restorable fiber is removed, there is up to isomorphism a unique way of reinserting the fiber.
Within critical coherent configurations, such restorable fibers are severely limited (see Lemma~\ref{critical:restorable/lem}).

\paragraph{Interspaces between small fibers (Section~\ref{small-cc/sec}).}
Based on exhaustive generation of small homogeneous coherent configurations (see~\cite{MiyamotoHanaki2000} and~\cite{DBLP:journals/dm/HanakiM03} for example), we classify the interspaces between two small fibers (see Lemma~\ref{small-cc:interspace/lem}).
This classification implies the following structural consequence:
in a critical coherent configuration all of whose fibers are small, either all fibers have the same size or each fiber has size $4$ or~$6$ (see Theorem~\ref{critical:small-cc/thm}).

\paragraph{Interspace patterns (Section~\ref{interspace-large-small/sec}).}
Our next goal is to describe the possible interspaces between a small fiber and a fiber of arbitrary size.

Since the large fiber is not bounded in size, there is of course an infinite number of isomorphism types that appear.
To solve this issue and classify the possible interspaces between a small fiber and a fiber of arbitrary size, we introduce a new concept (see Definition~\ref{large-small-interspace:interspace-pattern/def}):
An \emph{interspace pattern} is a string of information which suffices to characterize the interspace up to a certain equivalence and allows us to extract all combinatorial information relevant for our purposes.

Using this concept, we classify the interspaces occurring between a fiber of size~$4$ or~$6$ and an arbitrary large fiber into finitely many classes (see Corollary~\ref{large-small-interspace:classification:uniqueness/cor}).
We remark that only considering cases with a small fiber of size~$4$ or~$6$ suffices for our overall proof.

\paragraph{Exploiting restorability (Section~\ref{critical:restorable/sec}).}
With the interspace classifications at hand, we now examine induced coherent subconfigurations containing a small fiber~$S$ with regard to criticality.
It turns out that certain combinations of interspace patterns occurring simultaneously at one fiber cause~$S$ to be restorable.
This dictates that certain other combinations of interspace patterns must appear (for example, see Lemma~\ref{critical:4cc:restorable:cycle/lem}).
It also gives us further structural consequences, e.g., lower bounds on the number of neighboring fibers~$S$ must have (for example, see Lemma~\ref{critical:4-cc:restorable:DUC/lem}).

\medskip
This concludes the first set of reductions.
We rely on their structural consequences heavily in Sections~\ref{wldim-small/sec} to~\ref{global-argument/sec} -- even though this is not apparent from this outline.
Before we move on to our second set of reductions, we consider the two base cases.

\paragraph{First base case: only small fibers (Section~\ref{wldim-small/sec}).}

With the preparations complete, we examine a critical coherent configuration~$\coherentConfig$ in which all fibers are small.
While there are multiple cases to consider, in the most important of these, the coherent configuration~$\coherentConfig$ is essentially a CFI-graph with a cubic base graph.
For such graphs we use general bounds on the treewidth of cubic graphs \cite{pathwidthCubicGraphs} to show that~$\wldim{\coherentConfig} \leq \frac{n}{24} + o(n)$.
Overall, we show an upper bound for a coherent configuration containing only small fibers.

\begin{center}
    If~$\coherentConfig$ has~$n$ and all its fibers are small,\linebreak then~$\wldim{\coherentConfig} \leq \frac{n}{20} + o(n)$. (See Theorem~\ref{small-cc:wldim/thm})
\end{center}

\paragraph{Second base case: limited fiber size (Section~\ref{sec:limit:fiber:sizes}).}

In our overall argument, we want to use a bound on the size of the fibers that appear in the coherent configurations. For this purpose
we adapt Zemlyachenko's degree reduction technique (see~\cite{DBLP:conf/fct/Babai81}) to coherent configurations. For the cost of a sublinear number of individualizations we obtain a bound~$d$ on the degree of all but the largest relation in each interspace and fiber.
Using an adaptation of the argument,  we exploit the fact that the degree can be bounded with another sublinear number of individualizations to obtain a bound on the fiber size of the resulting coherent configuration of interest (see Theorem~\ref{lem:bound-on-cc-size}):

If we assume that, in addition to an upper bound~$t \in \Nat$ on the fiber size, the treewidth of~$\quotientGraph{\coherentConfig}$ is bounded as well, the WL-dimension of~$\coherentConfig$ can be bounded in terms of the treewidth of its quotient graph and its fiber size (see Lemma~\ref{lem:bd:tw:and:fibre:size:bd:WL}).

\paragraph{Local reductions (Sections~\ref{sec:potential:func} and~\ref{recursive-argement/sec}).}

Now let us consider the local reductions, which are our second set of reductions:
the overall goal for local reductions is that, after their exhaustive application, we obtain a critical coherent configuration structurally similar to the base cases previously discussed.
In particular, this means avoiding certain local subconfigurations~$\mathfrak{S}$ of~$\coherentConfig$, e.g., interspaces between a large and a small fiber having certain interspace patterns.
We individualize~$t$ vertices and re-establish the coherence and criticality.
Generally, this causes the large fibers to split into multiple fibers smaller in size.
The result is a critical coherent configuration~$\coherentConfig'$ which does not contain the unwanted local configuration~$\mathfrak{S}$.
While we are one step closer to our goal, there are two challenges we need to address:
\begin{itemize}
    \item
    Since the use of individualization increases our final bound on the WL-dimension, the configuration~$\coherentConfig'$ must be obtained from~$\coherentConfig$ in a controlled manner.
    We must carefully choose the vertices to be individualized and charge the costs of individualizations to the progress we made towards our goal.
    To this end, we define the potential function~$\tau$ with~$\tau(\coherentConfig) \coloneqq \frac{3n_\ell + n_s - 8k_\ell}{20}$ where parameters~$n_\ell$, $k_\ell$, and~$n_s$ of~$\coherentConfig$ are the number of vertices in large fibers, the number of large fibers, and the number of vertices in small fibers respectively.
    The progress of a local reduction can be quantified by~$\tau(\coherentConfig') - \tau(\coherentConfig)$, and only local reductions where~$\tau(\coherentConfig) \geq t + \tau(\coherentConfig')$ for~$t$ individualizations are eligible.
    \item
    While reestablishing the coherence causes fibers contained in~$\mathfrak{S}$ to split, other fibers outside of~$\mathfrak{S}$ might also be affected.
    So the resulting coherent configuration~$\coherentConfig'$ depends on~$\coherentConfig - \mathfrak{S}$.
    To dodge this pitfall, we introduce a ``monotonization''~$\widetilde{f}$ of the WL-dimension:
    it maximizes the WL-dimension over all critical coherent configurations whose potential is smaller or equal to a given one.
\end{itemize}
While the machinery is more technical than presented here, each local reduction roughly adheres to the following blueprint:
if~$\coherentConfig$ contains~$\mathfrak{S}$, then~$\wldim{\coherentConfig} \leq t + \widetilde{f}(\tau(\coherentConfig) - t')$ where~$t' \geq t$.

Even though local reductions imply many structural consequences, we want to highlight a particular one:
if a certain reduction (Lemma~\ref{lem:local-argument:3-large-neighbors}) is exhaustively applied, the treewidth of the quotient graph involving large fibers is bounded.
Down the line, this allows us to apply the second base case to these parts.

In general, the use of the reductions allows us to reduce our problem to coherent configurations that satisfy a set of very restrictive structural properties.
This allows us to employ a global argument.

\paragraph{Global argument (Sections~\ref{structure-reduced-cc/sec} and~\ref{global-argument/sec}).}

Now assume that the fiber size of the critical coherent configuration~$\coherentConfig$ is bounded by a constant~$t$ and no local reduction is applicable to~$\coherentConfig$.
We call~$\coherentConfig$ a \emph{$t$-reduced} coherent configuration.
Considering the global structure of~$\coherentConfig$, it can be partitioned into large components, small components, and a boundary:
\begin{itemize}
    \item
    \emph{Large components} contain at least one large fiber.
    If these components are separated from the rest of~$\coherentConfig$, their WL-dimension can be bounded by the second base case.
    \item
    \emph{Small components} contain only small fibers that have size~$4$ or~$6$.
    If these components are separated from the rest of~$\coherentConfig$, an upper bound on their WL-dimension is given by the first base case.
    \item
    The \emph{boundary} contains the inner fibers of the induced paths that connect small and large components.
    All fibers of the boundary are small fibers. The end fibers of the paths are small fibers of small or large components.
\end{itemize}

We now aim to decompose~$\coherentConfig$ into the small and large components.
Due to various structural properties, only a few carefully chosen vertices within the paths are needed to take care of the paths between the large and small components.
Further, since the size of the boundary can be bounded by the number of large fibers, we can charge the individualizations to the large components.
Thus the overall number of individualizations required in the boundary is not too large.

\paragraph{Combining the arguments (Section~\ref{sec:proof:of:main:thm}).}
We finally combine the local reductions, the global argument, and both base cases and tally up the number of vertices that we have individualized up to now.
Overall, we conclude the proof of our main theorem by bounding the WL-dimension by~$\frac{3}{20}n + o(n)$.


\section{Critical Coherent Configurations}
\label{critical-graph/sec}

We start our investigation  of coherent configurations by introducing the two core concepts of this paper and by presenting the first properties.

\begin{definition}
\label{def:wldim}
    Let~$\coherentConfig$ be a coherent configuration, and let~$(G_\coherentConfig,\coloring_\coherentConfig)$ be an associated graph.
    We set~$\wldim{\coherentConfig} \coloneqq \wldim{(G_\coherentConfig,\coloring_\coherentConfig)}$.
\end{definition}

Note that, the WL-dimension of a coherent configuration only depends on the color partition induced by~$\coloring_\coherentConfig$ and not on the actual colors.\footnote{The definition essentially treats the colors as distinguishable. It is possible to  alternatively define a related but different notion of WL-dimension related to separability (see~\cite{CC} for more information). However, our Definition~\ref{def:wldim} is the one that is suitable for the application to graphs.}

\begin{definition}
\label{def:critical-graph}
    We call a coherent configuration~$\coherentConfig$ \emph{critical} if~$\wldim{\coherentConfig}\geq 2$ and  there is no set of fibers~$\mathcal{R} \subsetneq \fibers{\coherentConfig}$ such that~$\wldim{\coherentConfig - \bigcup_{R \in \mathcal{R}} R} = \wldim{\coherentConfig}$.
\end{definition}

\begin{lemma}
\label{crictial:quotientGraph-connected/lem} \label{critical:disjoint-union/lem}
    Let~$\coherentConfig$ be a coherent configuration.
    If~$\coherentConfig$ is critical, then~$\quotientGraph{\coherentConfig}$ is connected.
\end{lemma}
\begin{proof}[Proof sketch]
    Since~$$\wldim{\coherentConfig}= \max \{\wldim{\coherentConfig[C]}\mid  \text{$C$ a connected component of~$\quotientGraph{\coherentConfig}$}\},$$ the quotient graph~$\quotientGraph{\coherentConfig}$ must consist of exactly one connected component.
\end{proof}

\begin{lemma}
\label{critical:star/lem}
    Let~$\coherentConfig$ be a coherent configuration.
    If~$\coherentConfig$ is critical, then there are no distinct fibers~$R,B \in \fibers{\coherentConfig}$ with~$\minimalDegree{B}{R}=1$, that is, fibers with~$\disjointCliques{\abs{R}}{1,\frac{\abs{B}}{\abs{R}}} \in \interspace{R}{B}$.
\end{lemma}
\begin{proof}[Proof sketch]
    We argue that~$\wldim{\coherentConfig-R} \geq \wldim{\coherentConfig}$.
    Recall that we can interpret coherent configuration as complete colored graphs.
    Suppose that~$\coherentConfig' \simeq_k \coherentConfig$.
    We can assume that the two configurations are defined on the same vertex set and their vertex colorings agree.
    Then~$\coherentConfig'[\vertices(\coherentConfig)-B] \simeq_k  \coherentConfig[\vertices(\coherentConfig)-B]$.

    It suffices now to observe that an isomorphism~$\varphi$ from~$\coherentConfig'[\vertices(\coherentConfig)-B]$ to~$\coherentConfig[\vertices(\coherentConfig)-B]$ extends to an isomorphism~$\widehat{\varphi}$ from~$\coherentConfig'$ to~$\coherentConfig$. The extension is defined as follows.
    Let~$U$ be a basis relation in~$\interspace{R}{B}$ such that~$(R \cup B, U)$ is isomorphic to~$\disjointCliques{\abs{R}}{1,\nicefrac{|B|}{|R|}}$.
    For vertex~$r\in R$ choose a neighbor~$b$ in~$B$ with respect to~$U$.
    Set~$\widehat{\varphi}(r)$ to be the unique neighbor with respect to~$U$ of~$\varphi(b)$. Since criticality requires~$\wldim{\coherentConfig}\geq 2$, it follows from coherence that~$\widehat{\varphi}(r)$ is an isomorphism.
\end{proof}

\begin{lemma}
\label{critical:cycle/lem}
    Let~$\coherentConfig$ be a coherent configuration.
    If~$\coherentConfig$ is critical, then there are no fibers~$R,B\in \fibers{\coherentConfig}$ with~$|R|=|B|$ odd such that~$\minimalDegree{R}{B}=2$.
    In particular, if~$\coherentConfig$ is critical, then there is no constituent~$\disjointCycles{t}{2x} \in \interspace{R}{B}$ for positive integers~$t,x$ with~$x$ odd.
\end{lemma}
\begin{proof}
    Recall that~$\wltwo$ is able to measure lengths of paths of particular colors between two vertices, and in particular, color-alternating paths between two vertices. Under the assumptions, for all~$r \in R$ there is a unique~$b \in B$ such that there are two color-alternating paths between~$R$ and~$B$ of equal length.
    By Property~\ref{coherent-config:wl2}, the arcs~$rb$ form a basis relation of~$\coherentConfig$, and thus~$\disjointCliques{\abs{R}}{1,1} \in \interspace{R}{B}$.
    This contradicts Lemma~\ref{critical:star/lem}.
\end{proof}

A fiber~$R$ of a coherent configuration is called~\emph{large} if~$8 \leq \abs{R}$, \emph{small} if~$4 \leq \abs{R} \leq 7$, and~\emph{tiny} if~$\abs{R} \leq 3$.
Thinking of red, blue and yellow, we will typically use the letters~$R$,~$B$ and~$Y$ for fibers in general. We will use~$L$ or~$S$ to indicate that the fiber in question is large or small, respectively.

\begin{lemma}
\label{critical:tiny-CC/lem}
    Let~$\coherentConfig$ be a coherent configuration.
    If~$\coherentConfig$ is critical, then there is no tiny fiber in~$\fibers{\coherentConfig}$.
\end{lemma}
\begin{proof}
    Let~$R \in \fibers{\coherentConfig}$ be a tiny fiber. It cannot be that $R$ is the only fiber of~$\coherentConfig$ since otherwise we would have~$ \wldim{\coherentConfig}=1$.
    If~$\abs{\interspace{R}{B}} > 1$, then~$\disjointCliques{\abs{R}}{1,\nicefrac{\abs{B}}{\abs{R}}} \in \interspace{R}{B}$ contradicting Lemma~\ref{critical:star/lem}.
    So all interspaces incident to~$R$ in~$\quotientGraph{\coherentConfig}$ are homogeneous which contradicts that the quotient graph is connected (Lemma~\ref{crictial:quotientGraph-connected/lem}).
\end{proof}

Let~$\mathcal{R}$ be a union of fibers of~$\coherentConfig$, and let~$\mathcal{B}$ be the union of fibers that are not in~$\mathcal{R}$ but adjacent (in~$\quotientGraph{\coherentConfig}$) to some fiber in~$\mathcal{R}$.
We call~$\mathcal{R}$ \emph{restorable} if every automorphism of~$\coherentConfig[\mathcal{B}]$ that extends to an automorphism of
$\coherentConfig-\mathcal{R}$ also extends to an automorphism of~$\coherentConfig[\mathcal{R} \cup \mathcal{B}]$.

\begin{lemma}
\label{critical:restorable/lem}
    If a coherent configuration~$\coherentConfig$ is critical, then there is no non-dominating set of fibers whose union is restorable.
\end{lemma}
\begin{proof}
    Assume that there is a non-dominating set of fibers whose union~$\mathcal{R}$ is restorable in a critical coherent configuration~$\coherentConfig$.
    Let~$\mathcal{B}$ be the union of fibers that are not in~$\mathcal{R}$ but adjacent (in~$\quotientGraph{\coherentConfig}$) to some fiber in~$\mathcal{R}$. Set~$k\coloneqq\max \{ \wldim{\coherentConfig[\mathcal{R}\cup \mathcal{B}]},\wldim{\coherentConfig-\mathcal{R}} \}$. We argue that~$k= \wldim{\coherentConfig}$, which proves the statement.
    Recall that we can interpret coherent configuration as complete colored graphs.
    So let~$\coherentConfig'$ be a second coherent configuration for which~$(\coherentConfig,\chi)\simeq_k (\coherentConfig',\chi')$ with suitable colorings~$\chi$ and~$\chi'$. Set~$\mathcal{R}'\coloneqq\coloring'^{-1}(\coloring(\mathcal{R}))$
    and~$\mathcal{B}'\coloneqq\coloring'^{-1}(\coloring(\mathcal{B}))$ to be the vertices in~$\coherentConfig'$ in classes corresponding to~$\mathcal{R}$ and~$\mathcal{B}$, respectively.
    By the choice of~$k$ there is an isomorphism~$\varphi$ from~$\coherentConfig-\mathcal{R}$ to~$\coherentConfig'-\mathcal{R}'$ and an isomorphism~$\psi$ from~$\coherentConfig[\mathcal{R}\cup \mathcal{B}]$ to~$\coherentConfig'[\mathcal{R}'\cup \mathcal{B}']$.
    The restriction~$\psi^{-1}(\varphi|_{\mathcal{B}})$ is an automorphism of~$\coherentConfig[\mathcal{B}]$. Since~$\mathcal{R}$ is restorable, this restriction extends to an automorphism~$\mu$ of~$ \coherentConfig[\mathcal{R}\cup \mathcal{B}]$.
    Now consider the following map~$\nu\colon \vertices(\coherentConfig)\rightarrow \vertices(\coherentConfig')$ defined by
    \[
        \nu (v) \coloneqq
        \begin{cases}
                \varphi(v),   & \text{for~$v\notin \mathcal{R}\cup \mathcal{B}$}\\
                \psi(\mu(v)), & \text{otherwise.}
        \end{cases}
    \]
    To see that this map is an isomorphism from~$\coherentConfig$ to~$\coherentConfig'$, it suffices to observe that for~$b\in \mathcal{B}$ we have~$\psi(\mu(b))= \varphi(v)$, so~$\nu$ combines two isomorphisms that agree on~$\mathcal{B}$. We conclude that~$\coherentConfig\cong \coherentConfig'$.
\end{proof}

Let~$R$ be a fiber of a coherent configuration. We say that a fiber~$Y$ distinct from~$R$ is \emph{taken care of (regarding restorability of~$R$)} if for every~$y\in Y$ and~$U\in \interspace{Y}{R}$ there is a fiber~$B$  different from~$R$ and~$Y$, a point~$b\in B$, and an interspace~$U_B\in \interspace{B}{R}$  such that~$bU_B\subseteq yU$.

\begin{lemma}
    \label{critical:restorable:take-care/lem}
    Let~$\coherentConfig$ be a coherent configuration with distinct fibers~$R$ and~$Y$.
    If~$Y$ is taken care of (regarding restorability of~$R$) and~$R$ is restorable in~$\coherentConfig-Y$, then~$R$ is restorable in~$\coherentConfig$.
\end{lemma}
\begin{proof}
    Let~$\mathcal{B}$ be the union of fibers that are neighbors of~$R$ in $\quotientGraph{\coherentConfig}$.
    Suppose~$\varphi$ is an automorphism of~$\coherentConfig[\mathcal{B}]$. Consider the restriction~$\varphi'\coloneqq \restr{\varphi}{\mathcal{B}\setminus{Y}}$. Since~$R$ is restorable in~$\coherentConfig-Y$ there is automorphism~$\widehat{\varphi}$ that is an extension of~$\varphi'$ to~$\coherentConfig[(\mathcal{B} \setminus Y)\cup R]$. Let~$\overline{\varphi}$ be the common extension of~$\widehat{\varphi}$ and~$\varphi'$.

    It suffices now to argue that~$\overline{\varphi}$ is an automorphism of~$\coherentConfig[\mathcal{B}\cup R]$. For this it suffices to consider non-homogeneous basis relations~$U\in \interspace{Y}{R}$.
    So assume~$yr \in U$.
    Since~$Y$ is taken care of, there is a fiber~$B\subseteq \mathcal{B}$ other than~$Y$, a vertex~$b\in B$, and~$U_B\in \interspace{B}{R}$ with~$bU_B\subseteq yU$. Due to coherence,~$b$ can be chosen so that~$r\in bU_B$.
    Since~$\widehat{\varphi}$ is an automorphism, we
    have that~$\widehat{\varphi}(r)\in \widehat{\varphi}(b) U_B$.
    Due to coherence we have also have
    that~$\varphi(b)U_B\subseteq \varphi(y)U$ because the color of an arc~$y'b'$ ``knows'' whether~$b'U_B\subseteq y'U$.
    We conclude that~$(\overline{\varphi}(y),\overline{\varphi}(r))\in U$.
\end{proof}

We call a vertex set~$R \subseteq \vertices(\coherentConfig)$ a~\emph{module} if for all~$b \in \vertices(\coherentConfig) \setminus R$ and all~$\arcs \in \relations(\coherentConfig)$ we have~$b\arcs \cap R \in \{\emptyset, R\}$.

\begin{lemma}
\label{critical:small-cc:module/lem}
    Let~$\coherentConfig$ be a coherent configuration that is not homogeneous.
    If~$\coherentConfig$ is critical, then there is no small fiber that can be properly partitioned into at most 3 modules.
\end{lemma}
\begin{proof}
    Let~$i \in \{2,3\}$.
    Towards a contradiction, assume that there is a small fiber~$S \in \fibers{\coherentConfig}$ which partitions into $i$~modules~$M_1,\dots,M_i$.
    By Property~\ref{coherent-config:wl2}, all these modules are of equal size~$\frac{\abs{S}}{i} \in \{2,3\}$.
    Furthermore, due to their small size, we have~$K_i \in \coherentConfig[M_j]$ for all~$j \in \{1,\dots,i\}$.
    Let~$\coherentConfig'$ be a copy of~$\coherentConfig$ in which we replace~$S$ by~$S' \coloneqq \{m_1,\dots,m_i\}$ where for all~$j \in \{1,\dots,i\}$ the vertex~$m_j$ is an arbitrary representative of~$M_j$. (So we contract the modules.)
    Since~$\wltwo$ identifies cliques and~$\interspace{M_j}{M_{j'}}$ is homogeneous for all distinct~$j,j' \in \{1,\dots,i\}$, $\wldim{\coherentConfig} = \wldim{\coherentConfig'}$.
    Combining Property~\ref{coherent-config:wl2} with the fact that~$M_j$ is a module for all~$j \in \{1,\dots,i\}$, we conclude for all~$R \in \fibers{\coherentConfig} \setminus \{S\}$ that either~$\coherentConfig'[R,S']$ is homogeneous or~$\disjointCliques{i}{\frac{\abs{R}}{i},1} \in \coherentConfig'[R,S']$.
    This contradicts the arguments used in the proof of Lemma~\ref{critical:tiny-CC/lem} which show that~$S'$ can be removed without decreasing the WL-dimension.
\end{proof}

\begin{theorem}
\label{interspace-divisor/thm}
    Let~$\coherentConfig$ be a coherent configuration.
    Suppose~$(R,B,Y)$ is an induced path in~$\quotientGraph{\coherentConfig}$, $r \in R$, $y \in Y$, $U \in \interspace{R}{B}$, and~$U' \in \interspace{Y}{B}$.
    Then~$\intDegree{U} \cdot \intDegree{U'} = \abs{B} \cdot \abs{rU \cap yU'}$.

    In particular~$\lcm\{\abs{B}, \intDegree{U}\}>1$ and~$\lcm\{\abs{B}, \intDegree{U'}\}>1$. In particular~$\abs{B}$ is not prime.
\end{theorem}
\begin{proof}
    Choose~$y \in Y$ and consider~$yU'$.
    Since~$\interspace{R}{Y}$ is homogeneous, the number~$\abs{rU \cap yU'}$ is independent of the choice of~$r \in R$ or~$y \in Y$.
    Note that all vertices in~$B$ have the same degree with respect to~$U^\star$.
    Thus the probability that~$rb \in U$ for~$r \in R$ and~$b \in B$ chosen independently uniformly at random is the same as the probability that~$rb \in U$ for~$r \in R$ and~$b \in yU'$ chosen independently uniformly at random.
    On top of that, both probabilities are the same if we first fix an arbitrary~$r$ and then choose~$b$ only  from~$B$ or~$yU'$, respectively.
    Thus~$\frac{\intDegree{U}}{\abs{B}} = \frac{\abs{vU \cap yU'}}{\intDegree{U'}}$.

    Observe that~$\intDegree{U'} > \abs{rU \cap yU'} \geq 1$ and that~$\intDegree{U} > \abs{rU \cap y U'} \geq 1$. This then implies the second part of the theorem.
\end{proof}


\section{Small fibers and their interspaces}
\label{small-cc/sec}

Observe that the constituents of a homogeneous coherent configuration with order at most~$7$ determine the isomorphism type of said homogeneous coherent configuration.
While this does not necessarily hold for homogeneous coherent configurations of larger orders, the observation translates to the coherent subconfiguration induced by a small fiber.
More formally, given an interspace~$\interspace{R}{B}$ between the (not necessarily distinct) fibers~$R$ and~$B$, let~$T$ be the collection of basis relations~$\arcs$ in a~$\interspace{R}{B}$ up to isomorphism while omitting~$1_R$,~$1_B$, and transposed basis relation~$\arcs^\star$,.
Define~$\type{\interspace{R}{B}}$ to be the collection of isomorphism types of the constituents induced by the basis relations in~$T$.


\begin{table}
    \centering\def\arraystretch{1.5}%
    \begin{tabular}{|c|l|}
        \hline
        $n$ & $\type{\coherentConfig}$ \\ \hline
        1 & $(K_1)$\\ \hline
        2 & $(K_2)$\\ \hline
        3 & $(K_3)$,~$(\overrightarrow{C_3})$\\ \hline
        4 & $(K_4)$,~$(C_4,2K_2)$,~$(2K_2,2K_2,2K_2)$,~$(\overrightarrow{C_4},2K_2)$\\ \hline
        5 & $(K_5)$,~$(C_5,C_5)$,~$(\overrightarrow{C_5},\overrightarrow{C_5})$, \\ \hline
        6 & $(K_6)$,~$(2 K_3,K_{3,3})$,~$(2\overrightarrow{C_3},K_{3,3})$,~$(3 K_2,K_{2,2,2})$,~$(3K_2,\overrightarrow{C_3}[K_2])$,\\
            & $(C_6,3K_2,2K_3)$,~$(\overrightarrow{C_6},2\overrightarrow{C_3},3K_2)$,~$(3K_2,3K_2,\overrightarrow{2C_3},3K_2)$\\ \hline
        7 & $(K_7)$,~$(C_7,C_7,C_7)$,~$(PTr(7))$,~$(\overrightarrow{C_7},\overrightarrow{C_7},\overrightarrow{C_7})$ \\ \hline
    \end{tabular}
    \caption{Classification of all homogeneous coherent configurations of order~$n \leq 7$.}
    \label{small-cc:classificaiton-small-cc/tab}
\end{table}

\begin{lemma}
\label{small-cc:induced-cc/lem} \label{small-cc:implied-cc/lem}
    Every tiny or small fiber~$R$ in a coherent configuration induces, up to isomorphism, one of~$23$  coherent configurations. These~$23$ options are given in Table~\ref{small-cc:classificaiton-small-cc/tab}.
\end{lemma}
\begin{proof}
    Observe that~$(R,\inducedCC{R})$ is a homogeneous coherent configuration.
    Small homogeneous coherent configurations have been classified, see for example~\cite{MiyamotoHanaki2000,DBLP:journals/dm/HanakiM03}.
\end{proof}

Note that for small fibers, the constituents of the induced homogeneous coherent configuration contained within a fiber determine the isomorphism type of said homogeneous coherent configuration.
(This is not necessarily the case for larger fibers.)

\begin{figure}[tbp]
    \centering
    \begin{subfigure}[]{.23\textwidth}
        \begin{center}

\begin{tikzpicture}[scale=0.85]

    \def\degree{60}

    \foreach \x in {0,...,5}
    {
        \node[vertex,lightgray,fill=lightgray] (p\x) at (\x*\degree+60:1) {};
    }
    \draw[arrow,blue]
        (p0) edge (p1)
        (p1) edge (p2)
        (p2) edge (p0);
    \draw[arrow,blue]
        (p3) edge (p4)
        (p4) edge (p5)
        (p5) edge (p3);
\end{tikzpicture}
        \end{center}
        \subcaption{$(2\overrightarrow{C_3},K_{3,3})$.}
        \label{small-cc:6cc:directed-2K3/fig}
    \end{subfigure}
    \hfil
    \begin{subfigure}[]{.23\textwidth}
        \begin{center}

\begin{tikzpicture}[scale=0.4]
    \node[vertex,lightgray,fill=lightgray] (p0) at (90:1) {};
    \node[vertex,lightgray,fill=lightgray] (p1) at (90:2.5) {};
    \node[vertex,lightgray,fill=lightgray] (p2) at (210:1) {};
    \node[vertex,lightgray,fill=lightgray] (p3) at (210:2.5) {};
    \node[vertex,lightgray,fill=lightgray] (p4) at (330:1) {};
    \node[vertex,lightgray,fill=lightgray] (p5) at (330:2.5) {};
    \draw[arrow,blue]
        (p0) edge (p2)
        (p2) edge (p4)
        (p4) edge (p0);
    \draw[arrow,blue, bend right=40]
        (p1) edge (p3)
        (p3) edge (p5)
        (p5) edge (p1);
    \draw[arrow,blue, bend right]
        (p0) edge (p3)
        (p3) edge (p4)
        (p4) edge (p1)
        (p1) edge (p2)
        (p2) edge (p5)
        (p5) edge (p0);

    \draw[edge,darkyellow] (p0) -- (p1);
    \draw[edge,darkyellow] (p2) -- (p3);
    \draw[edge,darkyellow] (p4) -- (p5);

\end{tikzpicture}
        \end{center}
        \subcaption{$(\overrightarrow{C_3}[K_2],3K_2)$.}
        \label{small-cc:6cc:wreath/fig}
    \end{subfigure}
    \hfil
    \begin{subfigure}[]{.27\textwidth}
        \begin{center}

\begin{tikzpicture}[scale=0.85]

    \def\degree{60}

    \foreach \x in {0,...,5}
    {
        \node[vertex,lightgray,fill=lightgray] (p\x) at (\x*\degree+60:1) {};
    }
    \draw[edge,darkred] (p0) -- (p1);
    \draw[edge,darkred] (p2) -- (p3);
    \draw[edge,darkred] (p4) -- (p5);

    \draw[edge,red] (p1) -- (p2);
    \draw[edge,red] (p3) -- (p4);
    \draw[edge,red] (p5) -- (p0);

    \draw[edge,darkyellow]
        (p0) edge (p3)
        (p1) edge (p4)
        (p5) edge (p2);

    \draw[arrow,blue]
        (p0) edge (p2)
        (p2) edge (p4)
        (p4) edge (p0);
    \draw[arrow,blue]
        (p1) edge (p5)
        (p5) edge (p3)
        (p3) edge (p1);
\end{tikzpicture}
        \end{center}
        \subcaption{$(3K_2,3K_2,2\overrightarrow{C_3},3K_2)$.}
        \label{small-cc:6cc:alternating-C6/fig}
    \end{subfigure}
    \hfil
    \begin{subfigure}[]{.23\textwidth}
        \begin{center}

\begin{tikzpicture}[scale=0.85]

    \def\degree{60}

    \foreach \x in {0,...,5}
    {
        \node[vertex,lightgray,fill=lightgray] (p\x) at (\x*\degree+60:1) {};
    }
    \draw[arrow,darkred]
        (p0) edge (p1)
        (p1) edge (p2)
        (p2) edge (p3)
        (p3) edge (p4)
        (p4) edge (p5)
        (p5) edge (p0);

    \draw[edge,darkyellow] (p0) -- (p3);
    \draw[edge,darkyellow] (p1) -- (p4);
    \draw[edge,darkyellow] (p5) -- (p2);

    \draw[arrow,blue]
        (p0) edge (p2)
        (p2) edge (p4)
        (p4) edge (p0);
    \draw[arrow,blue]
        (p1) edge (p3)
        (p3) edge (p5)
        (p5) edge (p1);
\end{tikzpicture}
        \end{center}
        \subcaption{$(\overrightarrow{C_6},2\overrightarrow{C_3},3K_2)$.}
        \label{small-cc:6cc:directed-C6/fig}
    \end{subfigure}
    \caption{Coherent configurations of order~$6$ with directed edges.}
    \label{small-cc:6cc:coherent-config/fig}
\end{figure}

\begin{lemma}
    \label{small-cc:interspace/lem}
    Let~$\coherentConfig$ be a coherent configuration, and let~$(R,B)$ be an edge in~$\quotientGraph{\coherentConfig}$ such that~$2 \leq \abs{R} \leq \abs{B} \leq 7$.
    \begin{enumerate}[label=(\arabic*)]
        \item If~$\abs{R} = 2$,
        then~$\abs{B} \in \{2,4,6\}$ and~$\disjointCliques{2}{1,\frac{\abs{B}}{2}} \in \interspace{R}{B}$.
        \item If~$\abs{R} = 3$,
        then~$\abs{B} \in \{3,6\}$ and~$\disjointCliques{3}{1,\frac{\abs{B}}{3}} \in \interspace{R}{B}$.
        \item If~$\abs{R} = 4$, then~$\abs{B} \in \{4,6\}$.
        Furthermore, if~$\abs{R} = 4$ and~$\abs{B} = 4$, then at least one of the constituents~$\cycle{8}$, $\matching{4}$, and~$\disjointCliques{2}{2,2}$ is contained in~$\interspace{R}{B}$, and if~$\abs{R} = 4$ and~$\abs{B} = 6$, then at least one of constituents~$\interspaceFourSix$ and~$\disjointCliques{2}{2,3}$ is contained in~$\interspace{R}{B}$.
        \item If~$\abs{R} = 5$,
        then~$\abs{B} = 5$ and~$\matching{5} \in \interspace{R}{B}$.
        \item If~$\abs{R} = 6$,
        then~$\abs{B} = 6$ and at least one of constituents~$\matching{6}$, $\cycle{12}$, $\disjointCliques{3}{2,2}$, or~$\disjointCliques{2}{3,3}$ is contained in~$\interspace{R}{B}$.
        \item If~$\abs{R} = 7$,
        then~$\abs{B} = 7$ and at least one of constituents~$\matching{7}$ and~$\leviFano$ is contained in~$\interspace{R}{B}$.
    \end{enumerate}
\end{lemma}
\begin{proof}
    The claims for~$\abs{R},\abs{B} \in \{1,\dots,4\}$ have been determined in \cite[Figures 2 and 3]{DBLP:journals/siamdm/FuhlbruckKV21}.
    Let~$U \in \interspace{R}{B}$ be such that~$\intDegree{U} = \minimalDegree{R}{B}$.
    Note that~$\intDegree{U} \in \{1,\dots,\lfloor \frac{\abs{B}}{2} \rfloor\}$.
    The constituent~$(R \disjointUnion B, U)$ satisfies~
    \begin{equation}
    \label{interspace:handshake/eq}
        \abs{R}\intDegree{U} = \abs{B} \intDegree{U^\star},
    \end{equation}
    and thus~$\intDegree{U^\star} \in \{1,\dots,\lfloor \frac{\abs{R}}{2} \rfloor\}$.
    Note~$\intDegree{U} = \intDegree{U^\star}$ if~$\abs{R} = \abs{B}$, and in  particular, the interspace~$\matching{\abs{R}} \in \interspace{R}{B}$ appears if~$\intDegree{U} = 1$.

    Now assume~$\abs{R} \leq 4 < \abs{B}$.
    If~$\abs{R} = 2$ (respectively~$\abs{R} = 3$), then by Equation~\eqref{interspace:handshake/eq} the size of~$B$ is~$6$ and~$\intDegree{U} = 3$ (respectively~$\intDegree{U} = 2$).
    Thus~$\disjointCliques{2}{1,3} \in \interspace{R}{B}$ (respectively~$\disjointCliques{3}{1,2} \in \interspace{R}{B}$).
    If~$\abs{R} = 4$, then~$\abs{B} = 6$,~$\intDegree{U} = 3$, and~$\intDegree{U^\star} = 2$.
    If there are distinct~$b,b' \in B$ such that~$bU^\star = b'U^\star$, then~$\disjointCliques{2}{2,3} \in \interspace{R}{B}$.
    Otherwise, for all distinct~$r,r' \in R$ there is a~$b$ with~$bU^\star = \{r,r'\}$.
    Hence~$\clique{4} \in \inducedCC{R}$ and~$\interspaceFourSix \in \interspace{R}{B}$.

    By Equation~\eqref{interspace:handshake/eq}, if~$\abs{R} = 5$, then~$\abs{B} = 5$ and~$\intDegree{U} = 2$.
    Due to Property~\ref{coherent-config:wl2}, we have~$\cycle{10} \in \interspace{R}{B}$ but also~$\matching{5} \in \interspace{R}{B}$.

    In the upcoming cases of~$\abs{R} \in \{6,7\}$, a finer case distinction will be necessary.
    So for the rest of the proof, let~$G$ be a constituent which is contained in~$\inducedCC{R}$, and let~$r,r'\in  R$ be such that~$rr' \in \arcs(G)$.
    Further, define~$s \coloneqq \abs{rU \cap r'U}$. We will repeatedly use the following \emph{coherence criterion}:
    if~$\Deg(G)$ is the degree of the graph underlying~$G$, then~$\Deg(G)\cdot s$ is a multiple of~$\intDegree{U}$.

    Assume~$\abs{R} = 6$. By Equation~\eqref{interspace:handshake/eq} we know that~$\abs{B} = 6$.
    If~$\intDegree{U} = 2$, then at least one of the following constituents is contained in~$\interspace{R}{B}$:
    $\cycle{12}$, $\disjointCliques{3}{2,2}$, or~$2\cycle{6}$.
    In the last case, we also have~$\matching{6} \in \interspace{R}{B}$ due to Property~\ref{coherent-config:wl2}.
    All other possible constituents imply a violation of Property~\ref{coherent-config:wl2}.
    Assume now that~$\intDegree{U} = 3$:

    \begin{enumerate}[label=(\arabic*)]
        \item
        $G \cong \clique{6}$.
        We will double count how many arcs from~$U$ have an endpoint in~$rU$:
        recall that~$\intDegree{U} = \intDegree{U^\star} = 3$ since~$\abs{R} = \abs{B}$.
        Thus there are~$\sum_{b \in rU} 3 = 9$ such arcs.
        On the other hand, there are~$3$ such arcs which start in~$r$.
        Furthermore, each~$r' \in R\setminus\{r\}$ has exactly~$s$ common neighbors with~$r$, so~$5s$ arcs that do not start in~$r$.
        Together, we have~$3+5s \neq 9$ for all~$s \in \{0,\dots,3\}$.
        Thus we obtain a contradiction.

        \item
        $G \cong \cycle{6}$ or~$G \cong \overrightarrow{C_6}$.
        We have~$\Deg(G)\in \{1,2\}$ and~$\intDegree{U}=3$. So~$s\notin \{1,2\}$ by the coherence criterion. If~$s=3$, then, since~$G$ is connected, all vertices of~$R$ would have the same neighbors under~$U$ in~$B$ which violates coherence. If~$s=0$, then~$U$ forms an~$\disjointCliques{2}{3,3}$ interspace.

        \item
        $G \cong \disjointCliques{2}{3}$ or~$G \cong 2\overrightarrow{\cycle{3}}$.
        Due to the coherence criterion~$s \notin \{1,2\}$.
        Also~$s=0$ is impossible since three vertices cannot have pairwise non-adjacent neighborhoods under~$U$ each covering half the vertices of~$B$.
        If~$s = 3$, then~$\disjointCliques{2}{3,3} \in \interspace{R}{B}$.

        \item
        $G \cong \matchingCC{3}$. Due to the coherence criterion~$s\notin \{1,2\}$.
        A contradiction of Equation~\eqref{interspace:handshake/eq} is implied if~$s = 3$ because then vertices of~$B$ have an even number of incoming arcs from~$U$.
        Suppose~$s = 0$ and let~$r'',r''' \in R$ such that~$r''r''' \in \arcs(G)$.
        If~$\abs{r U \cap r''U} \in \{0,3\}$, then~$\disjointCliques{2}{3,3} \in \interspace{R}{B}$.
        If~$\abs{r U \cap r''U} \in \{1,2\}$, then~$\abs{r U \cap r''U} = \abs{r' U \cap r'''U}$ and implying that~$\cycle{6} \in \inducedCC{B}$.
    \end{enumerate}

    Assume~$\abs{R} = \abs{B} = 7$.
    If~$\intDegree{U} = 2$, then~$\cycle{14} \in \interspace{R}{B}$ and~$\matching{7} \in \interspace{R}{B}$.
    So assume~$\intDegree{U} = 3$.

    \begin{enumerate}[label=(\arabic*)]
        \item
        $G \cong \clique{7}$.
        We will double count how many arcs from~$U$ have an endpoint in~$rU$:
        recall that~$\intDegree{U} = \intDegree{U^\star} = 3$ since~$\abs{R} = \abs{B}$.
        Thus we have~$\sum_{b \in rU} 3 = 9$.
        On the other hand, there are~$3$ ingoing arcs which start in~$r$.
        Furthermore, each~$r' \in R\setminus\{r\}$ has exactly~$s$ common neighbors with~$r$ so~$6s$ arcs that do not start in~$r$.
        Together, we have~$3+6s \neq 9$ for all~$s \in \{0,2,3\}$.
        If~$s = 1$, then~$\leviFano \in \interspace{R}{B}$ since the interspace satisfies the combinatorial properties of a projective plane.

        \item
        $G \cong \cycle{7}$ or~$G \cong \overrightarrow{\cycle{7}}$.
        As before, since~$G$ is connected, $s = 3$ is impossible.
        By the coherence criterion~$s\notin \{1,2\}$.
        In fact, this means that $\abs{\overline{r}U \cap \overline{r'}U} = 0$ for all~$\overline{r},\overline{r'}\in R$ since each such pair forms an arc in some constituent of~$R$ isomorphic to~$\cycle{7}$ or~$\overrightarrow{\cycle{7}}$ (See Table~\ref{small-cc:classificaiton-small-cc/tab}). This is however a contradiction since~$7$ is not divisible by~$3$.

        \item
        $G \cong PTr(7)$. Recall that~$G$ is a complete oriented graph.
        If~$s = 3$, all vertices of~$R$ have the same neighbors under~$U$ since~$G$ is connected, so the interspace is trivial.
        If~$s = 1$, then~$\leviFano \in \interspace{R}{B}$, because the combinatorial properties of a projective plane are fulfilled. Suppose~$R = \{r_1,\dots,r_7\}$.
        If~$s = 0$, then~$r_1U \cap r_2U = \emptyset$ and~$r_1U \cap r_3U = \emptyset$.
        This implies~$r_3U \cap r_2U \neq \emptyset$, a contradiction.
        Finally if~$s=2$ then at least~$2\cdot 7=14$ of the arcs from~$U$ end in~$r_1U$.
        But only~$|r_1U|\cdot d(U^\star)= 3\cdot 3=9$ of the arcs from~$U$ can end in~$r_1U$.\qedhere
    \end{enumerate}
\end{proof}


\begin{table}
    \centering\def\arraystretch{1.2}%
    \begin{tabular}{|c|c|l|}
        \hline
        $\abs{R}$ & $\abs{B}$ & $\type{\interspace{R}{B}}$ \\ \hline
        4 & 4 & $(\cycle{8}, \cycle{8})$, $(\disjointCliques{2}{2,2}, \disjointCliques{2}{2,2})$\\ \hline
        4 & 6 & $(\interspaceFourSix, \interspaceFourSix)$, $(\disjointCliques{2}{2,3},\disjointCliques{2}{2,3})$\\ \hline
        6 & 6 & \makecell[l]{$(\cycle{12}, \cycle{12}, \disjointCliques{3}{2,2})$, $(\disjointCliques{2}{3,3},\disjointCliques{2}{3,3})$,\\$(\disjointCliques{3}{2,2}, \disjointCliques{3}{2,2}, \disjointCliques{3}{2,2})$, $(\disjointCliques{3}{2,2},R \times B - \disjointCliques{3}{2,2})$} \\ \hline
        7 & 7 & $(\leviFano, R \times B -\leviFano)$ \\ \hline
    \end{tabular}
    \caption{Classification of all interspaces~$\interspace{R}{B}$ between fibers~$R$ and~$B$ with~$4$ $\leq \abs{R} \leq \abs{B} \leq 7$ in critical coherent configurations.}
    \label{small-cc:classificaiton-small-interspaces/tab}
\end{table}

To classify all possible isomorphism types of interspaces between small fibers, we can examine the coarsest coherent configurations containing one of the graphs guaranteed to exist by Lemma~\ref{small-cc:interspace/lem} as constituent. Table~\ref{small-cc:classificaiton-small-interspaces/tab} lists all isomorphism types between two small fibers in a critical coherent configuration. Note that the isomorphism type is determined by the isomorphism types of its constituents. (This is not necessarily the case for interspaces between larger fibers). For fibers of size~$4$ and~$6$ the classification can alternatively be derived from the upcoming Theorem~\ref{global-argument:large-small-interspace:classification}.

\begin{lemma}
\label{small-cc:interspace-implies-cc/lem}
    Let~$\coherentConfig$ be a coherent configuration, and let~$(R,B)$ be an edge in~$\quotientGraph{\coherentConfig}$ with~$\abs{R} \leq \abs{B}$, and suppose~$x \in \{4,6\}$ and~$y \in \{2,3\}$.
    \begin{enumerate}[label = (\arabic*)]
        \item
        If~$\cycle{2x} \in \interspace{R}{B}$, then~$\cycle{x} \in \inducedCC{R}$.
        \item
        If~$y\clique{2,2} \in \interspace{R}{B}$, then~$\disjointCliques{y}{2} \in \inducedCC{R}$.
        \item
        If~$\disjointCliques{2}{2,3} \in \interspace{R}{B}$, then~$\matchingCC{2} \in \inducedCC{R}$ and either~$\disjointCliques{2}{3} \in \inducedCC{B}$ or~$2\overrightarrow{C_3} \in \inducedCC{B}$.
        \item
        If~$\disjointCliques{2}{3,3} \in \interspace{R}{B}$, then either~$\disjointCliques{2}{3} \in \inducedCC{R}$ or~$2\overrightarrow{\cycle{3}} \in \inducedCC{B}$.
        \item
        If~$\leviFano \in \interspace{R}{B}$, then either~$\clique{7} \in \inducedCC{R}$ or~$PTr(7) \in \inducedCC{R}$.
        \item
        If~$\interspaceFourSix \in \interspace{R}{B}$, then~$\clique{4} \in \inducedCC{R}$ and either~$\matchingCC{3}, K_{2,2,2} \in \inducedCC{B}$ or~$\matchingCC{3},\overrightarrow{\cycle{3}}[\clique{2}] \in \inducedCC{B}$.
    \end{enumerate}
\end{lemma}
\begin{proof}
    The first five claims follow from Property~\ref{coherent-config:wl2} and from the proof of Lemma~\ref{small-cc:interspace/lem}.
    Consider the last claim and suppose~$U \in \interspace{R}{B}$ with~$(R \disjointUnion B, U) \cong \interspaceFourSix$.
    Then~$\clique{4} \in \inducedCC{R}$ since otherwise~$R$ violates Property~\ref{coherent-config:wl2}.
    For all~$r \in R$ the neighborhood~$rU$ induces up to isomorphism the same subgraph in~$\inducedCC{B}$.
    Thus either~$\clique{3} \in \inducedCC{rU}$ or~$\overrightarrow{\cycle{3}} \in \inducedCC{rU}$.
    In the first case~$\clique{2,2,2} \in \inducedCC{B}$, in the latter case~$\overrightarrow{\cycle{3}}[\clique{2}]$.
\end{proof}

\begin{theorem}
    \label{critical:small-cc/thm}
    Let~$\coherentConfig$ be a critical coherent configuration.
    If all fibers of~$\coherentConfig$ are small, then either
    \begin{enumerate}
        \item all fibers have the same size, or
        \item for all~$R \in \fibers{\coherentConfig}$ we have~$\abs{R} \in \{4,6\}$.
    \end{enumerate}
\end{theorem}
\begin{proof}
    This follows directly from Lemmas~\ref{critical:disjoint-union/lem} and~\ref{small-cc:interspace/lem}.
\end{proof}


\section{Interspaces between large and small fibers}
\label{interspace-large-small/sec}

Our next goal is to study the possible interspaces between a small fiber~$S$ and a fiber~$L$ of arbitrary size.
Since~$L$ is not bounded in size, there is of course an infinite number of isomorphism types that appear.
The challenge is to capture the all relevant information about the structure between the fiber~$L$ and the small fiber~$S$ using only a finite amount of data.

Our solution turns out to be rather technical, but the idea is simple. If we forget the structure inside the fiber~$L$, vertices in~$L$ will partition into equivalence class depending on their neighborhood in~$S$ (taking basis relations of arcs into account). We now record the isomorphism type of the interspace we obtain by just retaining one representative from each class. However, we also capture, for one vertex~$\ell$ of~$L$ what structure the neighborhoods with respect to each basis relation induce in~$S$. Due to coherence and the small size of~$S$, this structure is the same (up to isomorphism) for every vertex of~$L$. Because~$S$ has size at most~$6$, it turns out that it is essentially sufficient to check the degrees of~$\ell$ under the basis relations and check which neighborhoods induce monochromatic cliques, in order to capture almost all relevant information.

To formally define the concept described above, one auxiliary notation is needed:
for a binary relation~$\arcs$, we define~$\ul(\arcs)$ to be~$\arcs \cup \arcs^\star$.
We extend this notation to graphs and coherent configurations to obtain their \emph{underlying undirected structure}.
For a coherent configuration~$\coherentConfig$, we set~$\ul(\coherentConfig)$ to be~$(\vertices(\coherentConfig), \{\ul(\arcs) \mid \arcs \in \relations(\coherentConfig)\})$.
We remark that~$\ul(\coherentConfig)$ is not necessarily coherent.
For a graph~$G$, we set~$\ul(G)$ to be~$(\vertices(G), \ul(\arcs(G)))$ and call it the \emph{underlying undirected graph} of~$G$.

We now introduce interspace patterns, before commenting on peculiarities of their definition:

\begin{definition}
\label{large-small-interspace:interspace-pattern/def}
    Let~$L,S$ be distinct fibers of a coherent configuration~$\coherentConfig$.
    We say the interspace~$\interspace{L}{S}$ has the \emph{interspace pattern} $\interspacePattern{G^1, d^1_1, d^1_2, \dots, d^1_{t_1}; \dots; G^k, d^k_1, d^k_2, \dots, d^k_{t_k}}$ if the three following conditions hold:
    $\sum_{j = 1}^{k} t_j = \abs{\interspace{L}{S}} - 1$, there are distinct basis relations~$\arcs^1, \dots, \arcs^k \in \inducedCC{S}$ such that for each~$i \in \{1,\dots,k\}$ we have~$\ul(S,\arcs^i) \cong G^i$, and there are distinct~$U^i_1, \dots, U^i_{t_i} \in \interspace{L}{S}$ such that for all~$\ell \in L$ and~$j \in \{1, \dots, t_i\}$ the set~$\ell U^i_j$ induces a~$d^i_j$-clique in~$\ul(S,A^i)$.

    If~$d^i_j = 3$ and for all~$S' \subseteq S$ that induce a~$3$-clique in~$\ul(S,A^i)$ the common neighborhood~$\bigcap_{s \in S'} s {U^i_j}^\star$ is not empty, then we mark the entry in the interspace pattern by~$^\dag$.
    If~$d^i_j =  3$ and for all~$S' \subseteq S$ inducing a~$3$-clique in~$\ul(S,A^i)$ we have that~$S \setminus S'$ is a also~$3$-clique and exactly one of~$\bigcap_{s \in S'} s {U^i_j}^\star$ or~$\bigcap_{s \in S \setminus S'} s {U^i_j}^\star$ is not empty, then we mark the entry by~$^\ddag$.
\end{definition}

We should emphasize that the sum of the~$t_i$ is~$\abs{\interspace{L}{S}} - 1$ and thus we omit exactly one basis relation of the interspace in the pattern.
In fact, in general the missing interspace will not satisfy a clique condition. Also, it is a priori not clear that an interspace always has a pattern of this form we just defined.

Regarding the subpatterns, if entry~$d^i_j$ is marked by~$^\dag$, then for every~$S' \subseteq S$ which induces a~$3$-clique in~$\ul(S,A^i)$ there is a vertex in~$\ell\in L$ such that~$\ell U^i_j= S'$.
However, if it is marked~$^\ddag$, then only for either~$S'$ or~$S \setminus S'$ such an~$\ell$ exists.
For example, the interspace pattern~$(\clique{2,2,2},3)$ has two subpatterns, namely~$\ipsixMatchingComplement$ and~$\ipsixMatchingComplementD$.
In principle, there could be other subcases, but it will follow from our classification that only these two subcases arise.

\begin{theorem}
\label{global-argument:large-small-interspace:classification}
    Let~$\coherentConfig$ be a critical coherent configuration and~$(L,S)$ be an edge in~$\quotientGraph{\coherentConfig}$.
    If~$\abs{S} \in \{4,6\}$, then~$\interspace{L}{S}$ has one of the following interspace patterns:
    \begin{multicols}{3}
        \begin{enumerate}
            \item \ipfourClique
            \item \ipfourMatching
            \item \ipfourCycle
            \item \ipsixCliqueTwo
            \item \ipsixCliqueTwoTwice
            \item \ipsixMatching
            \item \ipsixMatchingTwice
            \item \ipsixMatchingAndCycle
            \item \ipsixMatchingMatching
            \item \ipsixMatchingAndComplement
            \item \ipsixTriangleComplement
            \item \ipsixTriangleComplementTwice
            \item \ipsixCliqueThree
            \item \ipsixCliqueThreeD
            \item \ipsixTriangle
            \item \ipsixMatchingComplement
            \item \ipsixMatchingComplementD.
        \end{enumerate}
    \end{multicols}
\end{theorem}
\begin{proof}
    Since~$\abs{S} \in \{4,6\}$, we have~$\minimalDegree{L}{S} \in \{2,3\}$ and~$\abs{\interspace{L}{S}} \in \{2,3\}$.

    Let us first observe that for every basis relation~$U\in \interspace{L}{S}$ with~$|\ell U|\in \{2,3\}$ and every~$\ell \in L$ the set~$lU$ forms a clique in some undirected underlying graph of a constituent of~$\inducedCC{S}$. Indeed, from coherence, it follows that every constituent of $\inducedCC{\ell U}$ is a regular graph. Thus, since~$|\ell U|\in \{2,3\}$, there is only a single constituent, which must be a~$2$-clique or a~$3$-clique. Note that the size of the clique is independent of the choice of~$\ell$ due to coherence. Also note that, again due to coherence, the constituent that contains the clique $\inducedCC{\ell U}$ must be the same for different choices of~$\ell$.

    Finally observe that, since~$\abs{S} \in \{4,6\}$, the interspace~$\interspace{L}{S}$ contains at most one basis relation~$U \in \interspace{L}{S}$ with~$\intDegree{U} >3$.
    Thus we conclude that~$\interspace{L}{S}$ has some interspace pattern of the form $\interspacePattern{G^1, d^1_1, d^1_2, \dots, d^1_{t_1}; \dots; G^k, d^k_1, d^k_2, \dots, d^k_{t_k}}$.

    First assume~$\abs{S} = 4$, and thus we have~$\abs{\interspace{L}{S}} = 2$.
    Then~$k = 1$, $t_1 = 1$, and~$d^1_1 = 2$.
    If there are~$\ell \in L$ and~$U \in \interspace{L}{S}$ such that~$\ell U$ induces a clique in~$\ul(G)$ where~$G$ is constituent of~$\inducedCC{S}$ isomorphic to~$\overrightarrow{C_4}$, then~$\inducedCC{\ell U}$ is not regular.
    Therefore by Lemma~\ref{small-cc:induced-cc/lem}, each constituent~$G$ of~$\inducedCC{S}$ (and thus the undirected graph~$G^1$) is isomorphic to one of the graphs
    $\clique{4}$,~$\cycle{4}$, or~$\disjointCliques{2}{2}$.
    Observe that for all distinct~$s,s' \in S$ which are adjacent in~$G$, the two vertices in~$S \setminus \{s,s'\}$ are also adjacent in~$G$.
    Thus~$\interspace{L}{S}$ has one of the following interspace patterns
    $\ipfourClique$,~$\ipfourCycle$, or~$\ipfourMatching$.

    Now assume~$\abs{S} = 6$.
    Let~$\ell \in L$,~$U \in \interspace{L}{S}$, and~$G$ be constituent of~$\inducedCC{S}$ such that~$\ell U$ induces a clique in~$\ul(G)$.
    By Lemma~\ref{small-cc:induced-cc/lem} the constituent~$G$ is isomorphic to one of the graphs $K_6$,~$2 K_3$,~$2\overrightarrow{C_3}$,~$K_{3,3}$,~$3 K_2$,~$K_{2,2,2}$,~$\overrightarrow{C_3}[K_2]$,
    $C_6$, or~$\overrightarrow{C_6}$.
    We distinguish cases according to~$G$.

    \textit{(Case~$G \cong C_6$)}.
    By Lemma~\ref{small-cc:induced-cc/lem} there must exist constituents~$G', G''$ of~$\inducedCC{S}$ such that~$G' \cong 3K_2$ and~$G'' \cong 2K_3$.
    Observe~$\intDegree{U} \neq 3$ since~$\cycle{6}$ does not contain a~$3$-clique.
    Thus~$\intDegree{U} = 2$.
    Assume~$\abs{\interspace{L}{S}} = 2$ and let~$\overline{U} \in \interspace{L}{S}$ with~$\overline{U} \neq U$.
    Then exactly two vertices in~$\ell \overline{U}$ have distance~$3$ in~$G$, which contradicts coherence.
    Thus we know that~$\abs{\interspace{L}{S}} = 3$.  So let~$\{U,\overline{U},\widetilde{U}\} = \interspace{L}{S}$.
    If both~$\ell \overline{U}$ and~$\ell \widetilde{U}$ induce $2$ cliques in~$G$, then there is exactly one pair~$(\overline{s},\widetilde{s}) \in \ell \overline{U} \times \ell \widetilde{U}$ such that~$\overline{s}\widetilde{s} \in \arcs(G')$.
    Thus by coherence~$\overline{U}$ cannot be a basis relation.
    A similar argument applies if both~$\ell \overline{U}$ and~$\ell \widetilde{U}$ induce $2$ cliques in~$G''$.
    In the last possible case,~$\ell \overline{U}$ induces a~$2$-clique in~$\ul(G')$ and~$\ell \widetilde{U}$ induces a~$2$-clique in~$\ul(G)$.
    Then~$\interspace{L}{S}$ has the interspace pattern~$\ipsixMatchingAndCycle$.

    \textit{(Case~$G \cong \overrightarrow{C_6}$)}.
    Once again, we observe that~$\intDegree{U} \neq 3$ since~$\cycle{6}$ does not contain a~$3$-clique.
    Thus~$\intDegree{U} = 2$.
    However~$\inducedCC{\ell U}$ is not regular.

    \textit{(Case~$G \cong 2 K_3$)}.
    If~$\intDegree{U} = 2$, then~$G-\ell U$ is not regular, which contradicts coherence.
    Thus~$\intDegree{U} = 3$. Then~$\interspace{L}{S}$ has the interspace pattern~$\interspacePattern{2K_3,3}$.
    To be precise, the interspace~$\interspace{L}{S}$ has the interspace subpattern~$\ipsixTriangle$:
    there are exactly two distinct vertex sets~$S',S \setminus S' \subseteq S$ inducing~$3$-cliques in~$G$ and for each vertex~$s \in S$ either~$s \in S'$ or~$s \in S \setminus S'$.
    Without loss of generality, assume~$\ell U  = S'$.
    Thus by Property~\ref{coherent-config:wl2} there is~$\ell' \in L$ such that~$\ell' U = S \setminus S'$.

    \textit{(Case~$G \cong 2\overrightarrow{C_3}$)}.
    This case is analogous to the previous case.

    \textit{(Case~$G \cong K_{3,3}$)}.
    Observe~$\intDegree{U} \neq 3$ since~$K_{3,3}$ does not contain a~$3$-clique.
    Thus~$\intDegree{U} = 2$.
    By Lemma~\ref{small-cc:induced-cc/lem} there is a constituent~$G'$ in~$\inducedCC{S}$ isomorphic to either~$2K_3$ or~$2\overrightarrow{C_3}$.
    If~$|\interspace{L}{S}| = 2$, then~$\interspace{L}{S}$ has the interspace pattern~$\ipsixTriangleComplement$ and~$G' \cong 2K_3$.
    Now we assume that~$|\interspace{L}{S}| = 3$.
    Let~$\{U,\overline{U},\widetilde{U}\} = \interspace{L}{S}$.
    If~$\ell \overline{U}$ induces a~$2$-clique in~$\ul(G')$, then~$G' - \ell \overline{U}$ is not regular, which contradicts coherence.
    If~$\ell \overline{U}$ induces a~$2$-clique in~$\ul(G)$, then~$\ell \widetilde{U}$ induces a~$2$-clique in~$\ul(G)$ as well.
    Hence, the interspace~$\interspace{L}{S}$ has the interspace pattern~$\ipsixTriangleComplementTwice$.

    \textit{(Case~$G \cong 3 K_2$)}.
    So~$\intDegree{U} = 2$.
    If~$|\interspace{L}{S}| = 2$, then~$\interspace{L}{S}$ has the interspace pattern~$\ipsixMatching$.
    Assume~$|\interspace{L}{S}| = 3$, and let~$\{U,\overline{U},\widetilde{U}\} = \interspace{L}{S}$.
    If~$\ell \overline{U}$ induces a~$2$-clique in~$\ul(G)$, then~$\ell \widetilde{U}$ induces a~$2$-clique in~$\ul(G)$ as well.
    Hence~$\interspace{L}{S}$ has the interspace pattern~$\ipsixMatchingTwice$.

    Let~$G'$ be another constituent of~$\inducedCC{S}$ such that~$\ell \overline{U}$ and~$\ell \widetilde{U}$ induce a~$2$-clique in~$\ul(G')$.
    Since we have already covered the previous cases, we can assume~$G'$ is not isomorphic to~$2K_3$, $2\overrightarrow{C_3}$, $C_6$, or~$\overrightarrow{C_6}$.
    If~$G' \cong \overrightarrow{C_3}[K_2]$, then~$\coherentConfig[\ell \overline{U}]$ is not regular.
    If~$G' \cong K_{2,2,2}$, then~$\interspace{L}{S}$ has the interspace pattern~$\ipsixMatchingAndComplement$.

    Finally, let $G', G''$ be two other constituent of~$\inducedCC{S}$ such that~$\ell \overline{U}$ (respectively~$\ell \widetilde{U}$) induces a~$2$-clique in~$\ul(G')$ (respectively~$\ul(G'')$).
    Since we have already covered the previous cases, we can assume~$G'$ is not isomorphic to~$2K_3$, $2\overrightarrow{C_3}$, $C_6$, or~$\overrightarrow{C_6}$.
    If~$G'$ and~$G''$ are both isomorphic to~$3K_2$, then~$\interspace{L}{S}$ has the interspace pattern~$\ipsixMatchingMatching$.

    \textit{(Case~$G \cong K_{2,2,2}$)}.
    By Lemma~\ref{small-cc:induced-cc/lem} there is one other constituent~$G'$ in~$\inducedCC{S}$ isomorphic to~$3K_2$.

    Assume~$|\interspace{L}{S}| = 2$ and~$\intDegree{U} = 2$.
    Let~$\overline{U} \in \interspace{L}{S}$ with~$\overline{U} \neq U$.
    Hence there are exactly two vertices~$s,s' \in \ell \overline{U}$ such that~$ss' \in \arcs(G)$.
    By coherence~$\overline{U}$ cannot be a basis relations of~$\interspace{L}{S}$.

    Assume~$|\interspace{L}{S}| = 3$ and let~$\{U,\overline{U},\widetilde{U}\} = \interspace{L}{S}$.
    Suppose both~$\ell \overline{U}$ and~$\ell \widetilde{U}$ induce a~$2$-clique in~$\ul(G)$.
    Then there is exactly one pair~$(\overline{s},\widetilde{s}) \in \ell \overline{U} \times \ell \widetilde{U}$ such that~$\overline{s}\widetilde{s} \in \arcs(G)$.
    By coherence~$\overline{U}$ cannot be a basis relation.
    The cases where~$\ell \overline{U}$ or~$\ell \widetilde{U}$ induce a~$2$-clique in~$\ul(G')$ are covered by the previous case.

    Finally, assume~$|\interspace{L}{S}| = 2$ and~$\intDegree{U} = 3$.
    If for all induced $3$-cliques~$S'$ in~$G$ there is an~$\ell \in L$ with~$\ell U = S'$, then~$\interspace{L}{S}$ has the interspace subpattern~$\ipsixMatchingComplement$.
    Suppose that there is an induced $3$-clique~$S'$ in~$G$ such that for all~$\ell' \in L$ we have~$\ell' U \neq S'$ and~$\ell' U \neq S\setminus S'$.
    Let~$s \in S' \cap \ell U$,~$s' \in \ell U \cap (S \setminus S')$, and~$s'' \in S' \cap (S \setminus \ell U)$.
    Then~$s$ and~$s'$ have a different number of common neighbors in~$L$ under~$U$ than~$s$ and~$s''$.
    By coherence~$G$ is not a constituent of~$\inducedCC{S}$.
    If for all induced $3$-cliques~$S'$ in~$G$ there is an~$\ell \in L$ with either~$\ell U = S$ or~$\ell U = S \setminus S'$, then~$\interspace{L}{S}$ has the interspace subpattern~$\ipsixMatchingComplementD$.

    \textit{(Case~$G \cong \overrightarrow{C_3}[K_2]$)}.
    This case works similar to the previous case.

    \textit{(Case~$G \cong K_6$)}.
    If~$|\interspace{L}{S}| = 2$ and~$\intDegree{U} = 2$, then~$\interspace{L}{S}$ has the interspace pattern~$\ipsixCliqueTwo$.
    If~$|\interspace{L}{S}| = 3$, then~$\interspace{L}{S}$ has the interspace pattern~$\ipsixCliqueTwoTwice$.
    Finally assume~$|\interspace{L}{S}| = 2$ and~$\intDegree{U} = 3$.
    If for all induced $3$-cliques~$S'$ in~$G$ there is an~$\ell \in L$ with~$\ell U = S'$, then~$\interspace{L}{S}$ has the interspace subpattern~$\ipsixCliqueThree$.
    Suppose that there is an induced $3$-clique~$S'$ in~$G$ such that for all~$\ell' \in L$ we have~$\ell' U \neq S'$ and~$\ell' U \neq S\setminus S'$.
    Let~$s \in S' \cap \ell U$,~$s' \in \ell U \cap (S \setminus S')$, and~$s'' \in S' \cap (S \setminus \ell U)$.
    Then~$s$ and~$s'$ have a different number of common neighbors in~$L$ under~$U$ than~$s$ and~$s''$.
    Hence~$G$ is not a constituent of~$\inducedCC{S}$.
    If for all induced $3$-cliques~$S'$ in~$G$ there is an~$\ell \in L$ with either~$\ell U = S$ or~$\ell U = S \setminus S'$, then~$\interspace{L}{S}$ has the interspace subpattern~$\ipsixCliqueThreeD$.
\end{proof}

We can observe that the interspace patterns in the theorem are mutually exclusive, and we conclude the following.

\begin{corollary}
\label{large-small-interspace:classification:uniqueness/cor}
    Let~$\coherentConfig$ be a critical coherent configuration, and let~$(L,S)$ be an edge in~$\quotientGraph{\coherentConfig}$.
    If~$\abs{S} \in \{4,6\}$, then~$\interspace{L}{S}$ has \underline{exactly} one of the interspace patterns listed in Theorem~\ref{global-argument:large-small-interspace:classification}.
\end{corollary}
\begin{proof}
    The proof follows from the previous theorem as follows.
    Assuming~$\interspace{L}{S} = \{U_1,\ldots,U_u\}$, let~$\mathcal{G}$ be the multiset of constituents~$G_i$ of~$\inducedCC{S}$ such that for all~$i \in \{1,\ldots,u\}$ and~$\ell \in L$ the set~$\ell U_i$ induces a clique in~$\ul(G_i)$.
    Then each possible coherent configuration from Table~\ref{small-cc:classificaiton-small-cc/tab} together with~$\mathcal{G}$ uniquely determines the interspace pattern of~$\interspace{L}{S}$.

\end{proof}

We should stress that this classification of the interspaces between fibers~$L$ and~$S$ based on interspace patterns is only possible because we consider small fibers with~$S$ of size~$4$ or~$6$.

For an interspace pattern~$\mathfrak{P}$ of~$\interspace{L}{S}$ we write~$A^i(\interspace{L}{S})$ and~$U^i_j(\interspace{L}{S})$ to denote~$A^i$ and~$U^i_j$ as given in the definition of the interspace pattern. While the interspace pattern~$\mathfrak{P}$ is unique for interspace~$\interspace{L}{S}$, the choice for~$A^i$ is only determined up to isomorphism. For a given coherent configuration we will always choose an arbitrary but fixed~$A^i$ and matching~$U^i_j$.

\begin{lemma}
\label{global-argument:k4-exclusion/lem}
    Let~$\coherentConfig$ be a critical coherent configuration and~$(R,Y,B)$ a path in~$\quotientGraph{\coherentConfig}$.
    If either~$\clique{4} \in \inducedCC{Y}$ or~$\clique{6} \in \inducedCC{Y}$, then~$\interspace{R}{B}$ is not homogeneous.
\end{lemma}
\begin{proof}
    Suppose~$r \in R$, $U_R = U^1_1(\interspace{R}{Y})$, and~$U_B = U^1_1(\interspace{B}{Y})$.
    By Theorem~\ref{global-argument:large-small-interspace:classification} we know that both~$\interspace{R}{Y}$ and~$\interspace{B}{Y}$ have one of the following interspace patterns:~$\ipfourClique$, $\ipsixCliqueTwo$, $\ipsixCliqueThree$, $\ipsixCliqueThreeD$. Thus, in all cases there are~$b, b' \in B$ such that~$|b U_B \cap r U_R|\neq |b' U_B \cap r U_R|$. Thus~$\interspace{R}{B}$ is not homogeneous.
\end{proof}

Let~$\interspacePattern{G^1, d^1_1, d^1_2, \dots, d^1_{t_1}; \dots; G^k, d^k_1, d^k_2, \dots, d^k_{t_k}}$ be the interspace pattern of~$\interspace{L}{S}$, and let~$i \in \{1,\ldots,k\}$ and~$j \in \{1,\ldots,t_i\}$.
For~$U = U^i_j(\interspace{L}{S})$, two vertices~$\ell, \ell' \in L$ are called \emph{equivalent} with respect to~$S$ and~$U$ if~$\ell U =\ell'U$.
We denote the set of all equivalence classes with respect to~$S$ and~$U^i_j(\interspace{L}{S})$ by~$\partitionRel{i}{j}{L,S}$.
Further, we define~
\[
    \equivalenceClasses{L,S} \coloneqq \bigwedge_{\substack{i \in \{1,\ldots,k\} \\ j \in \{1,\ldots,t_i\}}}  \partitionRel{i}{j}{L,S},
\]
that is,~$\equivalenceClasses{L,S}$ is the meet of all partitions~$\partitionRel{i}{j}{L,S}$.
In other words,~$\equivalenceClasses{L,S}$ is the coarsest partition which is still finer than all~$\partitionRel{i}{j}{L,S}$. Also note that~$\equivalenceClasses{L,S} = \partition{L,S}$ if~$|\interspace{L}{S}|=2$.
For a union of small fibers~$\mathcal{S}$, we define~$\equivalenceClasses{L,\mathcal{S}} \coloneqq \bigwedge_{S \in \mathcal{S}} \equivalenceClasses{L,S}$.

\begin{lemma}
\label{interspace-pattern:partition-size/lem}
    Let~$\coherentConfig$ be a critical coherent configuration, and let~$(L,S)$ be an edge in~$\quotientGraph{\coherentConfig}$ such that~$|S|\in \{4,6\}$.
    Further, assume that interspace~$\interspace{L}{S}$ has the interspace pattern~$\interspacePattern{G^1, d^1_1, \dots, d^1_{t_1}; \dots; G^k, d^k_1, \dots, d^k_{t_k}}$.
    Let~$x$ be the number of $3$-cliques in~$G^1$.

    Then~$\abs{\partition{L,S}} = \abs{A^1(\interspace{L}{S})}$ if~$d^1_1 = 2$, $\abs{\equivalenceClasses{L,S}} = x$ if~$d^1_1 = 3^\dag$, and~$\abs{\equivalenceClasses{L,S}} = \frac{x}{2}$ if~$d^1_1 = 3^\ddag$.
\end{lemma}
\begin{proof}
    This follows by inspecting the interspaces described in Theorem~\ref{global-argument:large-small-interspace:classification}.
\end{proof}

Table~\ref{interspace-pattern:partition-size/tab} gives an explicit overview of how the partition size~$\abs{\partition{L,S}}$ depends on the interspace patterns of Theorem~\ref{global-argument:large-small-interspace:classification}.


\begin{table}[tbp]
    \centering\def\arraystretch{1.2}%
    \begin{tabular}{|c|l|}
        \hline
        $\abs{\partition{L,S}}$ & interspace pattern of~$\interspace{L}{S}$                   \\ \hline
        $2$                     & $\ipfourMatching$, $\ipsixTriangle$                                   \\ \hline
        $3$                     & $\ipsixMatching$, $\ipsixMatchingTwice$, $\ipsixMatchingMatching$     \\ \hline
        $4$                     & $\ipfourCycle$, $\ipsixMatchingComplementD$                           \\ \hline
        $6$                     & $\ipsixMatchingAndCycle$, $\ipfourClique$                             \\ \hline
        $8$                     & $\ipsixMatchingComplement$                                            \\ \hline
        $9$                     & $\ipsixTriangleComplement$, $\ipsixTriangleComplementTwice$           \\ \hline
        $10$                    & $\ipsixCliqueThreeD$                                                  \\ \hline
        $12$                    & $\ipsixMatchingAndComplement$                                         \\ \hline
        $15$                    & $\ipsixCliqueTwo$, $\ipsixCliqueTwoTwice$                             \\ \hline
        $20$                    & $\ipsixCliqueThree$                                                   \\ \hline
    \end{tabular}
    \caption{Overview of~$\abs{\partition{L,S}}$ depending on the interspace pattern in~$\interspace{L}{S}$.}
    \label{interspace-pattern:partition-size/tab}
\end{table}

\begin{lemma}
\label{global-argument:partition:fully-intersecting/lem}
    Let~$\coherentConfig$ be a coherent configuration, and let~$(S_1, L, S_2)$ be a path in~$\quotientGraph{\coherentConfig}$.
    If~$\abs{\partition{L,S_1}}$ and~$\abs{\partition{L,S_2}}$ are coprime, then partitions~$\partition{L,S_1}$ and~$\partition{L, S_2}$ are \emph{fully intersecting}, that is, $P_1 \cap P_2 \neq \emptyset$ for all~$P_1 \in \partition{L,S_1}, P_2 \in \partition{L,S_2}$.
\end{lemma}
\begin{proof}
    Due to coherence all the non-trivial intersections of a part from~$\partition{L,S_1}$ with a part from~$\partition{L,S_2}$ have the same size, say~$q$. For each~$i\in\{1,2\}$ set~$p_i= |\partition{L,S_i}|$
    and set~$x_i$ to be the size of the parts in the equipartition~$\partition{L,S_i}$.
    We have~$|L|= x_1p_1=x_2p_2$.
    Each part of~$\partition{L,S_i}$ intersects~$x_i/q$ parts of~$\partition{L,S_{3-i}}$. We want to show that~$x_1/q= p_2$, so assume~$x_1/q<p_2$. We have~$x_1/q\cdot p_1 = x_2/q\cdot p_2$.
    But~$x_1/q< p_2$, which means~$p_1$ and~$p_2$ must have a common divisor.
\end{proof}

Let~$\mathcal{S}$ be a union of small fibers each adjacent to a large fiber~$L$ in the quotient graph of a coherent configuration~$\coherentConfig$.
For~$U,U' \in \interspace{L}{\mathcal{S}}$, we define a function~$\eta_{U,U'} \colon \equivalenceClasses{L,\mathcal{S}}^2 \longrightarrow \Nat$ with~$(P,P') \mapsto | p U \cap p' U'|$ where~$p \in P,p' \in P'$ are arbitrary.
We define~$\mathcal{Q}$ to be the set of equivalence classes on~$\equivalenceClasses{L,S}^2$ such that for all~$Q \in \mathcal{Q}$ we have
\[
    (P_1,P_2),(P'_1,P'_2) \in Q \text{ if and only if } \eta_{U,U'}(P_1,P_2) = \eta_{U,U'}(P'_1,P'_2) \text{ for all } U,U' \in \interspace{L}{\mathcal{S}}.
\]
We call the coherent configuration~$(\equivalenceClasses{L,\mathcal{S}}, \mathcal{Q})$ the~\emph{partition structure} and denoted it by~$\partitionStructure{L,\mathcal{S}}$.
Intuitively speaking, the vertices of the partition structure correspond to the parts of $\equivalenceClasses{L,\mathcal{S}}$. For the arcs between a pair of parts, the relations represent what the arcs in~$L$ ``know'' about the intersection of common neighborhoods of their end vertices in~$\mathcal{S}$ (with respect to basis relations in~$\interspace{L}{\mathcal{S}}$).
Observe that~$\partitionStructure{L,\mathcal{S}}$ is indeed a coherent configuration itself.


\begin{table}[tb]
    \centering\def\arraystretch{1.5}%
    \begin{tabular}{|l|l|l|}
        \hline
        \makecell{Interspace pat-\\tern of~$\interspace{L}{S}$} & $\type{\inducedCC{S}}$                                                                                                                                                                    & $\type{\partitionStructure{L,S}}$                                     \\ \hline
        $\ipfourClique$                                         & $(K_4)$                                                                                                                                                                                   & $(3K_2,K_{2,2,2})$                                                    \\ \hline
        $\ipfourMatching$                                       & $(2K_2,2K_2,2K_2)$, $(C_4,2K_2)$, $(\overrightarrow{C_4},2K_2)$                                                                                                                           & $(K_2)$                                                              \\ \hline
        $\ipfourCycle$                                          & $(2K_2,C_4)$                                                                                                                                                                              & $(2K_2,C_4)$                                                          \\ \hline
        $\ipsixMatching$                                        & $(3K_2,K_{2,2})$, $(C_6,2C_3,3K_2)$,                                                                                                                                                      & $(K_3)$                                                              \\ \hline
        \multirow{2}{*}{$\ipsixMatchingTwice$}                  & $(3K_2,K_{2,2,2})$, $(C_6,2C_3,3K_2)$,                                                                                                                                                      & $(3K_2,3K_2,2\overrightarrow{C_3},3K_2)$                                                    \\ \cline{2-3}
                                                                & $(3K_2,\overrightarrow{C_3}[K_2])$, $(\overrightarrow{C_6},2\overrightarrow{C_3},3K_2)$                                                                                                   & $(\overrightarrow{C_3})$                                              \\ \hline
        $\ipsixMatchingAndCycle$                                & $(C_6,2C_3,3K_2)$                                                                                                                                                                         & $(C_6,2C_3,3K_2)$                                                     \\ \hline
        $\ipsixTriangleComplement    $                          & $(2K_3,K_{3,3})$                                                                                                                                                                          & $(\rookGraph{3},\rookGraph{3})$                                       \\ \hline
        \multirow{2}{*}{$\ipsixTriangle$}                       & $(2C_3,K_{3,3})$, $(2\overrightarrow{C_3},K_{3,3})$, $(C_6,2C_3,3K_2)$,                                                                                                                   & \multirow{2}{*}{$(K_2)$}                                             \\
                                                                & $(\overrightarrow{C_6},2\overrightarrow{C_3},3K_2)$, $(3K_2,3K_2,\overrightarrow{2C_3},3K_2)$                                                                                             &                                                                       \\ \hline
        $\ipsixMatchingMatching$                                & $(3K_2,3K_2,\overrightarrow{2C_3},3K_2)$                                                                                                                                                  & $(C_3)$                                                               \\ \hline
        $\ipsixMatchingComplement    $                          & $(3K_2,K_{2,2,2})$, $(3K_2,\overrightarrow{C_3}[K_2])$                                                                                                                                    & $(2K_4,4K_2,K_{4,4}-4K_2)$                                            \\ \hline
        $\ipsixMatchingComplementD    $                         & $(3K_2,K_{2,2,2})$, $(3K_2,\overrightarrow{C_3}[K_2])$                                                                                                                                    & $(K_4)$                                                               \\ \hline
    \end{tabular}
    \caption{Isomorphism type of the partition structure~$\partitionStructure{L,S}$ uniquely determined by the isomorphism type of~$\inducedCC{S}$ and the interspace pattern of~$\interspace{L}{S}$ in a critical coherent configuration~$\coherentConfig$.}
    \label{interspace-pattern:partition-structure/tab}
\end{table}

\begin{lemma}
\label{interspace-pattern:partition-structure/lem}
    Let~$\coherentConfig$ be a critical coherent configuration, and let~$L,S$ be fibers of~$\coherentConfig$.
    For~$|\partition{L,S}| \leq 8$, the coherent configuration~$\inducedCC{S}$ together with the interspace pattern of~$\interspace{L}{S}$ uniquely determines the partition structure~$\partitionStructure{L,S}$.
\end{lemma}
\begin{proof}
    Recall that the isomorphism type of~$\inducedCC{S}$ is determined by the set of constituents. For number of
    parts in~$\equivalenceClasses{L,S}$ is determined by the interspace pattern.
    On top of that, for all interspace pattern, the set~$\ell U^1_1$ determines all sets~$\ell U^i_j$.
    It follows that up to isomorphism there is a unique partition structure.
\end{proof}

We should comment that the interspace pattern and the isomorphism type of subconfiguration induced by the small fiber~$S$ does not uniquely determine the partition structure~$\partitionStructure{L,S}$ if~$|\partition{L,S}| \geq 9$.
Further, we remark that in a critical coherent configuration the interspace pattern might restrict the possible constituents in the small fiber, as we have shown in the proof of Theorem~\ref{large-small-interspace:classification:uniqueness/cor}:
For example, if an interspace~$\interspace{L}{S}$ has the interspace pattern~$\ipfourCycle$ or~$\ipsixMatchingAndCycle$, then by coherence~$\overrightarrow{C_4} \notin \inducedCC{S}$ (respectively~$\overrightarrow{C_6} \notin \inducedCC{S}$) since~$U^1_1(\interspace{L}{S})$ would not be a basis relation.
For similar reasons~$2\overrightarrow{C_3} \notin \inducedCC{S}$ if~$\interspace{L}{S}$ has the interspace pattern~$\ipsixTriangleComplement$.
For each interspace pattern and isomorphism type of configuration~$\inducedCC{S}$, the possible partition structures in a critical coherent configuration~$\coherentConfig$ are listed in Table~\ref{interspace-pattern:partition-structure/tab}.


\section{Restorable fibers}
\label{critical:restorable/sec}

Equipped with the classification of the interspaces between a small fiber and a fiber of arbitrary size, we present several claims that examine induced coherent subconfigurations with respect to criticality and restorability in particular.
With the exception of Lemmas~\ref{critical:adjacent-interspace-cycle/lem},~\ref{6-cc:implied-interspace:DUC-DUC/lem}, and~\ref{critical:7-cc:leviFano/lem}, the proofs of these claims follow a similar pattern:
we start with an induced coherent subconfigurations consisting of multiple fibers and interspaces with a small non-dominating fiber~$S$ (and in rare cases multiple non-dominating fibers) at its center.
The adjacency of~$S$ to other fibers is described by the interspace patterns of its incident interspaces in quotient graph~$\quotientGraph{\coherentConfig}$.
The homogeneous coherent configuration induced by~$S$ together with the interspace patterns determines the intersection of neighborhoods of vertices of possible different fibers, for example say~$L$, which are adjacent to~$S$ in~$\quotientGraph{\coherentConfig}$.
The~\emph{partition structures} $\partitionStructure{L,S}$ captures this information.
Furthermore, the automorphisms of~$\partitionStructure{L,S}$ are a subset of the automorphisms of~$\inducedCC{L}$, and each automorphism of~$\partitionStructure{L,S}$ is induced by an automorphism of~$\inducedCC{S}$.
Altogether, the automorphisms of~$\inducedCC{L}$ extends to automorphisms of~$\inducedCC{L \cup S}$.
Thus~$S$ is restorable, which contradicts Lemma~\ref{critical:restorable/lem}.

\begin{lemma}
    \label{critical:adjacent-interspace-cycle/lem}
    Let~$\coherentConfig$ be a critical coherent configuration, and let~$(R,S,Y)$ be a path in~$\quotientGraph{\coherentConfig}$ with~$S$ being small.
    \begin{enumerate}
        \item If both~$\interspace{R}{S}$ and~$\interspace{Y}{S}$ have the interspace pattern~$\ipfourCycle$, then there is a union~$\arcs$ of basis relations in~$\interspace{R}{Y}$ such that~$(R \disjointUnion Y,\arcs) \cong \disjointCliques{4}{\frac{\abs{R}}{4},\frac{\abs{Y}}{4}}$.
        \item If both~$\interspace{R}{S}$ and~$\interspace{Y}{S}$ have the interspace pattern~$\ipsixMatchingAndCycle$, then there is a union~$\arcs$ of basis relations in~$\interspace{R}{Y}$ such that $(R \disjointUnion Y,\arcs) \cong \disjointCliques{6}{\frac{\abs{R}}{6},\frac{\abs{Y}}{6}}$.
    \end{enumerate}
    In particular, fibers~$R$ and~$Y$ are large.
\end{lemma}
\begin{proof}
    By coherence, there is a unique basis relation~$\arcs \in \inducedCC{S}$ with~$\cycle{\abs{S}} \cong (S,\arcs)$.
    For~$U = U^1_1(\interspace{R}{S})$ and~$U' = U^1_1( \interspace{Y}{S})$, it holds that for all~$r \in R$ (respectively~$y \in Y$) the set~$r U$ (respectively~$y U'$) induces a~$2$-clique within~$(S,\arcs)$.
    Thus for all~$P \in \partition{R,S}$ there is a unique~$P' \in \partition{Y,S}$ such that for all~$v \in P, w \in P'$ we have~$r U = y U'$.
    Therefore, the interspace~$\interspace{R}{Y}$ contains a union of basis relations~$U''$ such that~$(R \disjointUnion Y, U'') \cong \disjointCliques{\abs{\arcs}}{\nicefrac{\abs{R}}{\abs{\arcs}},\nicefrac{\abs{Y}}{\abs{\arcs}}}$.

    In particular, if~$R$ or~$Y$ are small, then~$\interspace{R}{Y}$ contains a constituent isomorphic to a matching or star, which violates Lemma~\ref{critical:star/lem}.
\end{proof}

\begin{lemma}
    \label{critical:4cc:restorable:2,C4/lem}
    Let~$\coherentConfig$ be a critical coherent configuration, and let~$S \in \fibers{\coherentConfig}$ such that~$\abs{S} = 4$.
    If~$\{S\}$ is not dominating and~$|\ul(\inducedCC{S})| = 3$, then there are fibers~$R,R'$ adjacent to~$S$ in~$\quotientGraph{\coherentConfig}$ such that~$\interspace{R}{S}$ has the interspace pattern~$\ipfourCycle$ and~$\interspace{R'}{S}$ has the interspace pattern~$\ipfourMatching$.
\end{lemma}
\begin{proof}
    Towards a contradiction, suppose that for all fibers~$R$ adjacent to~$S$ in~$\quotientGraph{\coherentConfig}$ the interspace~$\interspace{R}{S}$ has the interspace pattern~$\ipfourCycle$.
    By coherence, there is a constituent in~$\inducedCC{S}$ isomorphic to~$\cycle{4}$.
    This constituent is unique in~$\inducedCC{S}$ due to Lemma~\ref{small-cc:induced-cc/lem}.
    Thus for all distinct fibers~$R,R' \in \fibers{\coherentConfig}$ adjacent to~$S$, the fiber~$R$ takes care of~$R'$ (with regard to restorability of~$S$).
    Together with Lemma~\ref{critical:restorable:take-care/lem}, we may assume that~$\colorDeg{S} = 1$, when we show in the following that~$S$ is restorable.
    Let~$R$ be a fiber adjacent to~$S$ in~$\quotientGraph{\coherentConfig}$.
    Lemma~\ref{interspace-pattern:partition-structure/lem} provides the partition structure of~$\interspace{R}{S}$, and its isomorphism type is determined by the isomorphism types of its constituents.
    Hence, we have~$\partitionStructure{L,S} \cong (C_4,2K_2)$.
    All automorphisms of~$\partitionStructure{L,S}$ are induced by an automorphism of~$\inducedCC{S}$.
    Thus each automorphism of~$\inducedCC{R}$ extends to an automorphism of~$\inducedCC{R \cup S}$.
    This contradicts to Lemma~\ref{critical:restorable/lem}.

    Suppose that, for all fibers~$R$ adjacent to~$S$ in~$\quotientGraph{\coherentConfig}$, the interspace~$\interspace{R}{S}$ has the interspace pattern~$\ipfourMatching$.
    Thus there is a constituent in~$\inducedCC{S}$ isomorphic to~$\disjointCliques{2}{2}$.
    By a reasoning similar to the one above, we may assume~$\colorDeg{S} = 1$.
    By Lemma~\ref{interspace-pattern:partition-structure/lem}, we have~$\partitionStructure{R,S} \cong K_2$.
    Since all automorphisms of~$\partitionStructure{L,S}$ are induced by an automorphism of~$\inducedCC{S}$, each automorphism of~$\inducedCC{R}$ extends to an automorphism of~$\inducedCC{R \cup S}$.
    Once again, a contradiction to Lemma~\ref{critical:restorable/lem} occurs.
\end{proof}

Let~$\coherentConfig_0$ and~$\coherentConfig_1$ be a coherent configuration.
We call~$(\vertices_0 \disjointUnion \vertices_1, \relations(\coherentConfig_0) \cup \relations(\coherentConfig_1) \cup \relations_{0,1} \cup \relations_{1,0})$ where
\[
    \relations_{0,1} = \{\Delta_0 \times \Delta_1 \mid (\Delta_0,\Delta_1) \in \fibers{\coherentConfig_0} \times \fibers{\coherentConfig_1} \}
\]
and~$\relations_{1,0} = \relations^\star_{0,1}$ the~\emph{direct sum} of~$\coherentConfig_0$ and~$\coherentConfig_1$ and denote it by~$\coherentConfig_0 \boxplus \coherentConfig_1$.
By~\cite{CC}, if~$\coherentConfig_0$ and~$\coherentConfig_1$ are a coherent configuration, then~$\coherentConfig_0 \boxplus \coherentConfig_1$ is a coherent configuration as well.
In the terms of colored graphs, the direct sum of two configurations is their color-disjoint union.
Note that for all~$i \in \{1,2\}$ we have~$\vertices(\coherentConfig_i) \in \fibers{\coherentConfig_1 \boxplus \coherentConfig_2}^\cup$.
Thus there is no (color-preserving) automorphism of~$\coherentConfig_1 \boxplus \coherentConfig_2$ which maps~$\vertices(\coherentConfig_1)$ to~$\vertices(\coherentConfig_2)$ and vice versa.

\begin{lemma}
\label{critical:4cc:restorable:cycle/lem}
    Let~$\coherentConfig$ be a critical coherent configuration, and let~$R_0,R_1 \in \fibers{\coherentConfig}$ be distinct.
    If~$\{R_0 , R_1\}$ is not dominating, then~$\cycle{8} \notin \interspace{R_0}{R_1}$.
\end{lemma}
\begin{proof}
    Towards a contradiction, suppose first that $\cycle{8} \in \interspace{R_0}{R_1}$.
    By assumption, neither~$\{R_0\}$ nor~$\{R_1\}$ are dominating.
    Due to Lemma~\ref{critical:4cc:restorable:2,C4/lem}, both~$R_0$ and~$R_1$ must have quotient degree at least~$2$.
    Recall by combining Lemmas~\ref{small-cc:induced-cc/lem} and~\ref{small-cc:interspace-implies-cc/lem}, there are two unique constituents in~$\inducedCC{R_0}$ isomorphic to~$\cycle{4}$ and~$\disjointCliques{2}{2}$, respectively.
    By Lemma~\ref{critical:adjacent-interspace-cycle/lem}, every non-homogeneous interspace between~$R_0$ and a fiber other than~$R_1$ has the interspace pattern~$\ipfourMatching$.
    Therefore, there is a fiber~$B_0$ adjacent to~$R_0$ which takes care (regarding restorability) of all other fibers neighboring~$R_0$ in~$\quotientGraph{\coherentConfig}$.
    So, by Lemma~\ref{critical:restorable:take-care/lem}, we may assume~$\colorDeg{R_0} = 2$.
    The same reasoning applies to~$R_1$ as well.

    Let~$(B_0,R_0,R_1,B_1)$ be a possibly closed path in~$\quotientGraph{\coherentConfig}$ such that~$\interspace{B_i}{R_i}$ has the interspace pattern~$\ipfourMatching$.
    By Lemma~\ref{interspace-pattern:partition-structure/lem}, both partitions structures~$\partitionStructure{B_0,R_0}$ and~$\partitionStructure{B_1,R_1}$ are isomorphic to~$K_2$.
    If~$B_0 \neq B_1$, then we set~$\mathfrak{S} \coloneqq \partitionStructure{B_0,R_0} \boxplus \partitionStructure{B_1,R_1}$.
    If~$B_0 = B_1$ and~$\equivalenceClasses{B_0,R_0} = \equivalenceClasses{B_1,R_1}$, then~$\mathfrak{S} \coloneqq \partitionStructure{B_0,R_0 \cup R_1}$ and~$\mathfrak{S} \cong K_2$.
    Finally, if~$B_0 = B_1$ and~$\equivalenceClasses{B_0,R_0} \neq \equivalenceClasses{B_1,R_1}$, then we set~$\mathfrak{S} \coloneqq \partitionStructure{B_0,R_0 \cup R_1}$, $\equivalenceClasses{B_0,R_0}$ and~$\equivalenceClasses{B_1,R_1}$ are fully intersecting, and the isomorphism type of~$\mathfrak{S}$ is determined by the isomorphism types of its constituents~$(2K_2,2K_2,2K_2)$.
    All automorphisms of~$\mathfrak{S}$ are induced by an automorphism of~$\inducedCC{R_0 \cup R_1}$.
    Thus each automorphism of~$\inducedCC{ B_0 \cup B_1 }$ extends to an automorphism of~$\inducedCC{ B_0 \cup R_0 \cup R_1 \cup B_1 }$. Thus $R_0 \cup R_1$ is restorable.
    Since~$\{R_0, R_1\}$ is not dominating, this contradicts Lemma~\ref{critical:restorable/lem}.
\end{proof}

\begin{lemma}
\label{critical:4-cc:restorable:DUC/lem}
    Let~$\coherentConfig$ be a critical coherent configuration, and let~$S \in \fibers{\coherentConfig}$ such that~$\{S\}$ is not dominating.
    For every~$\arcs \in \inducedCC{S}$ with~$(S,\arcs) \cong \disjointCliques{2}{2}$ there are~$R \in \fibers{\coherentConfig}$ and~$U \in \interspace{R}{S}$ such that for all~$r \in R$ the set~$r U$ induces a~$2$-clique in~$(S,\arcs)$.

    In particular, if~$\abs{S} = 4$ and~$\ul(\inducedCC{S})| = 4$, then~$\colorDeg{S} \geq 3$.
\end{lemma}
\begin{proof}
    First assume that there are two constituents in~$\inducedCC{S}$ whose underlying graphs are isomorphic~$\cycle{4}$ and~$2K_2$ respectively.
    Thus~$|\ul(\inducedCC{S})| = 3$ and by Lemma~\ref{critical:4cc:restorable:2,C4/lem} the claim holds.

    Next assume that~$\inducedCC{S}$ does not contain a constituent whose underlying graph is isomorphic to~$\cycle{4}$.
    Thus, by Lemma~\ref{small-cc:induced-cc/lem}, we may assume that there are three constituents~$G_0,G_1,G_2$ in~$\inducedCC{S}$ each of which is isomorphic to~$\disjointCliques{2}{2}$.
    Suppose~$J \subseteq \{0,1,2\}$ is the largest set such that for each~$j \in J$ there is a fiber~$R_j \in \fibers{\coherentConfig}$ and a basis relation~$U_j \in \interspace{R_j}{S}$ such that for all~$r_j \in R_j$ the set~$r_j U$ induces a~$2$-clique in~$G_j$.

    Let~$R,R'$ be fibers adjacent to~$S$ in~$\quotientGraph{\coherentConfig}$,~$U \in \interspace{R}{S}$, and~$U' \in \interspace{R'}{S}$.
    If there is an~$i \in J$ such that for all~$r \in R$ and~$r' \in R'$ the set~$rU$ (respectively~$r'U'$) induces a~$2$-clique in~$G_i$, then~$R$ takes care of~$R'$  with regard to restorability of~$S$.
    Hence, by Lemma~\ref{critical:restorable:take-care/lem}, we may assume that~$R_j$ is unique with respect to~$G_j$ when we show in the following that~$S$ is restorable if~$|J| < 3$.

    Suppose~that~$|J| < 3$.
    For all~$j \in J$, Lemma~\ref{interspace-pattern:partition-structure/lem} determines~$\partitionStructure{R_j,S}$, all of which are isomorphic to~$K_2$.
    Define~$\mathfrak{S} \coloneqq \boxplus_{j \in J} \partitionStructure{R_j,S}$.
    All automorphisms of~$\mathfrak{S}$ are induced by an automorphism of~$\inducedCC{S}$.
    Thus each automorphism of~$\inducedCC{ \bigcup_{j \in J} R_j}$ extends to an automorphism of~$\inducedCC{S \cup \bigcup_{j \in J} R_j}$.
    Overall, since~$S$ is not dominating but restorable, this contradicts Lemma~\ref{critical:restorable/lem}.
\end{proof}

\begin{lemma}
\label{6-cc:implied-interspace:DUC-DUC/lem}
    Let~$\coherentConfig$ be a coherent configuration, and let~$(R,B,Y)$ be path in~$\quotientGraph{\coherentConfig}$ with~$\abs{B} = 6$.
    \begin{enumerate}
        \item
        If~$\disjointCliques{3}{\frac{\abs{R}}{3},2} \in \interspace{R}{B}$ and~$\disjointCliques{3}{\frac{\abs{Y}}{3},2} \in \interspace{Y}{B}$,
        then there is a union~$U$ of basis relations in~$\interspace{R}{Y}$ such that~$(R \disjointUnion Y, U) \cong \disjointCliques{3}{\frac{\abs{R}}{3},\frac{\abs{Y}}{3}}$.

        \item
        If both~$\interspace{R}{B}$ and~$\interspace{Y}{B}$ have interspace pattern~$\ipsixTriangle$,
        then there is a union~$U$ of basis relations in~$\interspace{R}{Y}$ such that~$(R \disjointUnion Y, U) \cong \disjointCliques{2}{\frac{\abs{R}}{2},\frac{\abs{Y}}{2}}$.
    \end{enumerate}
\end{lemma}
\begin{proof}
    Let~$x \in \{2,3\}$.
    Let~$U_B \in \interspace{R}{B}$ and~$U_Y \in \interspace{Y}{B}$ such that~$(R \disjointUnion B,U_B) \cong \disjointCliques{x}{\nicefrac{\abs{R}}{x},\nicefrac{6}{x}}$ and~$(Y \disjointUnion B,U_Y) \cong \disjointCliques{x}{\nicefrac{\abs{Y}}{x},\nicefrac{6}{x}}$ respectively.
    Choose~$r\in R$ and~$y\in Y$ so that~$z \coloneqq \abs{rU_B \cap yU_Y}>0$.
    Observe that~$\abs{r U_B} = \nicefrac{6}{x}$.

    If~$z = \frac{6}{x}$, then for all~$y' \in Y,r' \in R$ either~$y'U_Y = r'U_B$ or~$y'U_Y = B \setminus r'U_B$.
    Thus, by coherence, there is a union of basis relations~$U$ such that~$(R \disjointUnion Y, U) \cong \disjointCliques{x}{\nicefrac{\abs{R}}{x},\nicefrac{\abs{Y}}{x}}$.

    Now assume~$x = 2$ and~$z = 2$.
    There is a constituent~$G$ in~$\ul(\inducedCC{B})$ isomorphic to~$\disjointCliques{2}{3}$, which, by Lemma~\ref{small-cc:induced-cc/lem}, is unique within~$\inducedCC{B}$.
    Let~$B'$ be a connected component of~$G$.
    There is a fiber~$F \in \{R,Y\}$ and a vertex~$f \in F$ such that~$fU_F \neq B'$ and~$fU_F \neq B \setminus B'$.
    However, the graph~$G - fU_F$ is not regular, which contradicts coherence.

    The case where~$x = 2$ and~$z=1$ works in a similar fashion.

    Finally, suppose~$x = 3$ and~$z = 1$.
    Let~$\mathcal{R}$ (respectively~$\mathcal{Y}$) be the twin classes of~$R$ (respectively~$Y$) with respect to~$B$ and~$U_R$ (respectively~$U_Y$).
    Note that $\mathcal{R}$ and~$\mathcal{Y}$ are both equipartitions due to coherence.
    For all~$R' \in \mathcal{R}$ there is exactly one~$Y' \in\mathcal{Y}$ such that the following holds:
    for all~$r' \in R', y' \in Y'$ we have~$y'U_Y \cap r'U_R = \emptyset$ and for all~$r' \in R', y' \in Y \setminus Y'$ we have~$\abs{y'U_Y \cap r'U_R} = 1$.
    Thus each part in~$\mathcal{R}$ has a unique corresponding part in~$\mathcal{Y}$.
    Hence there is a union~$U$ of basis relation in~$\interspace{R}{Y}$ such that~$(R \disjointUnion Y, U) \cong \disjointCliques{3}{\nicefrac{\abs{R}}{3},\nicefrac{\abs{Y}}{3}}$.
\end{proof}

For the convenience of the reader, Table~\ref{6-cc:implied-interspace:DUC-DUC/tab} gives an overview of the results of Lemmas~\ref{critical:adjacent-interspace-cycle/lem} and~\ref{6-cc:implied-interspace:DUC-DUC/lem} for size~$6$ fibers.


\begin{table}[tbp]
    \centering\def\arraystretch{1.2}%
    \begin{tabular}{|ll|lll|}
        \hline
        \multicolumn{2}{|l|}{\multirow{2}{*}{}}                   & \multicolumn{3}{c|}{$\interspace{B}{Y}$}                                                         \\
        \multicolumn{2}{|l|}{}                                    & \multicolumn{1}{l|}{\cycle{12}} & \multicolumn{1}{l|}{\disjointCliques{3}{2,2}} &~$\disjointCliques{2}{3,3}$ \\ \hline
        \multicolumn{1}{|l}{\multirow{3}{*}{$\interspace{R}{B}$}} &~$\cycle{12}$                    & \multicolumn{1}{l|}{\matching{6}}             & \multicolumn{1}{l|}{\disjointCliques{3}{2,2}}   & $R \times B$ \\ \cline{2-5}
        \multicolumn{1}{|l}{}                                     &~$\disjointCliques{3}{2,2}$      & \multicolumn{1}{l|}{\disjointCliques{3}{2,2}} & \multicolumn{1}{l|}{$\disjointCliques{3}{2,2}$} & $R \times B$ \\ \cline{2-5}
        \multicolumn{1}{|l}{}                                     &~$\disjointCliques{2}{3,3}$      & \multicolumn{1}{l|}{$R \times B$}                  & \multicolumn{1}{l|}{$R \times B$}                    &~$\disjointCliques{2}{3,3}$ \\ \hline
    \end{tabular}
    \caption{For a path~$(R,B,Y)$ of size~$6$ fibers, the isomorphism type of a constituent in~$\interspace{R}{Y}$ depending on the isomorphism types of the constituents contained in~$\interspace{R}{B}$ and~$\interspace{B}{Y}$.}
    \label{6-cc:implied-interspace:DUC-DUC/tab}
\end{table}

\begin{lemma}
\label{critical:6-cc:restorable:DUC:deg1/lem}
    Let~$\coherentConfig$ be a critical coherent configuration, and let~$S \in \fibers{\coherentConfig}$ such that~$\abs{S} = 6$.
    If~$\{S\}$ is not dominating, then there is no interspace pattern~$\mathfrak{P}$ in the following list such that for all fibers~$R$ adjacent to~$S$ in~$\quotientGraph{\coherentConfig}$ the interspace~$\interspace{R}{S}$ has this interspace pattern~$\mathfrak{P}$:
    \begin{multicols}{3}
        \begin{enumerate}
            \item $\ipsixTriangle$
            \item $\ipsixMatching$
            \item $\ipsixMatchingTwice$
            \item $\ipsixMatchingAndCycle$
            \item $\ipsixMatchingMatching$
            \item $\ipsixTriangleComplement$
        \end{enumerate}
    \end{multicols}
    In particular, if~$\{S\}$ is not dominating, then there is at least one fiber~$R$ adjacent to~$S$ in~$\quotientGraph{\coherentConfig}$ such that~$\interspace{R}{S}$ has neither interspace pattern~$\ipsixMatching$ nor~$\ipsixMatchingTwice$.
\end{lemma}
\begin{proof}
    We will show that~$S$ is restorable if all interspace incident to~$S$ have the same interspace pattern chosen from the list.
    Since by assumption fiber~$S$ is not dominating, this contradicts Lemma~\ref{critical:restorable/lem}.

    Before we start, we first argue that by Lemma~\ref{critical:restorable:take-care/lem} we may assume~$\colorDeg{S} = 1$.
    Let~$R$ be a fiber adjacent to~$S$ in~$\quotientGraph{\coherentConfig}$ and assume that~$\interspace{R}{S}$ has the interspace pattern~$\mathfrak{P}$ where~$\mathfrak{P}$ is a pattern of the list above.
    If~$\mathfrak{P}$ is~$\ipsixMatchingAndCycle$ or~$\ipsixTriangleComplement$, then by coherence there is a constituent isomorphic to~$\cycle{6}$ or~$\clique{3,3}$ respectively.
    By Lemma~\ref{small-cc:induced-cc/lem} this constituent is unique within~$S$.
    If~$\mathfrak{P}$ is~$\ipsixTriangle$, then there is a constituent isomorphic to~$\disjointCliques{2}{3}$ or~$2\overrightarrow{C_3}\clique{3,3}$.
    By Lemma~\ref{small-cc:induced-cc/lem} this constituent is unique in~$\ul(\inducedCC{S})$.
    Assume that there are at least two constituents~$G,G'$ in~$\inducedCC{S}$ isomorphic to~$\disjointCliques{3}{2}$, and let~$r \in R$.
    Suppose that~$\mathfrak{P}$ is~$\ipsixMatching$.
    Let~$U$ be a basis relation of~$\interspace{R}{S}$ other than~$U^1_1(\interspace{R}{S})$.
    If~$A(G) \neq A^1(\interspace{R}{S})$, then there are exactly two vertices in~$rU$ which are adjacent in~$G$.
    This implies that~$U$ cannot be a basis relation.
    Suppose that~$\mathfrak{P}$ is~$\ipsixMatchingTwice$.
    Let~$U,U'$ be distinct basis relations of~$\interspace{R}{S}$ other than~$U^1_1(\interspace{R}{S})$.
    If~$A(G) \neq A^1(\interspace{R}{S})$, then there is exactly one vertex of~$rU$ which is adjacent to a vertex of~$rU'$ in~$G$.
    This implies that~$U$ cannot be a basis relation, giving a contradiction.
    Thus, if~$\mathfrak{P}$ is~$\ipsixMatching$ or~$\ipsixMatchingTwice$, then there is exactly one constituent in~$\inducedCC{S}$ isomorphic to~$\disjointCliques{3}{2}$.
    Finally, if~$\mathfrak{P}$ is~$\ipsixMatchingMatching$, then there are exactly three constituents in~$\inducedCC{S}$ each isomorphic to~$\disjointCliques{3}{2}$ by Lemma~\ref{small-cc:induced-cc/lem}.
    In each case, we conclude due to the uniqueness of the constituents in~$\inducedCC{S}$: if~$R'$ is another fiber adjacent to~$S$ other than~$R$ such that~$\interspace{R'}{S}$ has the interspace pattern~$\mathfrak{P}$, then~$R$ takes care of~$R'$ with regard to restorability of~$S$.
    By Lemma~\ref{critical:restorable:take-care/lem}, we may therefore assume that~$\colorDeg{S} = 1$.

    Let~$R$ be the only fiber adjacent to~$S$ in~$\quotientGraph{\coherentConfig}$.
    For each case, Lemma~\ref{interspace-pattern:partition-structure/lem} determines the partition structure of~$\interspace{R}{S}$ as follows.
    If~$\mathfrak{P}$ is~$\ipsixTriangle$, then~$\partitionStructure{L,S} \cong K_2$.
    If~$\mathfrak{P}$ is~$\ipsixMatching$, then~$\partitionStructure{L,S} \cong K_3$.
    If~$\mathfrak{P}$ is~$\ipsixMatchingTwice$, then the isomorphism type of~$\partitionStructure{L,S}$ is determined by the isomorphism types of its constituents~$(3K_2,3K_2,2\overrightarrow{C_3},3K_2)$ or~$(\overrightarrow{C_3})$.
    If~$\mathfrak{P}$ is~$\ipsixMatchingAndCycle$, then the isomorphism type of~$\partitionStructure{L,S}$ is determined by the isomorphism types of its constituents~$(C_6,2C_3,3K_2)$.
    If~$\mathfrak{P}$ is~$\ipsixMatchingMatching$, then~$\partitionStructure{L,S} \cong C_3$.
    If~$\mathfrak{P}$ is~$\ipsixTriangleComplement$, then the isomorphism type of~$\partitionStructure{L,S}$ is determined by the isomorphism types of its constituents~$(\rookGraph{3},\rookGraph{3})$.
    In all cases, all automorphisms of~$\partitionStructure{L,S}$ are induced by an automorphism of~$\inducedCC{S}$.
    Thus each automorphism of~$\inducedCC{R}$ extends to an automorphism of~$\inducedCC{R \cup S}$.

    Finally, let~$R,R' \in \fibers{\coherentConfig}$ such that~$\interspace{R}{S}$ has the interspace pattern~$\ipsixMatchingTwice$ and~$\interspace{R'}{S}$ has the interspace pattern~$\ipsixMatching$.
    Then~$R$ is taken care of by~$R'$.
    By Lemma~\ref{critical:restorable:take-care/lem} and the previous reasoning, fiber~$S$ is restorable if for all fibers~$R$ adjacent to~$S$ in~$\quotientGraph{\coherentConfig}$ the interspace~$\interspace{R}{S}$ has the interspace pattern~$\ipsixMatching$ or~$\ipsixMatchingTwice$.
    Since~$\{S\}$ is not dominating, this contradicts Lemma~\ref{critical:restorable/lem}.
\end{proof}

\begin{figure}[tbp]
    \centering
    \begin{subfigure}{.6\textwidth}
        \centering

\begin{tikzpicture}[scale=0.8]
    \begin{scope}[shift = {(0,0)}]
        \draw[black] (-1.5,-1.25) rectangle ++(3,2.5);
        \node at (-1.25,-0.95) {$S$};
        \foreach \x in {0,...,5}
        {
            \node[vertex,lightgray,fill=lightgray] (p\x) at (\x*60+60:1) {};
        }
        \draw[edge,red] (p0) -- (p1);
        \draw[edge,red] (p1) -- (p2);
        \draw[edge,red] (p2) -- (p3);
        \draw[edge,red] (p3) -- (p4);
        \draw[edge,red] (p4) -- (p5);
        \draw[edge,red] (p5) -- (p0);

        \draw[edge,blue] (p0) -- (p2);
        \draw[edge,blue] (p2) -- (p4);
        \draw[edge,blue] (p4) -- (p0);
        \draw[edge,blue] (p1) -- (p3);
        \draw[edge,blue] (p3) -- (p5);
        \draw[edge,blue] (p5) -- (p1);

        \draw[edge,darkyellow] (p0) -- (p3);
        \draw[edge,darkyellow] (p1) -- (p4);
        \draw[edge,darkyellow] (p2) -- (p5);
    \end{scope}

    \begin{scope}[shift = {(-3.5,-4)}]
        \filldraw[draw=black,fill=lightred] (-1.6,-2) rectangle ++(3.2,3);
        \node at (0,-1) {$R$};
        \node at (0,0) {$\ipsixMatchingAndCycle$};
        \coordinate (largeRed) at (0,1) {};
    \end{scope}

    \begin{scope}[shift = {(0,-4)}]
        \filldraw[draw=black,fill=lightblue]   (-1.6,-2) rectangle ++(3.2,3);
        \node at (0,-1) {$B$};
        \node at (0,0) {$\ipsixTriangle$};
        \coordinate (largeBlue) at (0,1) {};
    \end{scope}

    \begin{scope}[shift = {(3.5,-4)}]
        \filldraw[draw=black,fill=lightyellow] (-1.6,-2) rectangle ++(3.2,3);
        \node at (0,-1) {$Y$};
        \node at (0,0.25) {$\ipsixMatching$};
        \node at (0,-0.25) {$\ipsixMatchingTwice$};
        \coordinate (largeYellow) at (0,1) {};
    \end{scope}
    \draw[edge,black,thick] (largeBlue) -- (0,-1.25);
    \draw[edge,black,thick] (largeRed) -- (-1,-1.25);
    \draw[edge,black,thick] (largeYellow) -- (1,-1.25);

\end{tikzpicture}
        \subcaption{$|\ul(\inducedCC{S})| = 4$.}
        \label{restorable:6cc:monochormatic-cycle/fig}
    \end{subfigure}
    \hfil
    \begin{subfigure}{.39\textwidth}
        \centering

\begin{tikzpicture}[scale=0.8]

    \begin{scope}[shift = {(0,0)}]
        \draw[black] (-1.5,-1.25) rectangle ++(3,2.5);
        \node at (-1.25,-0.95) {$S$};
        \foreach \x in {0,...,5}
        {
            \node[vertex,lightgray,fill=lightgray] (p\x) at (\x*60+60:1) {};
        }
        \draw[edge,darkred] (p0) -- (p1);
        \draw[edge,red] (p1) -- (p2);
        \draw[edge,darkred] (p2) -- (p3);
        \draw[edge,red] (p3) -- (p4);
        \draw[edge,darkred] (p4) -- (p5);
        \draw[edge,red] (p5) -- (p0);

        \draw[arrow,thick,blue] (p0) -- (p2);
        \draw[arrow,thick,blue] (p2) -- (p4);
        \draw[arrow,thick,blue] (p4) -- (p0);
        \draw[arrow,thick,blue] (p1) -- (p5);
        \draw[arrow,thick,blue] (p5) -- (p3);
        \draw[arrow,thick,blue] (p3) -- (p1);

        \draw[edge,darkyellow] (p0) -- (p3);
        \draw[edge,darkyellow] (p1) -- (p4);
        \draw[edge,darkyellow] (p2) -- (p5);
    \end{scope}

    \begin{scope}[shift = {(-1.75,-4)}]
        \filldraw[draw=black,fill=lightred] (-1.6,-2) rectangle ++(3.2,3);
        \node at (0,-1) {$R$};
        \node at (0,0) {$\ipsixMatchingMatching$};
        \coordinate (largeRed) at (0,1) {};
    \end{scope}

    \begin{scope}[shift = {(1.75,-4)}]
        \filldraw[draw=black,fill=lightblue]   (-1.6,-2) rectangle ++(3.2,3);
        \node at (0,-1) {$B$};
        \node at (0,0) {$\ipsixTriangle$};
        \coordinate (largeBlue) at (0,1) {};
    \end{scope}

    \draw[edge,black,thick] (largeRed) -- (-0.5,-1.25);
    \draw[edge,black,thick] (largeBlue) -- (0.5,-1.25);

\end{tikzpicture}
        \subcaption{$|\ul(\inducedCC{S})| = 5$.}
        \label{restorable:6cc:alternating-cycle/fig}
    \end{subfigure}
    \caption[Visualisation of the partition of the neighborhood.]{Visualisation of the partition of the neighborhood in Lemma~\ref{critical:6-cc:restorable:large-neighborhood/lem}.}
    \label{restorable:6cc/fig}
\end{figure}

\begin{lemma}
\label{critical:6-cc:restorable:large-neighborhood/lem}
    Let~$\coherentConfig$ be a critical coherent configuration.
    Let~$S \in \fibers{\coherentConfig}$ such that~$\abs{S} = 6$, $\abs{\ul(\inducedCC{S})} > 3$, and~$\{S\}$ is not dominating.
    There exist disjoint subsets~$\mathcal{R},\mathcal{B}, \mathcal{Y}$ partitioning the neighborhood of~$S$ in~$\quotientGraph{\coherentConfig}$ with~$\mathcal{R},\mathcal{B}$ non-empty (and~$\mathcal{Y}$ possibly empty) such that
    \begin{enumerate}
        \item for all~$R \in \mathcal{R}$ the interspace~$\interspace{R}{S}$ has either the interspace pattern~$\ipsixMatchingAndCycle$ or the interspace pattern~$\ipsixMatchingMatching$,
        \item for all~$B \in \mathcal{B}$ the interspace~$\interspace{B}{S}$ has the interspace pattern~$\ipsixTriangle$,
        \item for all~$Y \in \mathcal{Y}$ the interspace~$\interspace{Y}{S}$ has either the interspace pattern~$\ipsixMatching$ or the interspace pattern~$\ipsixMatchingTwice$.
    \end{enumerate}
    In particular, $\colorDeg{S} \geq 2$. (See Figure~\ref{restorable:6cc/fig})
\end{lemma}
\begin{proof}
    Let~$\mathcal{N}$ be the neighborhood of the fiber~$S$ in~$\quotientGraph{\coherentConfig}$.
    Due to Lemma~\ref{small-cc:induced-cc/lem}, Corollary~\ref{large-small-interspace:classification:uniqueness/cor}, and~$\abs{\ul(\inducedCC{S})} > 3$, for all~$R \in \mathcal{N}$ the interspace~$\interspace{R}{S}$ has one of the following interspace patterns:
    $\ipsixMatchingAndCycle$, $\ipsixMatchingMatching$, $\ipsixMatching$, $\ipsixMatchingTwice$, or~$\ipsixTriangle$.
    Thus~$\mathcal{N}$ partitions into~$\{\mathcal{R},\mathcal{B}, \mathcal{Y}\}$ as defined in the statement.
    Before we prove that~$S$ is restorable if~$\mathcal{R}$ or~$\mathcal{B}$ are empty, we first argue that by Lemma~\ref{critical:restorable:take-care/lem} we may assume the each part of~$\{\mathcal{R},\mathcal{B}, \mathcal{Y}\}$ has at most size~$1$.

    Assume that~$|\ul(\inducedCC{S})| = 4$.
    For all fibers~$R \in \mathcal{R}$ the interspace~$\interspace{R}{S}$ has the interspace pattern~$\ipsixMatchingAndCycle$.
    By coherence, there are the constituents in~$\inducedCC{S}$ isomorphic to~$\cycle{6}$, $\disjointCliques{3}{2}$, or~$\disjointCliques{2}{3}$, which by Lemma~\ref{small-cc:induced-cc/lem} are unique in~$S$.
    Hence for all~$\mathcal{F} \in \{\mathcal{R},\mathcal{B}, \mathcal{Y}\}$ the following holds:
    for all distinct~$F,F' \in \mathcal{F}$ the fiber~$F$ takes care of~$F'$ with regard to restorability of~$S$.
    Thus by Lemma~\ref{critical:restorable:take-care/lem} we may assume~$|\mathcal{F}| \leq 1$.
    Furthermore, if~$R \in \mathcal{R}$ exists, then~$R$ takes care of all~$Y \in \mathcal{Y}$, and thus by Lemma~\ref{critical:restorable:take-care/lem} we may assume~$\mathcal{Y} = \emptyset$ if~$|\mathcal{R}| = 1$.

    Assume that~$|\ul(\inducedCC{S})| = 5$.
    For all fibers~$R \in \mathcal{R}$ the interspace~$\interspace{R}{S}$ has the interspace pattern~$\ipsixMatchingMatching$, and hence~$\mathcal{Y} = \emptyset$ due to coherence.
    In~$\ul(\inducedCC{S})$ the constituent isomorphic to~$\disjointCliques{2}{3}$ is unique, and thus by Lemma~\ref{critical:restorable:take-care/lem} we may assume~$|\mathcal{B}| \leq 1$.
    Further, since for all constituents~$G$ isomorphic to~$\disjointCliques{3}{2}$ there is a basis relation~$U \in \interspace{R}{S}$ such that for all~$r \in R$ the set~$rU$ induces a~$2$-clique in~$G$, for any distinct~$R,R' \in \mathcal{R}$ the fiber~$R$ takes care of~$R'$ with regard to restorability of~$S$.
    Thus we may assume~$|\mathcal{R}| \leq 1$ because of Lemma~\ref{critical:restorable:take-care/lem}.

    In the following, we show that~$S$ is restorable if~$\mathcal{R}$ or~$\mathcal{B}$ are empty.
    Assume that there is exactly one non-empty~$\mathcal{F} \in \{\mathcal{R},\mathcal{B}, \mathcal{Y}\}$.
    Hence, by the reasoning above, we may assume that there is only one fiber adjacent to~$S$ in~$\quotientGraph{\coherentConfig}$ which is not taken care of by some other fiber.
    Since Lemma~\ref{critical:6-cc:restorable:DUC:deg1/lem} implies that there are at least two incident interspace having different interspace patterns, a contradiction arises.
    Together with the reasoning above, in the last remaining case, the part~$\mathcal{R}$ is empty while~$|\mathcal{B}| = |\mathcal{Y}| = 1$.

    Let~$(Y,S,B)$ be a path in~$\quotientGraph{\coherentConfig}$ such that~$Y \in \mathcal{Y}$ and~$B \in \mathcal{B}$.
    Lemma~\ref{interspace-pattern:partition-structure/lem} determines~$\partitionStructure{B,S}$ and~$\partitionStructure{Y,S}$ as follows.
    While~$\partitionStructure{B,S}$ is isomorphic to~$K_2$, the isomorphism type of the partition structure~$\partitionStructure{Y,S}$ is determined by isomorphism types of its constituents~$(3K_2,3K_2,2\overrightarrow{C_3},3K_2)$ or~$(K_3)$ or~$(\overrightarrow{C_3})$.
    We define~$\mathfrak{S} \coloneqq \partitionStructure{B,S} \boxplus \partitionStructure{Y,S}$.
    All automorphisms of~$\mathfrak{S}$ are induced by an automorphism of~$\inducedCC{S}$.
    Thus all automorphisms of~$\inducedCC{B \cup Y}$ extend to an automorphism of~$\inducedCC{B \cup Y \cup S}$.
\end{proof}

\begin{lemma}
    \label{critical:6-cc:restorable:cycle/lem}
    Let~$\coherentConfig$ be a critical coherent configuration, and let~$R,S$ be distinct fibers of~$\coherentConfig$ such that~$\abs{R} = \abs{S} =6$, $\cycle{12} \in \interspace{R}{S}$, and~$\{R,S\}$ is not dominating.
    There exist disjoint, non-empty subsets~$\mathcal{B}, \mathcal{Y}$ of the neighborhood of~$S$ without~$R$ in~$\quotientGraph{\coherentConfig}$ such that
    \begin{enumerate}
        \item for all~$B \in \mathcal{B}$ the interspace~$\interspace{B}{S}$ has the interspace pattern~$\ipsixTriangle$ and
        \item for all~$Y \in \mathcal{Y}$ the interspace~$\interspace{Y}{S}$ has either the interspace pattern~$\ipsixMatching$ or the interspace pattern~$\ipsixMatchingTwice$.
    \end{enumerate}
    In particular, $\colorDeg{S} \geq 3$.
\end{lemma}
\begin{proof}
    Lemma~\ref{critical:adjacent-interspace-cycle/lem} implies the uniqueness of~$R$ with respect to~$S$ and vice versa.
    Since the set~$\{R,S\}$ is not dominating, Lemma~\ref{critical:6-cc:restorable:large-neighborhood/lem} applies to both~$R$ and~$S$:
    hence for both~$R$ and~$S$ there is a fiber~$B$ such that~$\interspace{B}{S}$ (respectively~$\interspace{B}{R}$) has the interspace pattern~$\ipsixTriangle$.
    Furthermore, we only need to consider~$\mathcal{Y} = \emptyset$.
    Recall that by Lemma~\ref{6-cc:implied-interspace:DUC-DUC/lem} for every fiber~$Y$ for which~$\interspace{Y}{R}$ has the interspace pattern~$\ipsixMatching$ or~$\ipsixMatchingTwice$ the interspace~$\interspace{Y}{S}$ has the interspace pattern~$\ipsixMatching$ or~$\ipsixMatchingTwice$ as well.
    Therefore, we only need to argue restorability in the following case, which will then contradict Lemma~\ref{critical:restorable/lem}:

    Let~$(B,R,S,B')$ be a possibly closed path in~$\quotientGraph{\coherentConfig}$ such that~$\cycle{12} \in \interspace{R}{S}$ and both~$\interspace{B}{R}$ and~$\interspace{B'}{S}$ have the interspace pattern~$\ipsixTriangle$, and let~$U \in \interspace{R}{S}$.
    By arguments similar to the ones which appeared in the proof of Lemma~\ref{critical:6-cc:restorable:large-neighborhood/lem}, we may assume that~$B$ (respectively~$B'$) are unique regarding~$R$ (respectively~$S$).
    If~$B \neq B'$, then let~$\mathfrak{S}$ be the direct sum of~$\partitionStructure{B,R}$ and~$\partitionStructure{B',S}$, both of which are isomorphic to~$K_2$.
    If~$B = B'$ and~$\equivalenceClasses{B,R} = \equivalenceClasses{B',S}$, then we set~$\mathfrak{S} \coloneqq \partitionStructure{B,R \cup S}$ and~$\mathfrak{S} \cong K_2$.
    If~$B = B'$ and~$\equivalenceClasses{B,R} \neq \equivalenceClasses{B',S}$, then we set~$\mathfrak{S} \coloneqq \partitionStructure{B,R \cup S}$,~$\equivalenceClasses{B,R}$ and~$\equivalenceClasses{B',S}$ are fully intersecting, and the isomorphism type of~$\mathfrak{S}$ is determined by the isomorphism types of its constituents~$(2K_2,2K_2,2K_2)$.
    All automorphisms of~$\mathfrak{S}$ are induced by an automorphism of~$\inducedCC{S \cup R}$.
    Thus each automorphism of~$\inducedCC{B \cup B'}$ extends to an automorphism of~$\inducedCC{B \cup R \cup S \cup B'}$.
    Overall~$\{R,S\}$ is not dominating, but~$R \cup S$ is restorable.
\end{proof}

Next we deal with interspaces between fibers~$R,B$ with~$|R| < |B|$, that contain a constituent isomorphic to~$\interspaceFourSix$.
The interspace pattern of such an interspace depends on its direction:
the interspace~$\interspace{R}{B}$ has the interspace pattern~$\ipsixMatchingComplementD$, while interspace~$\interspace{B}{R}$ has the interspace pattern~$\ipfourClique$.

\begin{lemma}
\label{critical:4cc-6cc/lem}
    Let~$\coherentConfig$ be a critical coherent configuration, and let~$R,S \in \fibers{\coherentConfig}$ such that the interspace~$\interspace{R}{S}$ has the interspace pattern~$\ipsixMatchingComplementD$.
    \begin{enumerate}
        \item\label{4cc-6cc:part1}
        If~$\colorDeg{R} = 1$, then~$\{S\}$ is dominating and for all fibers~$Y$ adjacent to~$S$ in~$\quotientGraph{\coherentConfig}$ other than~$R$ the interspace~$\interspace{Y}{S}$ has either the interspace pattern~$\ipsixMatching$ or the interspace pattern~$\ipsixMatchingTwice$.
        \item
        If~$\interspaceFourSix \in \interspace{R}{S}$, then~$\colorDeg{R} = 1$.
        \item
        If all fibers of~$\coherentConfig$ are small, then~$\wldim{\coherentConfig} = 2$.
    \end{enumerate}
\end{lemma}
\begin{proof}
    We begin by proving the first statement.
    Set~$U \coloneqq U^1_1(\interspace{R}{S})$.
    Due to the interspace pattern of~$\interspace{R}{S}$, we have either~$\clique{2,2,2} \in \inducedCC{S}$ or~$\overrightarrow{C_3}[K_2] \in \inducedCC{S}$.
    Thus by Lemma~\ref{small-cc:induced-cc/lem} we have~$\disjointCliques{3}{2} \in \inducedCC{S}$.
    Hence~$|\ul(\inducedCC{S})| = 3$.

    Next we show that all interspaces incident to~$S$ in~$\quotientGraph{\coherentConfig}$ (except~$\interspace{R}{S}$) have either the interspace pattern~$\ipsixMatching$ or the interspace pattern~$\ipsixMatchingTwice$.
    Towards a contradiction, suppose that there is a fiber~$Y$ adjacent to~$S$ in~$\quotientGraph{\coherentConfig}$ other than~$R$ such that~$\interspace{Y}{S}$ has an interspace pattern other than~$\ipsixMatching$ or~$\ipsixMatchingTwice$.
    Since~$\clique{2,2,2} \in \inducedCC{S}$ or~$\overrightarrow{C_3}[K_2] \in \inducedCC{S}$, the interspace~$\interspace{Y}{S}$ has the interspace pattern~$\ipsixMatchingComplement$ or~$\ipsixMatchingComplementD$.
    We define~$U' \coloneqq U^1_1(\interspace{Y}{S})$.
    For all~$P \in \equivalenceClasses{Y,S}$ there is exactly one~$Q\in \equivalenceClasses{R,S}$ such that for all~$p \in P,q\in Q$ we have either~$p U = q U'$ or~$p U \cap q U' = \emptyset$.
    Thus interspace~$\interspace{Y}{R}$ is not homogeneous, which contradicts~$\colorDeg{R} = 1$.

    Before we proof that the set~$R \cup S$ is restorable, we argue that by Lemma~\ref{critical:restorable:take-care/lem} we may assume that~$\colorDeg{S} = 2$:
    let~$\mathcal{Y}$ the set of all fibers adjacent to~$B$ other than~$R$.
    By the previous reasoning for all~$Y \in \mathcal{Y}$ the interspace~$\interspace{Y}{S}$ has the interspace pattern~$\ipsixMatching$ or~$\ipsixMatchingTwice$.
    Since~$|\ul(\inducedCC{B})| = 3$, there is~$Y \in \mathcal{Y}$ which takes care of all~$Y' \in \mathcal{Y}$ (other than~$Y$).
    Hence we may assume that~$\mathcal{Y} = \{Y\}$ and~$\colorDeg{S} = 2$.

    Lemma~\ref{interspace-pattern:partition-structure/lem} determines the partition structure of~$\interspace{R}{S}$ and~$\interspace{Y}{S}$ as follows.
    While the partition structure~$\partitionStructure{R,S}$ is isomorphic to~$K_4$, either~$\partitionStructure{Y,S}$ is isomorphic to~$K_3$ or~$\overrightarrow{C_3}$ or the isomorphism type of~$\partitionStructure{Y,S}$ is determined by the isomorphism types of its constituents~$(3K_2,3K_2,2\overrightarrow{C_3},3K_2)$.
    All automorphisms of~$\partitionStructure{Y,S}$ are induced by an automorphism of~$\inducedCC{S}$ and all automorphisms of~$\inducedCC{S}$ are induced by an automorphism of~$\partitionStructure{R,S}$.
    Thus all automorphisms of~$\inducedCC{Y}$ extend to an automorphism of~$\inducedCC{Y \cup S \cup R}$, and hence~$R \cup B$ is restorable.
    By Lemma~\ref{critical:restorable/lem} the set~$\{R,S\}$ is dominating.
    Since~$\colorDeg{R} = 1$, we conclude that~$\{S\}$ is dominating.

    Next we prove the second statement.
    So let~$R,S$ be fibers of~$\coherentConfig$ such that~$\interspaceFourSix \in \interspace{R}{S}$.
    Observe that the interspace~$\interspace{R}{S}$ has the interspace pattern~$\ipsixMatchingComplementD$.
    Set~$U \coloneqq U^1_1(\interspace{R}{S})$.
    By Lemma~\ref{small-cc:interspace-implies-cc/lem} it holds that~$\clique{4} \in \inducedCC{R}$.
    Towards a contradiction, suppose~$\colorDeg{R} > 1$ and thus there is a fiber~$B$ other than~$R$ such that~$\interspace{B}{R}$ has interspace pattern~$\ipfourClique$.
    Let~$U' = U^1_1(\interspace{B}{R})$.
    For all~$s \in S$ there is exactly one~$P \in \equivalenceClasses{B,R}$ such that for all~$p \in P$ we have~$pU' = sU^\star$.
    Hence~$\disjointCliques{6}{\frac{\abs{B}}{6},1} \in \interspace{B}{S}$.
    This violates Lemma~\ref{critical:star/lem}.

    Finally, we show the third statement with the help the previous two.
    Let~$\mathcal{Y}$ be the set of fibers adjacent to~$S$ other than~$R$.
    Since all fibers, including those in~$\mathcal{Y}$, are small, by the second statement of the lemma, we have~$\colorDeg{R} = 1$.
    Together with the first statement, we conclude~$\{S\}$ is dominating and for all fibers~$Y \in \mathcal{Y}$ we have~$\disjointCliques{3}{2,2} \in \interspace{Y}{S}$.
    Therefore, for all distinct~$Y,Y' \in \mathcal{Y}$, Lemma~\ref{6-cc:implied-interspace:DUC-DUC/lem} implies that~$\disjointCliques{3}{2,2} \in \interspace{Y}{Y'}$.
    Note that for all~$Y \in \mathcal{Y}$ the interspace~$\interspace{Y}{R}$ is homogeneous.
    Therefore for all~$Y \in \mathcal{Y}$ there is no interspace incident to~$Y$ in~$\quotientGraph{\coherentConfig}$ containing a constituent isomorphic to~$\disjointCliques{2}{3,3}$.
    Furthermore, there are no dominating fibers in~$\mathcal{Y}$.
    Lemma~\ref{critical:6-cc:restorable:cycle/lem} implies that~$\cycle{12} \notin \interspace{Y}{Y'}$ for all~$Y,Y' \in \mathcal{Y}$.
    By Lemma~\ref{critical:6-cc:restorable:DUC:deg1/lem} we have~$\mathcal{Y} = \emptyset$ and~$\vertices(\coherentConfig) = R \cup S$.
    Recall that~\wltwo identifies cliques, and complements of matchings.
    Furthermore on small fibers, the set of constituent graphs defines the coherent configuration.
    Since for all distinct~$rr' \in R^2$ we have~$\abs{r U \cap r' U} = 1$, \wltwo distinguishes~$\interspaceFourSix$ from~$\disjointCliques{2}{2,3}$.
    This yields~$\wldim{\coherentConfig} = 2$.
\end{proof}

\begin{lemma}
\label{critical:7-cc:leviFano/lem}
    If~$\coherentConfig$ is a critical coherent configuration, then there is no path~$(R,B,Y)$ in~$\quotientGraph{\coherentConfig}$ such that~$\leviFano \in \interspace{R}{B}$ and~$\leviFano \in \interspace{B}{Y}$.
\end{lemma}
\begin{proof}
    Towards a contradiction, suppose that such a path exists.
    Let~$U \in \interspace{R}{B}$ and~$U' \in \interspace{Y}{B}$ be such that~$(R \disjointUnion B, U) \cong \leviFano$ and~$(Y \disjointUnion B, U') \cong \leviFano$.
    Suppose~$r \in R$.
    \begin{itemize}
        \item
        If there is~$y \in Y$ such that~$\abs{r U \cap y U'} = 3$, then~$\matching{7} \in \interspace{R}{Y}$ due to Property~\ref{coherent-config:wl2}.

        \item
        Assume there is~$y \in Y$ such that~$\abs{r U \cap y U'} \in \{1,2\}$.
        Therefore, there are exactly three vertices~$y'\in \mathcal{Y}$ such that~$\abs{r U \cap y' U'} = 2$, and there are exactly three vertices~$y''\in \mathcal{Y}$ such that~$\abs{r U \cap y'' U'} = 1$.
        Hence, there is exactly one vertex~$y''' \in \mathcal{Y}$ such that~$\abs{r U \cap y''' U'} = 0$.
        Thus~$\matching{7} \in \interspace{R}{Y}$.

        \item
        Assume there is~$y \in Y$ such that~$\abs{r U \cap y U'} = 0$.
        Suppose that there is~$y' \in  Y \setminus \{y\}$ such that~$\abs{r U \cap y' U'} = 0$.
        Thus there are~$\{b,b'\} \subset B \setminus r U$ such that~$\abs{b U'^\star \cap b' U'^\star} = 2$.
        Since this violates the properties of~$\leviFano$, we have~$\matching{7} \in \interspace{R}{Y}$.
    \end{itemize}
    By Lemma~\ref{critical:star/lem}, the coherent configuration~$\coherentConfig$ is not critical in any of these cases.
\end{proof}

\begin{lemma}
\label{restorable:two-3K2,2-ip:fully-intersecting/lem}
    Let~$\coherentConfig$ be a critical coherent configuration and~$(S,L,S')$ be a path in~$\quotientGraph{\coherentConfig}$ such that interspaces~$\interspace{L}{S}$ and~$\interspace{L}{S'}$ both have the interspace pattern~$\ipsixMatching$.
    If~$|L| = 9$, $\colorDeg{L} = 2$, $|\ul(\inducedCC{L})| = 4$, and~$\equivalenceClasses{L,S}$ and~$\equivalenceClasses{L,S'}$ are fully intersecting, then~$\{L\}$ is dominating.
\end{lemma}
\begin{proof}
    Let the vertex sets~$S_0,S_1,S_2$ and~$S'_0,S'_1,S'_2$ induce the $2$-cliques in~$(S,\arcs^1(\interspace{L}{S}))$ and~$(S',\arcs^1(\interspace{L}{S'}))$ respectively.
    Let~$U  = U^1_1(\interspace{L}{S})$ and~$U' = U^1_1(\interspace{L}{S'})$.
    To write the proof in a concise manner, we rename the vertices of~$L$ as follows.
    Assume that~$L = \{\ell^r_c \mid r,c \in \{0,1,2\}\}$.
    Further, assume that~$\ell^r_c U = S_c$ and~$\ell^r_c U' = S'_r$ for all~$r,c \in \{0,1,2\}$.
    Let~$\psi$ be an automorphism of~$\inducedCC{S \cup S'}$.
    Due to the interspace pattern of~$\interspace{L}{S}$ and~$\interspace{L}{S'}$, we are only interested in the permutation of~$\{S_0,S_1,S_2\}$ and~$\{S'_0,S'_1,S'_2\}$.
    Thus, let~$\psi'$ be the permutation of~$\{S_0,S_1,S_2\}$ and~$\{S'_0,S'_1,S'_2\}$ induced by~$\psi$, let~$\pi$ be projection mapping~$S_r$ to~$r$, and~$\mu$ be the projection of~$S'_c$ to~$c$.
    We now define~$\varphi$ to be an automorphism on~$L \cup S$:
    \[
        \varphi(v) \coloneqq
        \begin{cases}
            \ell^{(\pi \circ \psi' \circ \reverseFunction{\pi})(r)}_{(\mu \circ \psi' \circ \reverseFunction{\mu})(c)}, &\text{if } v = \ell^r_c, \\
            \psi(v),                                                                                                    &\text{if } v \in S \cup S'.
        \end{cases}
    \]
    Thus~$\varphi$ is a extension of~$\psi$ and~$L$ is restorable.
    Therefore by Lemma~\ref{critical:restorable/lem} the set~$\{L\}$ is dominating.
\end{proof}


\section{WL-dimension of graphs with small fibers}
\label{wldim-small/sec}

In this section, we prove an upper bound of~$\frac{n}{20} + o(n)$ on the WL-dimension of a coherent configuration~$\coherentConfig$ on~$n$ vertices under the condition that all fibers of~$\coherentConfig$ are small.
Recall that by Theorem~\ref{critical:small-cc/thm}, if all fibers a of critical coherent configuration~$\coherentConfig$ are small, then either all fibers have the same size, or for all~$R \in \fibers{\coherentConfig}$ we have~$\abs{R} \in \{4,6\}$.
Utilizing this theorem, we separately deal with each case and prove upper bounds on the WL-dimension depending on the fiber size of~$\coherentConfig$.

\begin{lemma}
\label{4-cc:cdeg-decrease/lem}
    Let~$\coherentConfig$ be a critical coherent configuration such that all fibers of~$\coherentConfig$ have size~$4$.
    A fiber~$R$ of~$\coherentConfig$ and all fibers adjacent to it in~$\quotientGraph{\coherentConfig}$ split entirely into tiny fibers in~$\coherentConfig_r$ where~$r \in R$.

    In particular, if~$\colorDeg{R} \geq 4$, then~$\wldim{\coherentConfig}=2$ or  there is a critical coherent configuration~$\coherentConfig'$ for which $\wldim{\coherentConfig} \leq  1 + \wldim{\coherentConfig_r} = 1 + \wldim{\coherentConfig'}$ with~$\abs{\vertices(\coherentConfig)} \leq \abs{\vertices(\coherentConfig')} - 20$.
\end{lemma}
\begin{proof}
    Let~$r \in R$ and~$\{B_1, \dots, B_{\colorDeg{R}}\}$ be the neighborhood of~$R$ in~$\quotientGraph{\coherentConfig}$.
    Since~$|R| = 4$, fiber~$R$ splits into tiny fibers in~$\coherentConfig_r$.
    By Lemma~\ref{small-cc:interspace/lem}, all interspaces incident to~$R$ in~$\quotientGraph{\coherentConfig}$ contain a constituent isomorphic to~$\cycle{8}$ or~$\disjointCliques{2}{2,2}$.
    Hence in~$\coherentConfig_r$ all fibers adjacent to~$R$ in~$\quotientGraph{\coherentConfig}$ split in half.
    This means~$\coherentConfig_r- \{R,B_1,\ldots,B_{\colorDeg{R}}\}$ has WL-dimension at least~$\wldim{\coherentConfig}-1$.
    If~$\wldim{\coherentConfig}>2$, we obtain a critical coherent configuration~$\coherentConfig'$ such that~$\vertices(\coherentConfig') \subseteq \vertices(\coherentConfig) \setminus (R \cup \bigcup_{i=1}^{\colorDeg{R}}B_i)$.
    Observe that~$|R \cup \bigcup_{i=1}^{\colorDeg{R}} B_i| \geq 20$ if~$\colorDeg{R} \geq 4$.
\end{proof}

The lemma essentially allows us to reduce the degree of the quotient graph. For graphs with maximum degree at most~$3$, we can make use of a bound by Fomin and H{\o}ie on the pathwidth of the graph.

\begin{theorem}[see \cite{pathwidthCubicGraphs}]
\label{tw-cubic-graph/lem}
    Every graph~$G$ on~$n$ vertices with maximum degree at most~$3$ satisfies~$\treewidth(G) \leq \pathwidth(G) \leq \frac{n}{6} + o(n)$.
\end{theorem}

We use the theorem to bound the WL-dimension for certain graphs as follows.

\begin{theorem}
\label{4-cc:cfi-wldim/thm}
    Let~$\coherentConfig$ be a critical coherent configuration on~$n$ vertices such that all fibers of~$\coherentConfig$ have size~$4$ and there is no interspace containing a constituent isomorphic to~$\cycle{8}$.
    If all fibers of~$\coherentConfig$ have at most~$3$ neighbors in~$\quotientGraph{\coherentConfig}$, then~$\wldim{\coherentConfig} \leq \frac{n}{24} + o(n)$.
\end{theorem}
\begin{proof}[Proof sketch]
    If there is a dominating fiber~$R$ for which~$\colorDeg{R} \leq 3$, then~$n\leq 16$.
    Thus $\wldim{\coherentConfig} \leq 3$.
    So assume there is no dominating fiber in~$\coherentConfig$.
    Since there are no interspaces containing constituents isomorphic to~$\cycle{8}$, Lemma~\ref{critical:4cc:restorable:2,C4/lem} implies~$|\ul(\inducedCC{R})| = 4$ for all~$R \in \fibers{\coherentConfig}$.
    By Lemma~\ref{critical:4-cc:restorable:DUC/lem}, the quotient graph~$\quotientGraph{\coherentConfig}$ is~$3$-regular.
    Furthermore, whenever we individualize a vertex~$r$ in a fiber~$R$, $R$ splits into singletons in~$\coherentConfig_r$.

    By Lemma~\ref{tw-cubic-graph/lem} we have~$\treewidth(\quotientGraph{\coherentConfig}) \leq \frac{\abs{\fibers{\coherentConfig}}}{6} + o(\abs{\fibers{\coherentConfig}})$.
    Let~$\coherentConfig'$ be another coherent configuration satisfying the same assumptions.
    Thus in the cops and robbers game played on the quotient graphs~$\quotientGraph{\coherentConfig}$ and~$\quotientGraph{\coherentConfig'}$ there is a winning strategy for~$t$ cops, where~$t \leq \frac{\abs{\fibers{\coherentConfig}}}{6} + o(\abs{\fibers{\coherentConfig}})$.
    We mimic this winning strategy now on the corresponding coherent configurations~$\coherentConfig$ and~$\coherentConfig'$ in the bijective pebble game as follows (see e.g.~\cite{DBLP:conf/swat/OtachiS14} for more information on the connection between the two games).
    Whenever we place a pebble on fiber~$R$ in~$\quotientGraph{\coherentConfig}$ (respectively~$\quotientGraph{\coherentConfig'}$), we now place the corresponding pebble on some vertex~$r$ of the corresponding fiber~$R$ in~$\coherentConfig$ (respectively~$\coherentConfig'$).
    By the reasoning above this fiber splits into singletons after applying the coherent closure.
    Therefore, Spoiler has a winning strategy if we add two additional pebbles to the game.

    We conclude~$\wldim{\coherentConfig} \leq \frac{\abs{\fibers{\coherentConfig}}}{6} + o(\fibers{\coherentConfig}) \leq \frac{n}{24} + o(n)$.
\end{proof}

\begin{corollary}
\label{4-cc:wldim/cor}
    Let~$\coherentConfig$ be a critical coherent configuration on~$n$ vertices such that all fibers of~$\coherentConfig$ have size~$4$.
    Then~$\wldim{\coherentConfig} \leq \frac{n}{20} + o(n)$.
\end{corollary}
\begin{proof}
    First, we apply Lemma~\ref{4-cc:cdeg-decrease/lem} repeatedly and restore the criticality after each application until the preconditions of Lemma~\ref{4-cc:cdeg-decrease/lem} are no longer satisfied. Say this happens~$k$ times.
    We then have a critical coherent configuration~$\coherentConfig'$ with~$\wldim{\coherentConfig} \leq k + \wldim{\coherentConfig}$ and whose quotient graph~$\quotientGraph{\coherentConfig'}$ has maximum degree~$3$.
    Observe that~$\abs{\vertices(\coherentConfig')} \leq n - 20 \cdot k$. Note that in the proof of Lemma~\ref{4-cc:cdeg-decrease/lem} we repeatedly take the coherent closure after individualizations and restore criticality, thus all fibers of $\coherentConfig'$ have size exactly~$4$.

    If there are distinct fibers~$R,R'$ in~$\coherentConfig'$ such that~$\cycle{8} \in \interspace{R}{R'}$, then by Lemma~\ref{critical:4cc:restorable:cycle/lem} the set~$\{R,R'\}$ is dominating.
    Thus~$\abs{\vertices(\coherentConfig')} \leq 24$ and~$\wldim{\coherentConfig'} \leq 4\in O(1)$.

    Now assume there is no interspace in~$\coherentConfig'$ containing a constituent isomorphic to~$\cycle{8}$.
    Then we apply Theorem~\ref{4-cc:cfi-wldim/thm} and obtain~$\wldim{\coherentConfig'} \leq \frac{\abs{\vertices(\coherentConfig')}}{24} + o(n)$.
    Altogether we have
    \[
        \wldim{\coherentConfig} \leq 2+k + \max \{4, \wldim{\coherentConfig'}\} \leq 2 + k + \max\left\{ 4, \frac{n - 20 \cdot k}{24} + o(n) \right\}.
    \]
    Observe that~$k \leq \left\lceil \frac{n}{20} \right\rceil$. Thus~$\wldim{\coherentConfig} \leq \frac{n}{20} + o(n)$.
\end{proof}

\begin{lemma}
\label{5-cc:wldim/lem}
    Let~$\coherentConfig$ be a critical coherent configuration such that all fibers of~$\coherentConfig$ have size~$5$.
    Then~$\abs{\fibers{\coherentConfig}} = 1$ and~$\wldim{\coherentConfig} \leq 2$.
\end{lemma}
\begin{proof}
    By Lemmas~\ref{critical:cycle/lem} and~\ref{small-cc:interspace/lem}, there is only one fiber~$R$ in~$\coherentConfig$.
    By Lemma~\ref{small-cc:induced-cc/lem}, either~$\clique{5} \in \inducedCC{R}$ or~$\cycle{5} \in \inducedCC{R}$.
    Thus~$\wldim{\coherentConfig} \leq 2$.
\end{proof}

\begin{lemma}
\label{6-cc:wldim/lem}
    If~$\coherentConfig$ is a critical coherent configuration whose fibers each have size~$4$ or~$6$, then~$\wldim{\coherentConfig} \leq \frac{n}{20}+ o(n)$.
\end{lemma}
\begin{proof}
    If there is an interspace in~$\coherentConfig$ containing the constituent~$\interspaceFourSix$, then by Lemma~\ref{critical:4cc-6cc/lem} we have~$\wldim{\coherentConfig} = 2$.
    If~$|\fibers{\coherentConfig}| = 1$, then we also have~$\wldim{\coherentConfig} \leq 2$.

    Let~$S \in \fibers{\coherentConfig}$ with~$\abs{S} = 6$ be incident to a non-homogeneous interspace in~$\quotientGraph{\coherentConfig}$, and let~$s \in S$.
    Thus by Lemmas~\ref{small-cc:interspace/lem} and~\ref{small-cc:interspace-implies-cc/lem} this yields~$\clique{6} \notin \inducedCC{S}$.
    By Lemma~\ref{small-cc:induced-cc/lem} we have~$\abs{\ul(\inducedCC{S})} \geq 3$.

    Assume that there are fibers~$F,F'$ such that the set~$\{F,F'\}$ is dominating in~$\quotientGraph{\coherentConfig}$.
    Since~$\abs{\ul(\inducedCC{F})} \geq 3$ and~$\abs{\ul(\inducedCC{F'})} \geq 3$, both fibers~$F$ and~$F'$ split into singletons if we individualize~$8$ carefully chosen vertices in~$F \cup F'$.
    Therefore all fibers adjacent to~$F$ or~$F'$ split into tiny fibers.
    Thus~$\wldim{\coherentConfig} \leq 8 + 2 \in \mathcal{O}(1)$.

    We may thus assume for the rest of the proof that there is are no distinct fibers~$F,F'$ in~$\coherentConfig$ such that~$\{F,F'\}$ is dominating.

    Assume that~$\abs{\ul(\inducedCC{S})} = 3$.
    By Lemma~\ref{small-cc:induced-cc/lem}, the configuration~$\inducedCC{S}$ contains two constituents whose underlying graphs are isomorphic either to~$\disjointCliques{3}{2}$ and~$\clique{2,2,2}$, respectively, or to~$\disjointCliques{2}{3}$ and~$\clique{3,3}$, respectively.
    Since all fibers are small, by Lemma~\ref{small-cc:interspace-implies-cc/lem} all interspaces incident to~$S$ in~$\quotientGraph{\coherentConfig}$ contain a constituent isomorphic to~$\disjointCliques{3}{2,2}$ but none isomorphic to~$\cycle{12}$ in the first case and one isomorphic to~$\disjointCliques{2}{3,3}$ in the second case.
    Since~$\{S\}$ is not dominating by assumption, this contradicts Lemma~\ref{critical:6-cc:restorable:DUC:deg1/lem}, and hence~$\abs{\ul(\inducedCC{S})} > 3$.

    Assume that~$\abs{\ul(\inducedCC{S})} = 5$.
    By Lemma~\ref{small-cc:induced-cc/lem} there are the following constituents in~$\inducedCC{S}$: three graphs isomorphic to~$\disjointCliques{3}{2}$ and two isomorphic to~$2\overrightarrow{C_3}$.
    By Lemma~\ref{small-cc:interspace-implies-cc/lem} there is no interspace incident to~$S$ in~$\quotientGraph{\coherentConfig}$ that contains a constituent isomorphic to~$\cycle{12}$.
    By Lemma~\ref{critical:6-cc:restorable:large-neighborhood/lem} there is a small fiber~$R$ adjacent to~$S$ such that~$\interspace{R}{S}$ has the interspace pattern~$\ipsixMatchingMatching$.
    Thus~$\interspace{S}{R}$ has the interspace pattern~$\ipsixMatchingTwice$.
    If~$|\ul(\inducedCC{R})| = 3$, then either~$\disjointCliques{3}{2},\clique{2,2,2} \in \inducedCC{R}$ or~~$\disjointCliques{3}{2},\overrightarrow{C_3}[K_2] \in \inducedCC{R}$.
    Therefore for all fibers~$R'$ adjacent to~$R$ in~$\quotientGraph{\coherentConfig}$ the interspace~$\interspace{R'}{R}$ has the interspace pattern~$\ipsixMatching$ or~$\ipsixMatchingTwice$.
    Since~$\{S\}$ is not dominating, this contradicts Lemma~\ref{critical:6-cc:restorable:DUC:deg1/lem}.
    So~$|\ul(\inducedCC{R})| > 3$.
    Thus Lemma~\ref{critical:6-cc:restorable:large-neighborhood/lem} applies and the neighborhood of~$B'$ partitions into~$\{\mathcal{R},\mathcal{Y},\mathcal{B}\}$ with the properties as described in Lemma~\ref{critical:6-cc:restorable:large-neighborhood/lem}.
    Since~$S \in \mathcal{Y}$ (follows from the interspace pattern of~$\interspace{S}{R}$), none of parts of~$\{\mathcal{R},\mathcal{Y},\mathcal{B}\}$ are empty.
    Let~$R' \in \mathcal{R}$ and~$B' \in \mathcal{B}$.
    Consider~$\coherentConfig_r$ where~$r \in R$.
    Then fibers~$R'$, $B'$, $R$, and~$S$ split into tiny fibers in~$\coherentConfig_r$.
    After restoring criticality, we obtain a critical coherent configuration~$\coherentConfig'$ such that~$\vertices(\coherentConfig') \subseteq \vertices(\coherentConfig) \setminus (R' \cup S \cup B' \cup Y)$.
    Note that~$\abs{\vertices(\coherentConfig')} \leq \abs{\vertices(\coherentConfig)} - 20$ since~$\abs{B} \in \{4,6\}$.

    Assume that~$|\ul(\inducedCC{S})| = 4$.
    By Lemma~\ref{small-cc:induced-cc/lem}, there is exactly one constituent in~$\inducedCC{S}$ which is isomorphic to~$\disjointCliques{3}{2}$.
    Hence there is no interspace incident to~$S$ which has the interspace pattern~$\ipsixMatchingMatching$.
    Lemma~\ref{critical:6-cc:restorable:large-neighborhood/lem} applies, and therefore the neighborhood of~$S$ in~$\quotientGraph{\coherentConfig}$ partitions in~$\{\mathcal{R},\mathcal{B},\mathcal{Y}\}$ such that
    \begin{itemize}
        \item
        for all~$R \in \mathcal{R}$ we have~$\cycle{12} \in \interspace{S}{R}$,
        \item
        for all~$B \in \mathcal{B}$ we have~$\disjointCliques{2}{3,\frac{\abs{B}}{2}} \in \interspace{S}{B}$, and
        \item
        for all~$Y \in \mathcal{Y}$ the number of constituents in~$\interspace{S}{Y}$ isomorphic to~$\disjointCliques{3}{2,2}$ is exactly~1 or~3.
    \end{itemize}
    By Lemma~\ref{critical:6-cc:restorable:cycle/lem}, all sets~$\mathcal{R}, \mathcal{B}, \mathcal{Y}$ are non-empty and~$\mathcal{R}$ contains a single fiber unique to~$S$.
    Thus let~$R\in\mathcal{R}$, $B \in \mathcal{B}$, and~$Y \in \mathcal{Y}$.
    Now we have two cases:

    If there are exactly three constituents in~$\interspace{S}{Y}$ which are isomorphic to~$\disjointCliques{3}{2,2}$, then fibers~$R$, $B$, $Y$, and~$S$ split into tiny fibers in~$\coherentConfig_s$.
    After restoring criticality, we obtain a critical coherent configuration~$\coherentConfig'$ such that~$\vertices(\coherentConfig') \subseteq \vertices(\coherentConfig) \setminus (R \cup S \cup B \cup Y)$.
    Further note that~$\abs{\vertices(\coherentConfig')} \leq \abs{\vertices(\coherentConfig)} - 20$ since~$\abs{B} \in \{4,6\}$.

    Now assume that there is exactly one constituent in~$\interspace{S}{Y}$ which is isomorphic to~$\disjointCliques{3}{2,2}$.
    This implies~$|\ul(\inducedCC{Y})| \leq 4$.
    By Lemma~\ref{critical:small-cc:module/lem} we conclude that~$|\ul(\inducedCC{Y})| = 4$.
    Therefore there is no fiber~$R'$ adjacent to~$Y$ such that~$\interspace{R'}{Y}$ has the interspace pattern~$\ipsixMatchingMatching$.
    Hence, Lemmas~\ref{critical:6-cc:restorable:large-neighborhood/lem} and additionally Lemma~\ref{critical:6-cc:restorable:cycle/lem} apply to~$Y$ since (otherwise we obtain a contradiction by reasoning similar to the one above; either~$\{Y\}$ would be dominating or~$|\ul(\inducedCC{Y})|=3$).
    Thus there is an~$F \in \fibers{\coherentConfig}$ such that~$\cycle{12} \in \interspace{Y}{F}$.
    Since~$\coherentConfig$ is critical, Lemma~\ref{critical:adjacent-interspace-cycle/lem} implies~$F \neq S$ and~$F \neq R$.
    Lemma~\ref{6-cc:implied-interspace:DUC-DUC/lem} implies~$\disjointCliques{3}{2,2} \in \interspace{S}{F}$.
    Observe that it is impossible in a critical coherent configuration for an interspace to contain both a constituent isomorphic to~$\disjointCliques{2}{3,3}$ and a constituent isomorphic to~$\disjointCliques{3}{2,2}$ (compare Table~\ref{small-cc:classificaiton-small-interspaces/tab}).
    Thus~$F \neq B$ and~$|\mathcal{Y}| \geq 2$.
    In~$\coherentConfig_s$ fibers~$R$, $S$, and~$B$ split into tiny fibers while size-$2$ fibers~$Y'$ and~$F'$ are split off from~$Y$ and~$F$, respectively.
    After we restore criticality, we obtain a critical coherent configuration~$\coherentConfig'$ such that~$\vertices(\coherentConfig') \subseteq \vertices(\coherentConfig) \setminus (R \cup Y' \cup S \cup B \cup F')$.
    Further note that~$\abs{\vertices(\coherentConfig')} = \abs{\vertices(\coherentConfig)} - 20$ since~$\abs{B} \in \{4,6\}$.

    We summarize the proof up to this point:
    in several edge cases, we have~$\wldim{\coherentConfig} \in \mathcal{O}(1)$.
    Since~$\coherentConfig$ is critical, there is no size-$6$ fiber with~$\abs{\ul(\inducedCC{S})} = 3$.
    If there is a size-$6$ fiber with~$\abs{\ul(\inducedCC{S})} > 3$, then $\wldim{\coherentConfig} \leq   1 + \wldim{\coherentConfig'}$ for a critical coherent configuration with~$\abs{\vertices(\coherentConfig')} \leq \abs{\vertices(\coherentConfig)} - 20$ also satisfying the assumptions of the lemma.

    We iterate the process above until no fiber of size~$6$ exists anymore and define~$k$ to be the number of iterations.
    What remains is a coherent configuration~$\coherentConfig''$ whose fibers all have size~$4$.
    Thus we apply Corollary~\ref{4-cc:wldim/cor}.
    Altogether we have
    \[
        \wldim{\coherentConfig} \leq 2+ k + \max \{8, \wldim{\coherentConfig''}\} \leq 2 + k + \max \left\{8, \frac{n - 20 \cdot k}{20} + o(n) \right\}.
    \]
    Observe that~$k \leq \left\lceil \frac{n}{20} \right\rceil$.
    We conclude~$\wldim{\coherentConfig} \leq \frac{n}{20} + o(n)$.
\end{proof}

\begin{lemma}
\label{7-cc:wldim/lem}
    Let~$\coherentConfig$ be a critical coherent configuration such that all fibers of~$\coherentConfig$ have size~$7$.
    Then~$\abs{\fibers{\coherentConfig}} \leq 2$ and~$\wldim{\coherentConfig} \leq 3$.
\end{lemma}

\begin{proof}
    By Lemmas~\ref{critical:cycle/lem} and~\ref{critical:7-cc:leviFano/lem}, we have~$\abs{\fibers{\coherentConfig}} \in\{1,2\}$:
    assuming~$\fibers{\coherentConfig} = \{R\}$, then by Lemma~\ref{small-cc:induced-cc/lem} there is constituent in~$\inducedCC{R}$ whose underlying graph is isomorphic to~$\clique{7}$ or~$\cycle{7}$.
    Hence~$\wldim{\coherentConfig} \leq 2$.
    So we assume that~$\{R,B\} = \fibers{\coherentConfig}$ and~$U \in \interspace{R}{B}$ such that~$(R \disjointUnion B, U) \cong \leviFano$.
    By Lemma~\ref{small-cc:induced-cc/lem} we have~$\clique{7} \in \inducedCC{R}$ and~$\clique{7} \in \inducedCC{B}$ and let~$r \in R$.
    Thus~$\interspaceFourSix \in \coherentConfig_r[R \setminus \{r\},B \setminus rU]$.
    By Lemma~\ref{critical:4cc-6cc/lem}, we have~$\wldim{\coherentConfig} \leq 3$.
\end{proof}

\begin{theorem}
\label{small-cc:wldim/thm}
    Let~$\coherentConfig$ be a critical coherent configuration such that all fibers of~$\coherentConfig$ are small.
    Then~$\wldim{\coherentConfig} \leq \frac{n}{20} + o(n)$.
\end{theorem}
\begin{proof}
    We conclude the theorem from Theorem~\ref{critical:small-cc/thm}, Corollary~\ref{4-cc:wldim/cor}, and Lemmas~\ref{5-cc:wldim/lem}, \ref{6-cc:wldim/lem}, and~\ref{7-cc:wldim/lem}.
\end{proof}


\section{Limiting fiber sizes}
\label{sec:limit:fiber:sizes}

In this section, we bound the fiber size in a coherent configuration by a constant in exchange for a sublinear cost.
This allows us to find an upper bound on the WL-dimension of coherent configurations using the treewidth of the quotient graph.
Recall that this is the second base case that we mentioned at the beginning of this chapter.
To prove the bound on the fiber size, we first adopt Zemlyachenko's argument for limited color valency.

For each~$\interspace{R}{B}$ where fibers~$R$ and~$B$ are not necessarily distinct, we arbitrarily choose a basis relation whose degree is maximal.
A \emph{non-maximal} basis relation is basis relation other than the chosen one.
Define~$\overline{\minimalDegree{R}{B}}$ to be the maximum degree among all degrees of non-maximal basis relations in~$\interspace{R}{B}$.
Set~$\overline{\minimalDegree{R}{B}}=0$ if~$|\interspace{R}{B}|=1$.

Define the \emph{max-modules} of a coherent configuration to be the connected components of the graph in which the edges are formed by all non-maximal basis relations.
(The max-modules are sections in the sense of~\cite{DBLP:journals/dam/Goldberg83} and colored modules in the sense of~\cite{DBLP:journals/mst/Schweitzer17}.)
It is not difficult to see that the WL-dimension of a coherent configuration is at most the maximum of~$2$ and the maximum WL-dimension of its max-modules.

\begin{lemma}[Color valence limit, (Zemlyachenko, see~\cite{DBLP:conf/fct/Babai81})]
\label{lem:bound-on-min:degree}
    Let~$\coherentConfig$ be a critical coherent configuration, let~$d \in \Nat$ be a positive integer.
    There is a set of vertices~$S$ of size at most~$\frac{2n}{d}$  so that in the components of
    $\coherentConfig_S$ we have~$\overline{\minimalDegree{R}{B}}\leq d$
     for all (not necessarily distinct) fibers~$R,B \in \fibers{\coherentConfig_S}$.
\end{lemma}
\begin{proof}[Proof sketch]
    Let us observe that it suffices to argue that we can decrease the degree bound from~$2d'$ to~$d'$ by individualizing~$n/d'$ vertices as follows. Indeed, if this is possible then to reduce from a degree bound of~$n$ to~$d$ we use at most~$n/d+n/(2d)+n/(4d) + \ldots \leq 2n/d$ individualizations.

    Now to reduce from~$2d'$ to~$d'$ suppose~$\overline{\minimalDegree{R}{B}}>d$.
    We individualize a vertex in~$B$ and refine.
    Then some subset of vertices in~$R$ of size more than~$d$ but size at most~$2d$ is separated from~$R$. Thus the number of vertices contained in fibers of size more than~$2d$ reduces by more than~$d$. This can only happen at most~$n/d$ times.
\end{proof}

\begin{lemma}[Fiber size limit]
\label{lem:bound-on-cc-size}
    Let~$\coherentConfig$ be a critical coherent configuration, and let~$c, d \in \Nat$ be positive integers.
    There is a set~$S$ of vertices of size at most~$\frac{2n}{d} + \frac{dn}{c}$ such that for each max-module~$\overline{\coherentConfig}$ of~$\coherentConfig_S$ we have~$\abs{R} \leq c$ and~$\overline{\minimalDegree{R}{B}}\leq d$
     for all (not necessarily distinct) fibers~$R,B \in \fibers{\overline{\coherentConfig}}$.
\end{lemma}
\begin{proof}
    Using the previous lemma we will assume that~$ \overline{\minimalDegree{R}{B}}\leq d$ for all fibers~$B$ and~$R$ and then argue that there is a set of size~$\frac{dn}{c}$ that achieves our goal. If~$d\geq c$ we can set~$S=\vertices(\coherentConfig)$ so we may assume that~$d<c$.
    Choose, if it exists, a vertex~$r$ and a sequence  of~$U_1,\ldots,U_t$, with each~$U_i$ a union of non-maximal basis relations  in some interspace or fiber,
    such that~$|rU_1U_2\cdots U_{i+1}|\leq d\cdot |rU_1U_2\cdots U_{i}|$ and~$|rU_1U_2\cdots U_{t}|>c$.
    Then there is an~$i$ such that~$|rU_1U_2\cdots U_{i+1}|> c$ and~$c\geq |rU_1U_2\cdots U_{i}|\geq \frac{c}{d}$.

    By individualizing~$r$ and refining, at least~$\frac{c}{d}$ vertices contained in a fiber of size at least~$c$ split off. Afterwards these vertices are in a fiber of size less than~$c$. This can happen at most~$\frac{nd}{c}$ times.
    We apply the individualizations one after the other and each time apply the 2-dimensional Weisfeiler-Leman algorithm.

    At some point, no sequence~$U_1,\ldots,U_t$ with the desired properties exists. At this point, the max-modules of the graph have fibers of size at most~$c$.
\end{proof}

\begin{lemma}
\label{lem:bd:tw:and:fibre:size:bd:WL}
    Let~$\coherentConfig$ be a coherent configuration.
    If all fibers of~$\coherentConfig$ have size at most~$t$, then~$\wldim{\coherentConfig}\leq t\cdot \treewidth(\quotientGraph{\coherentConfig})$, where~$\treewidth(\quotientGraph{\coherentConfig})$ is the treewidth of the quotient graph.
\end{lemma}
\begin{proof}[Proof sketch]
    We first reduce the problem to a similar problem on graphs\footnote{We should remark that there are also simpler, more direct ways to show the statement with only slightly worse bounds. For example using~\cite{DBLP:conf/swat/OtachiS14} would give a bound of~$t\cdot \treewidth(\quotientGraph{\coherentConfig})+3$ by simple dynamic programming. That bound is also sufficient for our purposes.}.
    Let~$G'$ be the graph on the same vertex set as~$\coherentConfig$ where two vertices are adjacent if they are in the same fiber or in fibers adjacent in~$\quotientGraph{\coherentConfig}$.
    From this, we construct a vertex colored graph~$(G,\chi)$ by subdividing edges twice as of~$G$ as follows. The graph~$(G,\chi)$ contains a vertex for every pair~$(v,u)$ of (not necessarily distinct) vertices of~$\coherentConfig$ which are contained in the same fiber or fibers adjacent in~$\quotientGraph{\coherentConfig}$.
    Vertex~$(v,u)$ is adjacent to~$(v,v)$ and to~$(u,v)$.
    The vertex~$(u,v)$ of~$G$ is colored with the color~$(F_u,F_v,U)$ where~$F_u$ is the fiber containing~$u$,~$F_v$ is the fiber containing~$v$ and~$U$ is the interspace containing~$(u,v)$.

    Since subdivision of edges does not increase the treewidth, the graph~$(G,\chi)$ has tree\-width at most~$t\cdot \treewidth(\quotientGraph{\coherentConfig})$.
    Moreover, if~$\wldim{\coherentConfig}\geq 2$ then, the graph~$G$ has the same WL-dimension as~$G$ since~$\wltwo$ can recover~$\coherentConfig$ from~$G$.

    To remove the vertex coloring of~$G$ we can transform~$G$ into a graph~$\widehat{G}$ by attaching leaves to every vertex of~$G$. Here, the number of leaves attached is so that each vertex gets at least two leaves and two vertices gain the same number of leaves exactly if the have the same color. The treewidth and the WL-dimension remain unchanged.

    Now it remains to to show that~$\treewidth(\widehat{G})\leq t\cdot  \treewidth(\quotientGraph{\coherentConfig})$ since then~$\wldim{\coherentConfig}\leq  \WLdim(\widehat{G})\leq t\cdot \treewidth(\quotientGraph{\coherentConfig})$.
    For this we simply observe that if~$\beta\colon T\rightarrow \mathcal{P}(V(\coherentConfig))$ is a tree decomposition for~$\quotientGraph{\coherentConfig}$ of width~$w$, then~$\beta'\colon T\rightarrow \mathcal{P}(V(\widehat{G}))$ with~$\beta'(r) = \{v\mid v\in \beta(t)\}$ is a tree decomposition~$\widehat{G}$ of width~$t\cdot w$.
\end{proof}


\section{A potential function}
\label{sec:potential:func}

For a coherent configuration~$\coherentConfig$ let~$\parameters(\coherentConfig)$ be the triple~$(n_\ell, k_\ell,n_s)$, where~$n_\ell$ is the number of vertices in large fibers,~$k_\ell$ is the number of large fibers and~$n_s$ is the number of small fibers.
Define the function~
\[
    \tau \colon \Nat^3 \rightarrow \Rel \text{ with }
    (n_\ell, k_\ell, n_s) \mapsto \frac{3n_\ell + n_s - 8k_\ell}{20}.
\]
Note that the parameters and function~$\tau$ are additive in the sense that we can compute the parameters separately for a configuration induced by an arbitrary partition of the fibers and then add them up.

Let~$\widehat{\f} \colon \Nat^3 \longrightarrow \Nat$ be the function that assigns to the triple~$(n_\ell, k_\ell,n_s)$ the maximum WL-dimension of a coherent configuration~$\coherentConfig$ with~$\parameters(\coherentConfig)=(n_\ell, k_\ell,n_s)$.
We define~$\f$ to be the function that assigns to a coherent configuration $\coherentConfig$ the maximum WL-dimension among all coherent configurations $\coherentConfig'$ for which~$\tau(\parameters(\coherentConfig'))\leq \tau(\parameters(\coherentConfig))$. (In some sense~$f$ is a monotonization of~$\widehat{f}$.)
Using this notation, note that the choice of potential function guarantees that if~$\coherentConfig'$ is finer than~$\coherentConfig$ then~$\f(\coherentConfig')\leq \f(\coherentConfig)$.
Finally, we define the function
\[
    \widetilde{\f} \colon \Rel \longrightarrow \Rel \text{ with } r \mapsto \max \{ \f(\coherentConfig) \mid \tau (\parameters(\coherentConfig)) \leq r \}.
\]

Our goal in the upcoming sections will be to show that~$\wldim{\coherentConfig}\leq \tau(\parameters(\coherentConfig))+o(n)$.
A key technique will be reductions that show that~$\wldim{\coherentConfig}\leq t+ \f(\tau (\parameters(\coherentConfig')))$ for some coherent configuration~$\coherentConfig'$ obtained from~$\coherentConfig$ in a controlled manner, which will give a recursive bound of the form~$\wldim{\coherentConfig} \leq t + \widetilde{f} (\tau(\parameters(\coherentConfig)) -t')$ where~$t' \geq t$.

\section{Local reductions}
\label{recursive-argement/sec}

In this section, we provide local reductions that trade progress in the potential with individualizations.
The previous section introduced the control mechanism balancing progress and costs.
However, the local reductions also have structural consequences which eventually allow us to restrict the structure of the coherent configuration in three ways:
first, it limits the number of neighbors a large fiber can have in the quotient graph.
Second, we eliminate certain interspace patterns.
For some other interspace patterns, we show that there can be at most one small neighboring fiber attached with this pattern.
Finally, we reduce the quotient graph degree of fibers of size~$4$ or~$6$.

Each reduction proceeds as follows.
We examine the neighborhood of a fiber.
In particular, we describe the neighborhood with the help of the occurring interspace patterns and the quotient degree.
We individualize one or two carefully chosen vertices, take the coherent closure, and restore criticality.
Afterwards, we weigh the number of individualized vertices against the progress in the terms of the potential function~$\tau$ we made due to fibers that split.

As already implied, many claims and proofs heavily rely on interspace patterns, which we defined in Section~\ref{interspace-large-small/sec}.
In particular, recall that by Corollary~\ref{large-small-interspace:classification:uniqueness/cor} an interspace between a small and a large fiber has exactly one of 17 interspace patterns listed in Theorem~\ref{global-argument:large-small-interspace:classification}.

To write the proofs of this section more concisely, we introduce an auxiliary notation as follows.
For triples~$(a,b,c),(a',b',c')\in \Nat^3$, we will use~$(a,b,c)\preceq (a',b',c')$ to denote that~$\tau(a,b,c)\leq \tau(a',b',c')$.
Define a function~$\hfunction\colon \Rel^+\rightarrow \Rel$ by setting~$\hfunction(a)= \max\{-0.4,-3\cdot \frac{\lceil 8a\rceil }{20}\}$.

\begin{lemma}
\label{local:progress-in-large/lem}
    Let~$\coherentConfig, \coherentConfig'$ be a coherent configurations so that~$\coherentConfig$ consists of a single large fiber and~$\coherentConfig'\finer \coherentConfig$.
    Let~$R_1,\ldots,R_u$ be the fibers of~$\coherentConfig'$ and~$\{t_i \mid i \in\{1,\ldots,u\}, |R_i|= t_i \cdot|\coherentConfig| \}$ be a set.
    Then~$\tau(\parameters(\coherentConfig'))\leq \tau(\parameters(\coherentConfig'))+\sum_{i = 1}^{u} \hfunction(t_i)+0.4$.
\end{lemma}
\begin{proof}
    For each  fiber~$R_i$ set~$\Delta \parameters(R_i)=
    (a_i,b_i,c_i)$ where~
    \[
        a_i= \begin{cases}
            0        & \text{if~$|R_i|$ is large}\\
            -|R_i|   & \text{otherwise,}
        \end{cases} \
        b_i=\begin{cases}
            1       & \text{if~$R_i$ is large} \\
            0       & \text{otherwise,}
        \end{cases}
        \ \text{and}\
        c_i =\begin{cases}
            |R_i|   & \text{if~$R_i$ is small}\\
            0       & \text{otherwise.}
        \end{cases}
    \]

    We have~$\parameters(\coherentConfig)- \parameters(\coherentConfig')= (0,-1,0)+\sum_{i=1}^{u} (\Delta \parameters(R_i))$. It thus suffices to show in the following that~$\tau(\Delta \parameters(R_i))\leq \hfunction(t_i)$ for each~$i\in \{1,\ldots,u\}$.
    \begin{itemize}
        \item
        If~$R_i$ is tiny, then~$\Delta \parameters(R_i) = (-|R_i|,0,0)\preceq (-\lceil 8t_i\rceil,0,0)$.
        So~$\tau(\Delta \parameters(R_i))\leq -3\cdot \frac{\lceil 8t_i\rceil }{20}$.
        \item
        If~$R_i$ is small, then~$\Delta \parameters(R_i) = (-|R_i|,0,|R_i|) \preceq (-4,0,4)$.
        So~$\tau(\Delta \parameters(R_i))\leq -0.4$.
        \item
        If~$R_i$ is large, then~$\Delta \parameters(R_i) = (0,1,0)$ independent of~$t_i$.
        Hence~$\tau(\Delta \parameters(R_i))\leq -0.4$.
        \qedhere
    \end{itemize}
\end{proof}

We should remark that while Lemma~\ref{local:progress-in-large/lem} is sufficient to calculate the progress made in~$L$ in many cases, in some other case a closer examination produces better results on the progress.
This is especially true if we know conditions on the size of~$L$ or sizes of the vertex sets into which~$L$ splits.

\begin{lemma}
    \label{lem:local-argument:3-large-neighbors}
    Let~$\coherentConfig$ be a critical coherent configuration, $R \in \fibers{\coherentConfig}$, and~$\{L_1,L_2,L_3\}$ a set of large fibers adjacent to~$R$ in~$\quotientGraph{\coherentConfig}$.
    If~$R$ is a large fiber or~$|\ul(\inducedCC{R})| \geq 3$, then~$\wldim{\coherentConfig} \leq 1 + \widetilde{\f}( \tau(\parameters(\coherentConfig)) - 1)$.
\end{lemma}
\begin{proof}
    Let~$r$ be a vertex of~$R$ and~$i \in \{1,2,3\}$.
    We define~$(a_i,b_i,c_i) \coloneqq \parameters(\coherentConfig_r[L_i]) - \parameters(\coherentConfig[L_i])$ and~$t_{L_i} \coloneqq \tau(a_i,b_i,c_i)$.
    Lastly define~$x_i \coloneqq \abs{L_i}-\minimalDegree{R}{L_i}$.
    In the following we examine~$\coherentConfig_r[L_i]$.
    \begin{itemize}
        \item
        First assume that~$\minimalDegree{R}{L_i} \geq 8$.
        Recall that~$\abs{R} \geq 2\minimalDegree{R}{L_i}$.
        The fiber~$L_i$ splits into two large fibers (or unions of fibers).
        Hence~$(a_i,b_i,c_i) \preceq (0,1,0)$ and~$t_{L_i} \leq -0.4$.

        \item
        Next assume that~$4 \leq \minimalDegree{R}{L_i} < 8$.
        If~$x_i \geq 8$, then a large fiber remains.
        Thus~$(a_i,b_i,c_i) \preceq (-\minimalDegree{R}{L_i},0,\minimalDegree{R}{L_i}) \preceq (-4,0,4)$ and~$t_{L_i} \leq -0.4$.
        If~$4 \leq x_i < 8$, then~$L_i$ splits entirely into small fibers.
        Hence~$(a_i,b_i,c_i) \preceq (-\abs{L_i},-1,\abs{L_i}) \preceq (-8,-1,8)$ and~$t_{L_i} \leq -0.4$.

        \item
        Finally assume that~$\minimalDegree{R}{L_i} < 4$.
        If~$x_i \geq 8$, then a large fiber remains.
        Thus $(a_i,b_i,c_i) \preceq (-\minimalDegree{R}{L_i},0,0) \preceq (-2,0,0)$ and~$t_{L_i} \leq -0.3$.
        If~$x_i < 8$, then fiber~$L_i$ splits entirely into small and tiny fibers with at least two vertices in tiny fibers.
        Hence~$(a_i,b_i,c_i) \preceq (-\abs{L_i},-1,\abs{L_i} - \minimalDegree{R}{L_i}) \preceq (-8,-1,6)$ and~$t_{L_i} \leq -0.5$.
    \end{itemize}
    Summarizing the case distinction, we obtain~$t_{B_j} \leq -0.3$ for all~$j \in \{1,2,3\}$.

    If~$R$ is large, we also obtain~$t_R\leq 0.15$ since one vertex in~$R$ is individualized.
    If~$R$ is small and~$|\ul(\inducedCC{R})| \geq 3$, then by Lemma~\ref{small-cc:induced-cc/lem} we obtain~$t_r\leq 0.1$ since at least one size-$2$ fiber splits from~$B$ in~$\coherentConfig_b$.

    Altogether, we conclude
    \[
        \tau(\parameters(\coherentConfig_r)) \leq \tau(\parameters(\coherentConfig)) + t_R + \sum_{i=1}^{3} t_{B_i} \leq \tau(\parameters(\coherentConfig)) - 1.
        \qedhere
    \]
\end{proof}

For a coherent configuration~$\coherentConfig$, we define~$\quotientGraphLarge{\coherentConfig}$ and~$\quotientGraphSmall{\coherentConfig}$ to be the subgraph of~$\quotientGraph{\coherentConfig}$ induced by the set of all large fibers and small fibers, respectively.

\begin{lemma}\label{lem:max:degree:2:means:tw:3}
    If~$\coherentConfig$ is a coherent configuration in which no large fiber has at least 3 large neighbors, then the connected components of~$\quotientGraphLarge{\coherentConfig}$ have treewidth at most~$2$.
\end{lemma}

\begin{proof}
    If no large fiber has 3 large neighbors, then~$\quotientGraphLarge{\coherentConfig}$ has maximum degree at most 2 and thus treewidth at most 2.
\end{proof}

\begin{theorem}
\label{local:L-S/thm}
    Let~$\coherentConfig$ be a critical coherent configuration and suppose~$L,S \in \fibers{\coherentConfig}$ with~$L$ large and~$S$ small.
    If the set~$\{S\}$ is not dominating and the interspace~$\interspace{L}{S}$ has one of the interspace patterns~$\ipsixMatchingAndCycle$, $\ipsixMatchingMatching$, $\ipsixMatchingTwice$, $\ipsixMatchingComplement$, $\ipsixMatchingAndComplement$, $\ipsixTriangleComplement$, or $\ipsixTriangleComplementTwice$,
    then~$\wldim{\coherentConfig} \leq 1 + \widetilde{\f}( \tau(\parameters(\coherentConfig)) - 1)$.
\end{theorem}
\begin{proof}
    Choose~$\ell \in L$ and~$s\in S$ arbitrarily.
    Let~$Z$ be a union of fibers of~$\coherentConfig$ and~$V' \subseteq \vertices(\coherentConfig)$.
    We define~$(a_Z,b_Z,c_Z) \coloneqq \parameters(\coherentConfig_{V'}[Z]) - \parameters(\coherentConfig[Z])$ and~$t_Z \coloneqq \tau(a_Z,b_Z,c_Z)$.

    (\textit{Case~$\ipsixTriangleComplement$}).
    Assume that~$\interspace{L}{S}$ has the interspace pattern~$\ipsixTriangleComplement$.
    This implies that~$9$ divides~$|L|$.
    By Lemma~\ref{critical:6-cc:restorable:DUC:deg1/lem}, there is a fiber~$R \in \fibers{\coherentConfig}\setminus\{L\}$ such that~$\interspace{R}{S}$ has the interspace pattern~$\ipsixTriangle$.
    \begin{itemize}
        \item
        If~$|L| \leq 18$, we set~$V' \coloneqq \{s\}$ and consider~$\coherentConfig_s$.
        Fiber~$L$ splits into fibers of size~$\frac{|L|}{3}$ and~$\frac{2|L|}{3}$.
        If~$|L| = 9$, then~$t_L \leq \tau(-9,-1,6) \leq -0.65$.
        If~$|L| = 18$, then~$t_L \leq \tau(-6,0,0) \leq -0.9$.
        Fiber~$R$ splits into two fibers of equal size.
        If~$|R| \geq 16$, then~$t_R \leq \tau(0,1,0)\leq-0.4$, and
        if~$|R| < 16$, then~$t_R \leq \tau(-8,-1,8) \leq -0.4$.
        Fiber~$S$ splits into tiny fibers and thus~$t_S \leq -0.3$.

        \item
        If~$|L| \geq 27$,  we set~$V' \coloneqq \{\ell\}$ and consider~$\coherentConfig_\ell$.
        Fiber~$L$ splits into three fibers~$L'$,~$L''$, and~$L'''$ of sizes~$\frac{|L|}{9}$, $\frac{4|L|}{9}$, and~$\frac{4|L|}{9}$ respectively (or something finer).
        If~$|L| = 27$, then~$t_L \leq \tau(-3,1,0) \leq -0.85$.
        If~$36 \leq |L| \leq 63$, then~$t_L \leq \tau(-4,1,4) \leq -0.8$.
        If~$72 \leq |L|$, then~$t_L \leq \tau(0,2,0) \leq -0.8$.
        Fiber~$S$ splits into a size~$2$ fiber and a fiber~$S'$ of size~$4$.
        Observe that either~$\coherentConfig_\ell[L'',S']$ or~$\coherentConfig_\ell[L''',S']$ contains a basis relation (or a union of basis relation) isomorphic to a star.
        Thus~$t_S \leq \tau(0,0,-6) \leq -0.3$.
    \end{itemize}
    We conclude that~$\tau(\parameters(\coherentConfig_{V'})) \leq \tau(\parameters(\coherentConfig)) - 1.1$.

    (\textit{Case~$\ipsixMatchingMatching$}).
    Assume that the interspace~$\interspace{L}{S}$ has the interspace pattern~$\ipsixMatchingMatching$.
    This implies that~$3$ divides~$|L|$.
    We set~$V' \coloneqq \{s\}$ where~$s \in S$.
    There are only singletons in~$\fibers{\coherentConfig_{s}[S]}$ and~$t_S \leq \tau(0,0,-6) \leq -0.3$.
    Since~$\coherentConfig_s[S]$ is discrete and~$|\partition{L,S}| = 3$, fiber~$L$ splits into three fibers (or unions of fibers) in~$\coherentConfig_{s}$.
    Each of these fibers has size~$\frac{\abs{L}}{3}$.
    By Lemma~\ref{local:progress-in-large/lem} we obtain~$t_L \leq -0.8$.
    We conclude that~$\tau(\parameters(\coherentConfig_{s})) \leq \tau(\parameters(\coherentConfig))- 1.1$.

    (\textit{Case~$\ipsixMatchingAndCycle$}).
    Assume that the interspace~$\interspace{L}{S}$ has the interspace pattern~$\ipsixMatchingAndCycle$.
    This implies that~$6$ divides~$|L|$.
    Let~$U' = U^1_1(\interspace{L}{S})$ and~$U'' = U^2_1(\interspace{L}{S})$.
    We set~$V' \coloneqq \{\ell\}$.
    There are only fibers of size at most~$2$ in~$\fibers{\coherentConfig_\ell[S]}$.
    Thus~$t_S \leq \tau(0,0,-6) \leq -0.3$.
    There are the following fibers (or unions of fibers) in~$\fibers{\coherentConfig_\ell[L]}$:
    $\{v \in L \mid vU' = \ell U'\}$ and
    $\{v \in L \mid vU' \not\subseteq \ell U' \cup \ell U''\}$, both of which have size~$\frac{\abs{L}}{6}$, as well as
    $\{v \in L \mid \abs{vU' \cap \ell U' } = 1 \}$ and
    $\{v \in L \mid \abs{vU' \cap \ell U''} = 1 \wedge vU' \cap \ell U' = \emptyset \}$, both of which have size~$\frac{\abs{L}}{3}$.
    By Lemma~\ref{local:progress-in-large/lem} we obtain~$t_L \leq -1$.
    We conclude that~$\tau(\parameters(\coherentConfig_\ell)) \leq \tau(\parameters(\coherentConfig))- 1.3$.

    (\textit{Case~$\ipsixMatchingTwice$}).
    Assume that~$\interspace{L}{S}$ has the interspace pattern~$\ipsixMatchingTwice$. This implies that~$3$ divides~$|L|$.
    We set~$U' = U^1_1(\interspace{L}{S})$ and~$U'' = U^1_2(\interspace{L}{S})$.
    We set~$V' \coloneqq \{\ell\}$.
    There are only size-$2$ fibers in~$\fibers{\coherentConfig_\ell[S]}$, and thus~$t_S \leq \tau(0,0,-6) \leq -0.3$.
    In~$\fibers{\coherentConfig_\ell[L]}$, there are the following fibers (or unions of fibers):
    $\{v \in L \mid vU' = \ell U' \}$,
    $\{v \in L \mid vU' = \ell U''\}$, and
    $\{v \in L \mid vU' \not\subseteq \ell U' \cup \ell U''\}$,
    all of which have size~$\frac{\abs{L}}{3}$.
    By Lemma~\ref{local:progress-in-large/lem} we obtain~$t_L \leq -0.8$.
    We conclude that~$\tau(\parameters(\coherentConfig_\ell)) \leq \tau(\parameters(\coherentConfig))- 1.1$.

    (\textit{Case~$\ipsixMatchingAndComplement$}).
    Assume that the interspace~$\interspace{L}{S}$ has the interspace pattern~$\ipsixMatchingAndComplement$.
    This implies that~$12$ divides~$|L|$.
    Let~$U' = U^1_1(\interspace{L}{S})$ and~$U'' = U^1_2(\interspace{L}{S})$.
    We set~$V' \coloneqq \{\ell\}$.
    There are only size-$2$ fibers in~$\fibers{\coherentConfig_\ell[S]}$ and thus~$t_S \leq \tau(0,0,-6) \leq -0.3$.
    In~$\fibers{\coherentConfig_\ell[L]}$, there are the following fibers (or unions of fibers):
    $\{v \in L \mid vU' = \ell U'\wedge v U'' = \ell U''\}$ and
    $\{v \in L \mid vU' = \ell U'\wedge v U '' \cap \ell U'' = \emptyset\}$, both of which have size~$\frac{\abs{L}}{12}$, as well as
    $\{v \in L \mid vU' = \ell U'\wedge \abs{v U'' \cap \ell U''} = 1\}$, which has size~$\frac{\abs{L}}{6}$, as well as
    $\{v \in L \mid vU' \neq \ell U''\wedge \abs{v U'' \cap \ell U''} = 1\}$ and
    $\{v \in L \mid vU' \neq \ell U''\wedge v U'' \cap \ell U'' = \emptyset\}$, both of which have size~$\frac{|L|}{3}$.
    By Lemma~\ref{local:progress-in-large/lem} we obtain~$t_L \leq -1$.
    We conclude that~$\tau(\parameters(\coherentConfig_\ell)) \leq \tau(\parameters(\coherentConfig))- 1.3$.

    (\textit{Case~$\ipsixMatchingComplement$}).
    Assume that the interspace~$\interspace{L}{S}$ has the interspace pattern~$\ipsixMatchingComplement$.
    We set~$V' \coloneqq \{\ell\}$.
    Let~$U' = U^1_1(\interspace{L}{S})$.
    There are two size-$3$ fibers in~$\fibers{\coherentConfig_\ell[S]}$  and thus~$t_S \leq \tau(0,0,-6) \leq -0.3$.
    In~$\fibers{\coherentConfig_\ell[L]}$, there are the following fibers (or unions of fibers):
    $\{v \in L \mid vU' = \ell U'\}$ and
    $\{v \in L \mid vU' \cap \ell U' = \emptyset\}$, both of which have size~$\frac{\abs{L}}{8}$,
    $\{v \in L \mid \abs{vU' \cap \ell U'} = 1\}$ and
    $\{v \in L \mid \abs{vU' \cap \ell U'} = 2\}$, both of which have size~$\frac{3\abs{L}}{8}$.
    By Lemma~\ref{local:progress-in-large/lem} we obtain~$t_L \leq -0.7$.
    We conclude that~$\tau(\parameters(\coherentConfig_\ell)) \leq \tau(\parameters(\coherentConfig))- 1$.

    (\textit{Case~$\ipsixTriangleComplementTwice$}).
    Assume that the interspace~$\interspace{L}{S}$ has the interspace pattern~$\ipsixTriangleComplementTwice$.
    This implies that~$9$ divides~$|L|$.
    Let~$U' = U^1_1(\interspace{L}{S})$ and~$U'' = U^2_1(\interspace{L}{S})$.
    We set~$V' \coloneqq \{\ell\}$.
    There are three size-$2$ fibers in~$\fibers{\coherentConfig_\ell[S]}$, and thus~$t_S \leq \tau(0,0,-6) \leq -0.3$.
    Let~$S' \coloneqq S \setminus (\ell U' \cup \ell U'')$.
    In~$\fibers{\coherentConfig_\ell[L]}$, there are the following fibers (or unions of fibers):
    $\{v \in L \mid v U' = \ell U',  v U'' \in \{\ell U'', S'      \} \}$,
    $\{v \in L \mid v U' = \ell U'', v U'' \in \{\ell U',  S'      \} \}$, and
    $\{v \in L \mid v U' = S',       v U'' \in \{\ell U',  \ell U''\} \}$, all three of which have~$\frac{\abs{L}}{9}$, as well as
    $\{v \in L \mid v U' = \ell U',  v U'' \notin \{\ell U'', S'      \} \}$,
    $\{v \in L \mid v U' = \ell U'', v U'' \notin \{\ell U',  S'      \} \}$, and
    $\{v \in L \mid v U' = S',       v U'' \notin \{\ell U',  \ell U''\} \}$, all three  of which have~$\frac{2\abs{L}}{9}$.
    By Lemma~\ref{local:progress-in-large/lem} we obtain~$t_L \leq -0.95$.
    We conclude that~$\tau(\parameters(\coherentConfig_\ell)) \leq \tau(\parameters(\coherentConfig))- 1.25$.
\end{proof}

\begin{theorem}
\label{local:S-L-S/thm}
    Let~$\coherentConfig$ be a critical coherent configuration and let~$(S,L,S')$ be a path in~$\quotientGraph{\coherentConfig}$ with~$L$ large and~$S,S'$ small.
    If~$\interspace{L}{S}$ has the interspace pattern~$\ipsixMatchingComplementD$ and~$\interspace{L}{S'}$ has one of the interspace patterns~$\ipsixMatchingComplementD$, $\ipfourCycle$, $\ipsixMatching$, $\ipfourMatching$, or~$\ipsixTriangle$,
    then~$\wldim{\coherentConfig} \leq 1 + \widetilde{\f}( \tau(\parameters(\coherentConfig)) - 1.1)$.
\end{theorem}
\begin{proof}
    Let~$\ell \in L$ be arbitrary and let~$U = U^1_1(\interspace{L}{S})$.
    Let~$Z$ be a union of fibers of~$\coherentConfig$.
    We define~$(a_Z,b_Z,c_Z) \coloneqq \parameters(\coherentConfig_{\ell}[Z]) - \parameters(\coherentConfig[Z])$ and~$t_Z \coloneqq \tau(a_Z,b_Z,c_Z)$.

    (\textit{Case~$\ipsixMatchingComplementD$ and~$\ipfourCycle$}).
    Assume that the interspace~$\interspace{L}{S'}$ has the interspace pattern~$\ipfourCycle$, and let~$U' = U^1_1(\interspace{L}{S'})$.
    Note that there are only tiny fibers in~$\fibers{\coherentConfig_\ell[S \cup S']}$.
    Thus~$(a_{S \cup S'},b_{S \cup S'}, c_{S \cup S'}) \preceq (0,0,-10)$ and~$t_{S \cup S'} \leq - 0.5$.
    Now consider~$\coherentConfig_\ell[L]$:
    with respect to~$S'$, the fiber~$L$ splits into the fibers (or unions of fibers)~$L'_1 \coloneqq \{v \in L \mid |v U' \cap \ell U'| \in \{0,2\}\}$
    and~$L'_2 \coloneqq \{ v \in L \mid |v U' \cap \ell U'| = 1\}$,
    both of which have size~$\frac{\abs{L}}{2}$.
    With respect to~$S $, the fiber~$L$ splits into fibers (or unions of fibers)~$L_1 \coloneqq \{v \in L \mid vU = \ell U\}$, which has size~$\frac{\abs{L}}{4}$, and~$L_2 \coloneqq \{v \in L \mid \abs{vU \cap \ell U} = 1\}$, which has size~$\frac{3\abs{L}}{4}$.

    If~$L_1 \not\subseteq L'_i$ for both~$i \in \{1,2\}$, then~$\equivalenceClasses{L,S}$ and~$\equivalenceClasses{L,S'}$ are fully intersecting.
    Hence, fiber~$L$ splits into two fibers of size~$\frac{3\abs{L}}{8}$ and two fibers of size~$\frac{\abs{L}}{8}$.
    By Lemma~\ref{local:progress-in-large/lem} we obtain~$t_L \leq -0.7$.
    We conclude that~$\tau(\parameters(\coherentConfig_\ell)) \leq \tau(\parameters(\coherentConfig))- 1.2$.
    If~$L_1 \subseteq L'_i$ for some~$i \in \{1,2\}$, then~$L$ splits into three fibers of size~$\frac{\abs{L}}{2}$, $\frac{\abs{L}}{4}$,and  $\frac{\abs{L}}{4}$ respectively.
    By Lemma~\ref{local:progress-in-large/lem} we obtain~$t_L \leq -0.6$.
    Overall, we obtain~$\tau(\parameters(\coherentConfig_\ell)) \leq \tau(\parameters(\coherentConfig))- 1.1$.

    (\textit{Case~$\ipsixMatchingComplementD$ and~$\ipsixMatching$}).
    Assume that the interspace~$\interspace{L}{S'}$ has the interspace pattern~$\ipsixMatching$.
    Due to Lemma~\ref{global-argument:partition:fully-intersecting/lem} the partitions~$\equivalenceClasses{L,S}$ and~$\equivalenceClasses{L,S'}$ are fully intersecting.
    This implies that~$12$ divides~$|L|$.
    Let~$U' = U^1_1(\interspace{L}{S'})$.
    There are only one size-$4$ fiber (or union of fibers) and multiple tiny fibers in~$\fibers{\coherentConfig_\ell[S \cup S']}$.
    Thus~$(a_{S \cup S'},b_{S \cup S'}, c_{S \cup S'}) \preceq (0,0,-8)$ and~$t_{S \cup S'} \leq - 0.4$.
    There are the following fibers (or unions of fibers) in~$\fibers{\coherentConfig_\ell[L]}$:
    $\{v \in L \mid v U' = \ell U'                    ,v U =    \ell U  \}$, which has size~$\frac{\abs{L}}{12}$,
    $\{v \in L \mid v U' = \ell U'                    ,v U \neq \ell U  \}$, which has size~$\frac{3\abs{L}}{12}$,
    $\{v \in L \mid v U' \cap \ell U' = \emptyset     ,v U =    \ell U  \}$, which has size~$\frac{2\abs{L}}{12}$, and
    $\{v \in L \mid v U' \cap \ell U' = \emptyset     ,v U \neq \ell U  \}$, which has size~$\frac{6\abs{L}}{12}$.
    By Lemma~\ref{local:progress-in-large/lem} we obtain~$t_L \leq -0.75$.
    We conclude that~$\tau(\parameters(\coherentConfig_\ell)) \leq \tau(\parameters(\coherentConfig))- 1.15$.

    (\textit{Case~$\ipsixMatchingComplementD$ and~$\ipsixMatchingComplementD$}).
    Assume that the interspace~$\interspace{L}{S'}$ has the interspace pattern~$\ipsixMatchingComplementD$ and let~$U' = U^1_1(\interspace{L}{S'})$.
    There are only tiny fibers in~$\fibers{\coherentConfig_\ell[S \cup S']}$.
    We conclude that $(a_{S \cup S'},b_{S \cup S'}, c_{S \cup S'}) \preceq (0,0,-12)$ and~$t_{S \cup S'} \leq - 0.6$.

    Towards a contradiction, assume~$\equivalenceClasses{L,S} = \equivalenceClasses{L,S'}$.
    Observe that for all~$P,Q \in \equivalenceClasses{L,S}$ there is exactly one~$s \in S$ such that~$sU^\star = P \cup Q$ and exactly one~$s' \in S'$ such that~$s'U'^\star = P \cup Q$.
    Thus~$\matching{6} \in \interspace{S}{S'}$, a contradiction.

    Assume that~$\equivalenceClasses{L,S} \neq \equivalenceClasses{L,S'}$ but also that~$\equivalenceClasses{L,S}$ and~$\equivalenceClasses{L, S'}$ are not fully intersecting.
    By coherence, for all~$P \in \equivalenceClasses{L,S}$ and~$P' \in \equivalenceClasses{L,S'}$ their intersection~$P \cap P'$ has size~$|L|/8$ or it is empty.
    We show that~$S$ is union of fibers, which contradicts the assumptions, as follows.
    Suppose that there are~$\{P,Q\} \subseteq \equivalenceClasses{L,S}$ and~$\{P',Q'\} \subseteq \equivalenceClasses{L,S'}$ such that~$P \cup Q = P' \cup Q'$.
    Let~$s_1 \in S$ and~$s'_1 \in S'$ such that~$s_1 U^\star = P \cup Q = s'_1 U'^\star$, and thus~$|s_1 U^\star \cap s'_1 U'^\star| = |L|/2$.
    However there is also~$s_2 \in S$ such that~$(P' \cup Q') \cap s_2 U^\star = P$.
    Hence all for all~$s'_2 \in S'$ we have~$|s_2 U^\star \cap s'_2 U'^\star| = |L|/4$.
    By coherence, vertices~$s_1$ and~$s_2$ are not elements of the same fiber.
    Suppose that there are~$\{P,T,Q\} \subseteq \equivalenceClasses{L,S}$ and~$\{P',Q'\} \subseteq \equivalenceClasses{L,S'}$ such that~$T \subseteq P' \cup Q'$ and~$|P \cap P'| = |Q \cap Q'| = |L|/4$.
    Let~$s_1,s_2 \in S$ and~$s'_1 \in S'$ such that~$s_1 U^\star = P \cup T$,~$s_2 U^\star = P \cup Q$, and~$s'_1 U'^\star = P' \cup Q'$.
    Observe that~$|s_1 U^\star \cap s'_1 U'^\star| = 3|L|/8$.
    However for all~$s'_2 \in S'$ we have~$|s_2 U^\star \cap s'_2 U'^\star| = |L|/4$.
    By coherence, vertices~$s_1$ and~$s_2$ are not elements of the same fiber.

    Assume~$\equivalenceClasses{L,S}$ and~$\equivalenceClasses{L, S'}$ are fully intersecting.
    This implies that~$16$ divides~$|L|$.
    There are the following fibers (or unions of fibers) in~$\coherentConfig_\ell[L]$:
    $\{v \in L \mid vU' =    \ell U'\wedge vU =    \ell U \}$, which has size~$\frac{\abs{L}}{16}$,
    $\{v \in L \mid vU' \neq \ell U'\wedge vU \neq \ell U \}$, which has size~$\frac{9\abs{L}}{16}$, as well as
    $\{v \in L \mid vU' \neq \ell U'\wedge vU =    \ell U \}$ and
    $\{v \in L \mid vU' =    \ell U'\wedge vU \neq \ell U \}$, both of which have size~$\frac{3\abs{L}}{16}$.
    By Lemma~\ref{local:progress-in-large/lem} we obtain~$t_L \leq -0.75$.
    We conclude that~$\tau(\parameters(\coherentConfig_\ell)) \leq \tau(\parameters(\coherentConfig))- 1.35$.

    (\textit{Case~$\ipsixMatchingComplementD$ and~$\ipfourMatching$}).
    Assume that the interspace~$\interspace{L}{S'}$ has the interspace pattern~$\ipfourMatching$ and let~$U' = U^1_1(\interspace{L}{S'})$.

    Towards a contradiction, assume that~$\equivalenceClasses{L,S}$ and~$\equivalenceClasses{L, S'}$ are not fully intersecting.
    By coherence, for all~$P \in \equivalenceClasses{L,S}$ and~$P' \in \equivalenceClasses{L,S'}$ their intersection~$P \cap P'$ has the same size or is empty.
    Therefore there are~$\{P, Q\} = \equivalenceClasses{L,S}$ and~$P' \in \equivalenceClasses{L,S'}$ such that~$P \cup Q = P'$.
    There is a vertex~$s_1 \in S$ and~$s'_1 \in S'$ such that~$s_1 U^\star = P'$ and~$s'_1 U'^\star = P'$ respectively.
    Observe that~$|s_1 U^\star \cap s'_1 U'^\star| = |L|/2$.
    Further there is also a vertex~$s_2 \in S$ such that~$s_2 U^\star \cap P' = P$.
    Thus for all~$s'_2 \in S'$ we have~$|s_2 U^\star \cap s'_2 U'^\star| = |L|/4$.
    By coherence, vertices~$s_1$ and~$s_2$ are not elements of the same fiber.
    This contradicts~$S$ being a fiber.

    We assume that~$\equivalenceClasses{L,S}$ and~$\equivalenceClasses{L, S'}$ are fully intersecting.
    This implies that~$16$ divides~$|L|$.
    There are only tiny fibers in~$\fibers{\coherentConfig_\ell[S \cup S']}$.
    Thus~$(a_{S \cup S'},b_{S \cup S'}, c_{S \cup S'}) \preceq (0,0,-10)$ and~$t_{S \cup S'} \leq - 0.5$.
    Now consider the fibers (or unions of fibers) in~$\coherentConfig_\ell[L]$:
    $\{ v\in L \mid v U =    \ell U, vU' =    \ell U'\}$ and
    $\{ v\in L \mid v U =    \ell U, vU' \neq \ell U'\}$, both of which have size~$\frac{\abs{L}}{8}$, as well as
    $\{ v\in L \mid v U \neq \ell U, vU' =    \ell U'\}$ and
    $\{ v\in L \mid v U \neq \ell U, vU' \neq \ell U'\}$, both of which have size~$\frac{3\abs{L}}{8}$.
    By Lemma~\ref{local:progress-in-large/lem} we obtain~$t_L \leq -0.7$.
    We conclude that~$\tau(\parameters(\coherentConfig_\ell)) \leq \tau(\parameters(\coherentConfig))- 1.2$.

    (\textit{Case~$\ipsixMatchingComplementD$ and~$\ipsixTriangle$}).
    Assuming that the interspace~$\interspace{L}{S'}$ has the interspace pattern~$\ipsixTriangle$, the proof follows the same steps as the previous case.
\end{proof}

\begin{theorem}
\label{local:S-L-S:rest/thm}
    Let~$\coherentConfig$ be a critical coherent configuration, and let~$(S,L,S')$ be a path in~$\quotientGraph{\coherentConfig}$ with~$L$ large and~$S,S'$ small.
    If~$\{S\}$ and~$\{S'\}$ are not dominating and both interspaces~$\interspace{L}{S}$ and~$\interspace{L}{S'}$ have one of the interspace patterns~$\ipfourCycle$, $\ipsixMatching$, $\ipfourMatching$, or~$\ipsixTriangle$,
    then~$\wldim{\coherentConfig} \leq 1 + \widetilde{\f}( \tau(\parameters(\coherentConfig)) - 1)$.
\end{theorem}
\begin{proof}
    Assume that Theorems~\ref{local:L-S/thm} and~\ref{local:S-L-S/thm} are not applicable.
    Let~$\ell \in L$ be arbitrary, and let~$U = U^1_1(\interspace{L}{S})$.
    Let~$Z$ be a union of fibers of~$\coherentConfig$ and~$V' \subseteq \vertices(\coherentConfig)$.
    We define~$(a_Z,b_Z,c_Z) \coloneqq \parameters(\coherentConfig_{V'}[Z]) - \parameters(\coherentConfig[Z])$ and~$t_Z \coloneqq \tau(a_Z,b_Z,c_Z)$.

    For the first four cases, assume that~$\interspace{L}{S}$ has the interspace pattern~$\ipfourCycle$.

    (\textit{Case~$\ipfourCycle$ and~$\ipfourCycle$}).
    Assume that~$\interspace{L}{S'}$ has the interspace pattern~$\ipfourCycle$ and let~$U' = U^1_1(\interspace{L}{S'})$.
    We set~$V' \coloneqq \{\ell\}$ and examine~$\coherentConfig_\ell$.
    There are only tiny fibers in~$\fibers{\coherentConfig_\ell[S \cup S']}$.
    Thus~$(a_{S \cup S'},b_{S \cup S'}, c_{S \cup S'}) \preceq (0,0,-8)$ and~$t_{S \cup S'} \leq - 0.4$.

    If~$\equivalenceClasses{L,S} = \equivalenceClasses{L,S'}$, then~$\matching{4}\in\interspace{S}{S'}$, a contradiction.

    Assume that there is a~$P \in \equivalenceClasses{L,S}$ and distinct~$P',Q' \in \equivalenceClasses{L,S'}$ such that~$\abs{P \cap P'} = \abs{P \cap Q'} = \frac{\abs{L}}{8}$.
    Then there are the following fibers (or unions of fibers) in~$\coherentConfig_\ell[L]$:
    $P \cap P'$,
    $P \cap Q'$,
    $P' \setminus P$, and
    $Q' \setminus P$, all four of which have size~$\frac{\abs{L}}{8}$, as well as
    $L \setminus (P' \cup Q')$, which has size~$\frac{\abs{L}}{2}$.
    By Lemma~\ref{local:progress-in-large/lem} we obtain~$t_L \leq -0.6$.

    Assume that~$\equivalenceClasses{L,S}$ and~$\equivalenceClasses{L, S'}$ are fully intersecting.
    This implies that~$16$ divides~$|L|$.
    Furthermore, there are the following fibers (or unions of fibers) in~$\coherentConfig_\ell[L]$:
    $\{ v \in L \mid vU =               \ell U,  vU' =              \ell U' \}$,
    $\{ v \in L \mid vU = S \setminus   \ell U,  vU' =              \ell U' \}$,
    $\{ v \in L \mid vU =               \ell U,  vU' = S' \setminus \ell U' \}$, and
    $\{ v \in L \mid vU = S \setminus   \ell U,  vU' = S' \setminus \ell U' \}$, all four of which have size~$\frac{\abs{L}}{16}$, as well as
    $\{ v \in L \mid \abs{vU \cap \ell U} = 1 ,  vU' =              \ell U' \}$,
    $\{ v \in L \mid \abs{vU \cap \ell U} = 1 ,  vU' = S' \setminus \ell U' \}$,
    $\{ v \in L \mid vU =               \ell U,  \abs{vU' \cap \ell U'} = 1 \}$, and
    $\{ v \in L \mid vU = S \setminus   \ell U,  \abs{vU' \cap \ell U'} = 1 \}$, all four of which have size~$\frac{2\abs{L}}{16}$, and finally
    $\{ v \in L \mid \abs{vU \cap \ell U} = 1 ,  \abs{vU' \cap \ell U'} = 1 \}$, which has size~$\frac{4\abs{L}}{16}$.
    By Lemma~\ref{local:progress-in-large/lem} we obtain~$t_L \leq -1.4$.

    Overall, we obtain~$\tau(\parameters(\coherentConfig_\ell)) \leq \tau(\parameters(\coherentConfig))- 1.0$.

    (\textit{Case~$\ipfourCycle$ and~$\ipsixMatching$}).
    Assume that the interspace~$\interspace{L}{S'}$ has the interspace pattern~$\ipsixMatching$ and let~$U' = U^1_1(\interspace{L}{S'})$.
    We set~$V' \coloneqq \{\ell\}$ and examine~$\coherentConfig_\ell$.
    Apart from tiny fibers there is at most one size-$4$ fiber (or union of fibers) in~$\fibers{\coherentConfig_\ell[S \cup S']}$.
    Thus~$(a_{S \cup S'},b_{S \cup S'}, c_{S \cup S'}) \preceq (0,0,-6)$ and~$t_{S \cup S'} \leq - 0.3$.
    Due to Lemma~\ref{global-argument:partition:fully-intersecting/lem}, there are the following fibers (or unions of fibers) in~$\fibers{\coherentConfig_\ell[L]}$:
    $\{ v \in L \mid vU =               \ell U,  vU' =    \ell U' \}$ and
    $\{ v \in L \mid vU = S \setminus   \ell U,  vU' =    \ell U' \}$, both of which have size~$\frac{\abs{L}}{12}$, as well as
    $\{ v \in L \mid \abs{vU \cap \ell U} = 1 ,  vU' =    \ell U' \}$,
    $\{ v \in L \mid vU =               \ell U,  vU' \neq \ell U' \}$,
    $\{ v \in L \mid vU = S \setminus   \ell U,  vU' \neq \ell U' \}$, all three of which have size~$\frac{2\abs{L}}{12}$, and finally
    $\{ v \in L \mid \abs{vU \cap \ell U} = 1 ,  vU' \neq \ell U' \}$, which has size~$\frac{4\abs{L}}{12}$.
    By Lemma~\ref{local:progress-in-large/lem} we obtain~$t_L \leq -1.2$.
    Overall, we obtain~$\tau(\parameters(\coherentConfig_\ell)) \leq \tau(\parameters(\coherentConfig))- 1.5$.

    (\textit{Case~$\ipfourCycle$ and~$\ipfourMatching$}).
    Assume that interspace~$\interspace{L}{S'}$ has the interspace pattern~$\ipfourMatching$ and let~$U' = U^1_1(\interspace{L}{S'})$.
    We set~$V' \coloneqq \{\ell\}$ and examine~$\coherentConfig_\ell$.
    There are only tiny fibers in~$\fibers{\coherentConfig_\ell[S \cup S']}$.
    Thus~$(a_{S \cup S'},b_{S \cup S'}, c_{S \cup S'}) \preceq (0,0,-8)$ and~$t_{S \cup S'} \leq - 0.4$.

    Assume that~$\equivalenceClasses{L,S}$ and~$\equivalenceClasses{L, S'}$ are not fully intersecting.
    Since~$\equivalenceClasses{L,S}$ is an equipartition, for each~$P \in \equivalenceClasses{L,S}$ there is~$P' \in \equivalenceClasses{L,S'}$ such that~$P \subseteq P'$.
    In this case the following fibers (or unions of fibers) occur in~$\fibers{\coherentConfig_\ell[L]}$:
    $\{ v \in L \mid v U' =    \ell  U'\wedge v U =             \ell U \}$ and
    $\{ v \in L \mid v U' =    \ell  U'\wedge v U = S \setminus \ell U \}$, both of which have size~$\frac{\abs{L}}{4}$, as well as
    $\{ v \in L \mid v U' \neq \ell  U' \}$, which has size~$\frac{\abs{L}}{2}$.
    By Lemma~\ref{local:progress-in-large/lem} we obtain~$t_L \leq -0.6$.

    Now assume that~$\equivalenceClasses{L,S}$ and~$\equivalenceClasses{L, S'}$ are fully intersecting.
    The following fibers (or unions of fibers) appear in~$\fibers{\coherentConfig_\ell[L]}$:
    $\{ v \in L \mid vU =               \ell U\wedge  vU' =    \ell U' \}$,
    $\{ v \in L \mid vU = S \setminus   \ell U\wedge  vU' =    \ell U' \}$,
    $\{ v \in L \mid vU =               \ell U\wedge  vU' \neq \ell U' \}$, and
    $\{ v \in L \mid vU = S \setminus   \ell U\wedge  vU' \neq \ell U' \}$, all of which have size~$\frac{\abs{L}}{8}$, as well as
    $\{ v \in L \mid \abs{vU \cap \ell U} = 1 \wedge  vU' =    \ell U' \}$ and
    $\{ v \in L \mid \abs{vU \cap \ell U} = 1 \wedge  vU' \neq \ell U' \}$, both of which have size~$\frac{\abs{L}}{4}$.
    By Lemma~\ref{local:progress-in-large/lem} we obtain~$t_L \leq -0.8$.

    Overall, we obtain~$\tau(\parameters(\coherentConfig_\ell)) \leq \tau(\parameters(\coherentConfig))- 1$.

    (\textit{Case~$\ipfourCycle$ and $\ipsixTriangle$}).
    Assuming that the interspace~$\interspace{L}{S'}$ has the interspace pattern~$\ipsixTriangle$, the proof follows the same steps as the previous case.

    For the next three cases, assume that~$\interspace{L}{S'}$ has the interspace pattern~$\ipsixMatching$.
    This implies that~$3$ divides~$|L|$.

    (\textit{Case~$\ipsixMatching$ and~$\ipsixMatching$}).
    Assume that the interspace~$\interspace{L}{S'}$ has the interspace pattern~$\ipsixMatching$.
    Let~$U' = U^1_1(\interspace{L}{S'})$.

    Assume that~$\equivalenceClasses{L,S} = \equivalenceClasses{L,S'}$.
    Thus~$\disjointCliques{3}{2,2} \in \interspace{S}{S'}$.
    We set~$V' \coloneqq \{s\}$ where~$s \in S$ and consider~$\coherentConfig_s$:
    the fiber~$L$ splits into two fibers which have sizes~$|L|/3$ and~$2|L|/3$ respectively.
    By Lemma~\ref{local:progress-in-large/lem} we obtain~$t_L \leq -0.4$.
    Since~$\{S\}$ is not dominating, by Lemma~\ref{critical:6-cc:restorable:DUC:deg1/lem} there is a fiber~$R$ adjacent to~$S$ in~$\quotientGraph{\coherentConfig}$ such that~$\interspace{R}{S}$ neither has the interspace pattern~$\ipsixMatching$ nor~$\ipsixMatchingTwice$.
    If~$R$ is large, then by Theorem~\ref{local:L-S/thm} the interspace~$\interspace{R}{S}$ has the interspace pattern~$\ipsixMatchingComplementD$ or $\ipsixTriangle$.
    If~$R$ is small, then by Lemma~\ref{critical:4cc-6cc/lem} interspace~$\interspace{R}{S}$ does not have the interspace pattern~$\ipsixMatchingComplementD$ since~$\{S\}$ is not dominating.
    Therefore, if~$R$ is small, then~$\interspace{R}{S}$ has the interspace pattern~$\ipsixMatchingAndCycle$, $\ipsixMatchingMatching$, or $\ipsixTriangle$.
    \begin{itemize}
        \item
        If~$\interspace{R}{S}$ has the interspace pattern~$\ipsixMatchingComplementD$, then~$R$ is large and~$R$ splits into two fibers of equal size in~$\coherentConfig_s$.
        By Lemma~\ref{local:progress-in-large/lem} we obtain~$t_R \leq -0.4$.
        Fibers~$S$ and~$S'$ each split into a size-$2$ and a size-$4$ fiber and thus~$t_{S \cup S'} \leq -0.2$.
        Together~$t_{S \cup S' \cup R} \leq -0.6$.

        \item
        If~$\interspace{R}{S}$ has the interspace pattern~$\ipsixTriangle$, then~$\disjointCliques{2}{3} \in \inducedCC{S}$.
        Hence fiber~$S$ splits entirely into tiny fibers while~$S'$ splits into a size-$2$ and a size-$4$ fiber and thus~$t_{S \cup S'} \leq -0.4$.
        Fiber~$R$ splits into fibers of equal size in~$\coherentConfig_s$.
        Thus by Lemma~\ref{local:progress-in-large/lem} we obtain~$t_R \leq -0.4$ if~$R$ is large, and~$t_R \leq \tau(0,0,-4) \leq -0.2$ if~$R$ is small.
        Together~$t_{S \cup S' \cup R} \leq -0.6$.

        \item
        If~$\interspace{R}{S}$ has the interspace pattern~$\ipsixMatchingMatching$, then~$R$ is small and there are three constituents in~$\inducedCC{S}$ isomorphic to~$\disjointCliques{3}{2}$.
        Hence fiber~$S$ splits entirely into singletons while~$S'$ and~$R$ split entirely into tiny fibers.
        Thus~$t_{S \cup S' \cup R} \leq \tau(0,0,-18) \leq -0.9$.

        \item
        If~$\interspace{R}{S}$ has the interspace pattern~$\ipsixMatchingAndCycle$, then~$R$ is small and~$\cycle{6} \in \inducedCC{S}$.
        Hence fibers~$S$ and~$R$ split entirely into tiny fibers while~$S'$ splits into a size-$2$ and a size-$4$ fiber
        Thus~$t_{S \cup S' \cup R} \leq \tau(0,0,-14) \leq -0.7$.
    \end{itemize}
    We conclude~$t_{L \cup S \cup S' \cup R} \leq -1$ and thus~$\tau(\parameters(\coherentConfig_s)) \leq \tau(\parameters(\coherentConfig))- 1.15$.

    If~$\equivalenceClasses{L,S}$ and~$\equivalenceClasses{L, S'}$ are fully intersecting, then~$9$ divides~$|L|$.

    Assume that~$\equivalenceClasses{L,S}$ and~$\equivalenceClasses{L, S'}$ are fully intersecting and~$|L| \geq 18$.
    We set~$V' \coloneqq \{\ell\}$ and examine~$\coherentConfig_\ell$.
    Besides tiny fibers, there are at most two size-$4$ fibers (or unions of fibers) in~$\fibers{\coherentConfig_\ell[S \cup S']}$.
    Thus~$(a_{S \cup S'},b_{S \cup S'}, c_{S \cup S'}) \preceq (0,0,-4)$ and~$t_{S \cup S'} \leq - 0.2$.
    There are the following fibers (or unions of fibers) in~$\fibers{\coherentConfig_\ell[L]}$:
    $ \{ v \in L \mid v U =    \ell U, v U' =    \ell U'\}$, which has size~$\frac{\abs{L}}{9}$,
    $ \{ v \in L \mid v U =    \ell U, v U' \neq \ell U'\}$ and
    $ \{ v \in L \mid v U \neq \ell U, v U' =    \ell U'\}$, both of which have size~$\frac{2\abs{L}}{9}$, as well as
    $ \{ v \in L \mid v U \neq \ell U, v U' \neq \ell U'\}$, which has size~$\frac{4\abs{L}}{9}$.
    \begin{itemize}
        \item
        If~$18 \leq \abs{L} < 36$, then~$L$ splits entirely into one large fiber and small or tiny fibers.
        Furthermore~$\frac{\abs{L}}{9}$ of the vertices of~$L$ end up in tiny fibers while~$\frac{4\abs{L}}{9}$ of the vertices of~$L$ form small fibers.
        Thus~$(a_L,b_L, c_L) \preceq (-\frac{5\abs{L}}{9},0,\frac{4\abs{L}}{9}) \preceq (-10,0,8)$ and~$t_L \leq - 1.1$.
        \item
        If~$36 \leq \abs{L} < 72$, then $L$ splits into at least three large fibers while~$\frac{\abs{L}}{9}$ of the vertices of~$L$ end up small fibers.
        Hence~$(a_L,b_L,c_L) \preceq (-\frac{\abs{L}}{9},2,\frac{\abs{L}}{9}) \preceq (-4,2,4)$ and~$t_L \leq - 1.2$.
        \item
        If~$72 \leq \abs{L}$, then $L$ splits into at least four large fibers.
        Hence~$(a_L,b_L,c_L) \preceq (0,3,0) \preceq (0,3,0)$ and~$t_L \leq - 1.2$.
    \end{itemize}
    Together with~$t_{S \cup S'} \leq - 0.2$, we obtain~$\tau(\parameters(\coherentConfig_\ell)) \leq \tau(\parameters(\coherentConfig))- 1.05$.

    Assume that~$\equivalenceClasses{L,S}$ and~$\equivalenceClasses{L, S'}$ are fully intersecting and~$|L| = 9$.
    Again we set~$V' \coloneqq \{\ell\}$ and examine~$\coherentConfig_\ell$.
    Then~$L$ splits into fibers of sizes~$1$, $2$, $2$, and~$4$.
    We refer to the last of those fibers as~$L'$.
    We distinguish cases according to how~$L'$ splits.
    \begin{itemize}
        \item
        In our first case, vertex set~$L'$ in~$\coherentConfig_\ell[L]$ is a union of tiny fibers.
        This yields $(a_L,b_L, c_L) \preceq (-9,-1,0)$ and~$t_L \leq - 0.95$.
        Together with~$t_{S \cup S'} \leq - 0.2$, we obtain~$\tau(\parameters(\coherentConfig_\ell)) \leq \tau(\parameters(\coherentConfig))- 1.15$.
        \item
        Assume~$L'$ is a small fiber in~$\coherentConfig_\ell[L]$ but~$\colorDeg{L} > 2$.
        Then there is~$R \in \fibers{\coherentConfig} \setminus \{S,S'\}$ adjacent to~$L$ in~$\quotientGraph{\coherentConfig}$.
        Thus at least two vertices split from~$R$ in~$\coherentConfig_\ell$.
        If fiber~$R$ is large, we have~$(a_{L\cup R},b_{L\cup R}, c_{L\cup R}) \preceq (-11,-1,4)$ and~$t_{L\cup R} \leq - 1.05$.
        If fiber~$R$ is small, we have~$(a_{L\cup R},b_{L\cup R}, c_{L\cup R}) \preceq (-9,-1,2)$ and~$t_{L\cup R} \leq - 0.85$.
        Together with~$t_{S \cup S'} \leq - 0.2$, we obtain~$\tau(\parameters(\coherentConfig_\ell)) \leq \tau(\parameters(\coherentConfig))- 1.05$.
        \item
        Assume~$L'$ is a small fiber in~$\coherentConfig_\ell[L]$ and~$\colorDeg{L} = 2$.
        Then by Lemma~\ref{restorable:two-3K2,2-ip:fully-intersecting/lem} the set~$\{L\}$ is dominating and thus~$L \cup S \cup S' = \vertices(\coherentConfig)$.
        Recall that in~$\coherentConfig_\ell$ fibers~$S$ and~$S'$ each split into size-$2$ fiber and a size-$4$ fiber. We refer to the size~$4$ into which~$S$ (respectively~$S'$) splits by~$\overline{S}$ (respectively~$\overline{S'}$).
        By restoring criticality, we obtain a critical coherent configuration~$\coherentConfig'$ with~$|\vertices(\coherentConfig')| = L' \cup \overline{S} \cup \overline{S'}$.

        We show that~$\wldim{\coherentConfig'} \leq 2$ as follows.
        Observe that there is~$\overline{U} \in \coherentConfig'[\overline{S},L']$ and~$\overline{U'} \in \coherentConfig'[\overline{S'},L']$ such that~$(\overline{S} \disjointUnion L',\overline{U})$ and~$(\overline{S'} \disjointUnion L',\overline{U'})$ are isomorphic to~$\disjointCliques{2}{2,2}$.
        Since~$\equivalenceClasses{L,S}$ and~$\equivalenceClasses{L,S'}$ are fully intersecting, for every pair~$(s,s') \in \overline{S} \times \overline{S'}$ there is exactly one vertex~$\ell$ in~$L'$ such that~$\{\ell\} = s \overline{U} \cap s' \overline{U'}$.
        Thus~$L'$ is restorable in~$\coherentConfig'$.
        The interspace~$\coherentConfig'[\overline{S},\overline{S'}]$ either is homogeneous or it contains a constituent~$G$ isomorphic to~$\cycle{8}$ or~$\disjointCliques{2}{2,2}$.
        Since \wltwo identifies cycles and disjoint union of bipartite complete graphs,~$\wldim{\coherentConfig'} \leq 2$.

        Since~$\coherentConfig'$ has WL-dimension at most~$2$, we may replace it by a coherent configuration~$\coherentConfig''$ which has WL-dimension~$2$.
        We choose~$\coherentConfig''$ to be the homogeneous coherent configuration on~$6$ vertices that contains three constituents isomorphic to~$\cycle{6}$, $\disjointCliques{3}{2}$, and~$\disjointCliques{2}{3}$ respectively.
        Thus~$(a_{L\cup S \cup S'},b_{L\cup S \cup S'},c_{L\cup S \cup S'}) \preceq (-9,-1,-6)$ and~$t_{L\cup S \cup S'} \leq -1.25$.
        Thus~$\tau(\parameters(\coherentConfig_\ell)) \leq \tau(\parameters(\coherentConfig))- 1.25$.
    \end{itemize}

    Overall, we obtain~$\tau(\parameters(\coherentConfig_\ell)) \leq \tau(\parameters(\coherentConfig))- 1.05$.

    (\textit{Case~$\ipsixMatching$ and~$\ipfourMatching$}).
    Assume that the interspace~$\interspace{L}{S'}$ has the interspace pattern~$\ipfourMatching$.
    Let~$U' = U^1_1(\interspace{L}{S'})$.
    Thus by Lemma~\ref{global-argument:partition:fully-intersecting/lem} the partitions~$\equivalenceClasses{L,S}$ and~$\equivalenceClasses{L,S'}$ are fully intersecting.
    This implies that~$6$ divides~$|L|$.
    We set~$V' \coloneqq \{\ell\}$ and examine~$\coherentConfig_\ell$.
    Besides tiny fibers, there is at most one size-$4$ fiber (or union of fibers) in~$\fibers{\coherentConfig_\ell[S \cup S']}$.
    Thus~$(a_{S \cup S'},b_{S \cup S'}, c_{S \cup S'}) \preceq (0,0,-6)$ and~$t_{S \cup S'} \leq - 0.3$.
    There are the following fibers (or unions of fibers) in~$\fibers{\coherentConfig_\ell[L]}$:
    $\{ v \in L \mid vU =    \ell U,  vU' =    \ell U' \}$ and
    $\{ v \in L \mid vU \neq \ell U,  vU' =    \ell U' \}$, both of which have size~$\frac{\abs{L}}{6}$, as well as
    $\{ v \in L \mid vU =    \ell U,  vU' \neq \ell U' \}$,
    $\{ v \in L \mid vU \neq \ell U,  vU' \neq \ell U' \}$, both of which have size~$\frac{2\abs{L}}{6}$.
    By Lemma~\ref{local:progress-in-large/lem} we obtain~$t_L \leq -1$.
    Overall, we obtain~$\tau(\parameters(\coherentConfig_\ell)) \leq \tau(\parameters(\coherentConfig))- 1.3$.

    (\textit{Case~$\ipsixMatching$ and~$\ipsixTriangle$}).
    Assuming that the interspace~$\interspace{L}{S'}$ has the interspace pattern~$\ipsixTriangle$, the proof follows the same steps as the previous case.

    (\textit{Cases~$\interspacePattern{2K_x,x}$ and~$\interspacePattern{2K_y,y}$ where~$x,y \in \{2,3\}$}).
    For the final cases, assume that the interspace~$\interspace{L}{S}$ and~$\interspace{L}{S'}$ have the interspace pattern~$\ipfourMatching$ or interspace pattern~$\ipsixTriangle$.

    Assume that~$\equivalenceClasses{L,S} = \equivalenceClasses{L,S'}$ and that~$S$ and~$S'$ both have size~$6$.
    Thus~$\disjointCliques{2}{3,3} \in \interspace{S}{S'}$.
    We set~$V' \coloneqq \{\ell\}$ and consider~$\coherentConfig_\ell$.
    All fibers split at least in half.
    Thus by Lemma~\ref{local:progress-in-large/lem} we obtain~$t_L \leq -0.4$ and~$t_{S \cup S'} \leq -0.6$.
    Overall, we obtain~$t_{S \cup S'\cup L} \leq -1$.

    Assume that~$\equivalenceClasses{L,S} = \equivalenceClasses{L,S'}$ and that~$S$ has size~$4$.
    Thus~$\disjointCliques{2}{|S|/2,|S'|/2} \in \interspace{S}{S'}$.
    Either by Lemma~\ref{critical:4cc:restorable:2,C4/lem} there is a fiber~$R$ such that~$\interspace{R}{S}$ has the interspace pattern~$\ipfourCycle$ or by Lemma~\ref{critical:4-cc:restorable:DUC/lem} there is a fiber~$R$ such that~$\interspace{R}{S}$ has the interspace pattern~$\ipfourMatching$.
    We set~$V' \coloneqq \{s\}$ where~$s \in S$ and consider~$\coherentConfig_s$.
    All fibers adjacent to~$S$ split at least in half.
    Thus by Lemma~\ref{local:progress-in-large/lem} we obtain~$t_L \leq -0.4$, and, if~$R$ is large, $t_R \leq -0.4$.
    If~$R$ is small, we obtain~$t_R  \leq -0.2$.
    Finally, we have~$t_{S \cup S'} \leq -0.4$.
    Overall, we obtain~$t_{S \cup S' \cup R \cup L} \leq -1$.

    Assume that~$\equivalenceClasses{L, S}$ and~$\equivalenceClasses{L, S'}$ is fully intersecting.
    We set~$V' \coloneqq \{\ell\}$ and consider configuration~$\coherentConfig_\ell$.
    In~$\fibers{\coherentConfig_\ell[S \cup S']}$, there are only tiny fibers.
    It follows that $(a_{S \cup S'},b_{S \cup S'}, c_{S \cup S'}) \preceq (0,0,-8)$ and~$t_{S \cup S'} \leq - 0.4$.
    Fiber~$L$ splits into four equally sized fibers.
    By Lemma~\ref{local:progress-in-large/lem}, we obtain~$t_L \leq -0.8$.

    Overall, we obtain~$\tau(\parameters(\coherentConfig_{V'})) \leq \tau(\parameters(\coherentConfig))- 1$.
\end{proof}

Given~$F \in \fibers{\coherentConfig}$, we denote the number of large (respectively small) fibers adjacent to~$F$ in~$\quotientGraph{\coherentConfig}$ by~$\colorDegLarge{F}$ (respectively~$\colorDegSmall{F}$).

\begin{lemma}
\label{local:4cc:3neighbors/lem}
    Let~$\coherentConfig$ be a critical coherent configuration, and let~$S$ be a size-$4$ fiber of~$\coherentConfig$.
    If~
    \begin{itemize}
        \item $\colorDeg{S} = 3$ and~$\colorDegLarge{S} \geq 1$ or
        \item $\colorDeg{S} \geq 4$,
    \end{itemize}
    then~$\wldim{\coherentConfig} \leq 1 + \widetilde{\f}( \tau(\parameters(\coherentConfig)) - 1)$.
\end{lemma}
\begin{proof}
    Let~$\{B_1,\dots,B_{\colorDeg{S}}\}$ be the neighborhood of~$S$ in~$\quotientGraph{\coherentConfig}$, $i \in \{1,\dots,\colorDeg{S}\}$, and~$s \in S$.
    For a union of fibers~$Z$  of~$\coherentConfig$, set~$(a_Z,b_Z,c_Z) \coloneqq \parameters(\coherentConfig_{s}[Z]) - \parameters(\coherentConfig[Z])$ and~$t_Z \coloneqq \tau(a_Z,b_Z,c_Z)$.

    We separately consider fibers adjacent to~$S$ in~$\quotientGraph{\coherentConfig}$:
    observe that~$\interspace{S}{B_i}$ has either the interspace pattern~$\ipfourClique$ or the interspace pattern~$\ipfourCycle$ or the interspace pattern~$\ipfourMatching$.
    In~$\coherentConfig_s$ the fiber~$B_i$ splits into two fibers, both of which have size~$\frac{\abs{B_i}}{2}$.
    If~$B_i$ is large, then by Lemma~\ref{local:progress-in-large/lem} we obtain~$t_{B_i} \leq -0.4$.
    If~$B_i$ is small, then~$(a_{B_i},b_{B_i},c_{B_i}) \preceq (0,0,-4) $ and~$t_{B_i} \leq - 0.2$.
    Additionally, we have~$(a_S,b_S,c_S) \preceq (0,0,-4)$ and~$t_S \leq - 0.2$ because we individualize~$s \in S$.
    We conclude:
    if~$\colorDeg{S} = 3$ and~$\colorDegLarge{S} \geq 1$, then we obtain~$\tau(\parameters(\coherentConfig_s)) \leq \tau(\parameters(\coherentConfig))- 1$.
    If~$\colorDeg{S} \geq 4$, then we also obtain~$\tau(\parameters(\coherentConfig_s)) \leq \tau(\parameters(\coherentConfig))- 1$.
\end{proof}

\begin{lemma}
\label{local:K222-3D/lem}
    Let~$\coherentConfig$ be a critical coherent configuration, let~$(L,S)$ be an edge in~$\quotientGraph{\coherentConfig}$ such that~$L$ is large, and the interspace~$\interspace{L}{S}$ has the interspace pattern~$\ipsixMatchingComplementD$.
    If~$\colorDeg{S} \geq 3$ and~$\{L,S\}$ is not dominating, then
    \begin{itemize}
        \item $\wldim{\coherentConfig} \leq 1 + \widetilde{\f}( \tau(\parameters(\coherentConfig)) - 1)$ or
        \item $\wldim{\coherentConfig} \leq 2 + \widetilde{\f}( \tau(\parameters(\coherentConfig)) - 2.1)$
    \end{itemize}
\end{lemma}
\begin{proof}
    Assume that Theorems~\ref{local:L-S/thm} and~\ref{local:S-L-S/thm} are not applicable.
    For a union of fibers~$Z$ of~$\coherentConfig$ and~$V' \subseteq \vertices(\coherentConfig)$, set~$(a_Z,b_Z,c_Z) \coloneqq \parameters(\coherentConfig_{V'}[Z]) - \parameters(\coherentConfig[Z])$ and~$t_Z \coloneqq \tau(a_Z,b_Z,c_Z)$.
    Note that~$|L|$ is divisible by~$4$.

    Let~$R_1,R_2$ be fibers adjacent to~$S$ in~$\quotientGraph{\coherentConfig}$ other than~$L$.
    Since interspace~$\interspace{L}{S}$ has the interspace pattern~$\ipsixMatchingComplementD$, there is a constituent in~$\inducedCC{S}$ whose underlying graph is isomorphic to~$\clique{2,2,2}$.
    By Lemma~\ref{small-cc:induced-cc/lem}, we have~$|\ul(\inducedCC{S})| = 3$.
    Since Theorem~\ref{local:L-S/thm} is not applicable, there are no interspaces having the interspace pattern~$\ipsixMatchingComplement$.
    Hence for each~$i \in \{1,2\}$ the interspace~$\interspace{R_i}{S}$ has one of the following interspace patterns:~$\ipsixMatching$ and~$\ipsixMatchingComplementD$.

    We first assume that~$\abs{L}\geq 16$.
    We set~$V' \coloneqq \{s,s'\}$ where~$s,s' \in S$ are distinct vertices which are not adjacent in~$A_1(\interspace{R_1}{S})$.
    Consider~$\coherentConfig_{s,s'}$.
    We consider each neighbor~$B$ of~$S$ in~$\quotientGraph{\coherentConfig}$ independent of the others:
    \begin{itemize}
        \item
        If~$\interspace{B}{S}$ has the interspace pattern~$\ipsixMatching$ or~$\ipsixMatchingTwice$, then~$B$ splits into three fibers (or unions of fibers) of equal size.
        If~$B$ is large, then by Lemma~\ref{local:progress-in-large/lem} we obtain~$t_B \leq -0.8$ and if~$B$ is small then~$t_{B} \leq -0.3$.

        \item
        If~$\interspace{B}{S}$ has the interspace pattern~$\ipsixMatchingComplementD$, then~$B$ splits into  fibers (or unions of fibers) of equal size.
        If~$|B| \geq 32$, we have~$t_{B} \leq \tau(0,3,0) \leq -1.2$.
        If~$16 \leq |B| < 32$, we have~$t_B \leq \tau(-16,-1,16) \leq -1.2$.
        If~$|B| < 16$, we have~$t_B \leq \tau(-8,-1,0) \leq -0.8$.
        Observe that if~$B$ is small, then~$B \neq L$ and~$\interspace{L}{B}$ contains a constituent isomorphic to a star, which contradicts Lemma~\ref{critical:star/lem}.
    \end{itemize}
    Therefore~$t_L \leq -1.2$.
    Together with~$t_S \leq \tau(0,0,-6)$, we conclude $\tau(\parameters(\coherentConfig_{s,s'})) \leq \tau(\parameters(\coherentConfig)) - 2.1$.

    Assume that~$\abs{L} \in \{8,12\}$ and $R_1$ or~$R_2$ are large.
    We set~$V' \coloneqq \{s\}$ where~$s \in S$ and consider~$\coherentConfig_{s}$.
    We have~$t_S \leq -0.1$.
    Fiber~$L$ splits in half and thus~$t_L \leq \tau(-8,-1,8) \leq -0.4$.
    If~$\interspace{R_1}{S}$ has the interspace pattern~$\ipsixMatchingComplementD$, then~$R_1$ is large and splits in half.
    Thus by Lemma~\ref{local:progress-in-large/lem} we obtain~$t_{R_1} \leq -0.4$.
    If~$\interspace{R_1}{S}$ has the interspace pattern~$\ipsixMatching$  or~$\ipsixMatchingTwice$, then~$R_1$ splits into two fibers of size~$|R_1|/3$ and~$2|R_1|/3$ (or something finer).
    If~$R_1$ is large, then by Lemma~\ref{local:progress-in-large/lem} we obtain~$t_{R_1} \leq -0.4$ and if~$R_1$ is small, then~$t_{R_1} \leq -0.1$.
    The values of~$R_2$ are computed analogously.
    We conclude $\tau(\parameters(\coherentConfig_{s})) \leq \tau(\parameters(\coherentConfig)) - 1$.

    Assume that~$\abs{L} \in \{8,12\}$, $\colorDeg{L} = 1$, and~$R_1$ and~$R_2$ both are small.
    By Lemma~\ref{critical:4cc-6cc/lem} the set~$\{S\}$ is dominating.
    Since~$|\ul(\inducedCC{S})| = 3$, we have~$\disjointCliques{3}{2,2} \in \interspace{R_1}{R_2}$.
    Since~$\colorDeg{L} = 1$, the set~$\{R_2\}$ is not dominating.
    If~$\interspace{R_1}{R_2}$ has the interspace pattern~$\ipsixMatching$ or~$\ipsixMatchingTwice$, then this contradicts Lemma~\ref{critical:6-cc:restorable:DUC:deg1/lem}.
    Thus~$\interspace{R_1}{R_2}$ has the interspace pattern~$\ipsixMatchingAndCycle$ or~$\ipsixMatchingMatching$, and so~$|\ul(\inducedCC{R_2})| >3$.
    However, by Corollary~\ref{large-small-interspace:classification:uniqueness/cor} there is no interspace incident to~$R_2$ in~$\quotientGraph{\coherentConfig}$ which has the interspace pattern~$\ipsixTriangle$.
    This contradicts Lemma~\ref{critical:6-cc:restorable:large-neighborhood/lem}.

    Assume that~$\abs{L} \in \{8,12\}$, $\colorDeg{L} > 1$, and~$R_1$ and~$R_2$ are both small.
    Hence~let~$Y$ be a fiber adjacent to~$L$ other than~$S$.
    Since Theorems~\ref{local:L-S/thm} and~\ref{local:S-L-S/thm} are not applicable, fiber~$L$ is adjacent to only one small fiber in~$\quotientGraph{\coherentConfig}$.
    Thus~$Y$ is large.
    We set~$V' \coloneqq \{\ell\}$ where~$\ell \in L$, and consider~$\coherentConfig_{\ell}$.
    Fiber~$S$ splits into tiny fibers and~$t_S \leq 0.3$.
    Fiber~$L$ splits into two fibers of size~$\frac{|L|}{4}$ and~$\frac{3|L|}{4}$ respectively.
    Hence~$t_L \leq \tau(-8,-1,6) \leq -0.5$ if~$|L| = 8$ and~$t_L \leq \tau(-3,0,0) \leq -0.45$ if~$|L| = 12$.
    At least two vertices are split from fiber~$Y$ and thus~$t_Y \leq \tau(-2,0,0) \leq -0.3$.
    Altogether, we conclude~ $\tau(\parameters(\coherentConfig_{\ell})) \leq \tau(\parameters(\coherentConfig)) - 0.6 + \max\{-0.5,-0.45\} \leq \tau(\parameters(\coherentConfig)) -1.05$.
\end{proof}

\begin{lemma}
\label{local:alternating-6cycle/lem}
    Let~$\coherentConfig$ be a critical coherent configuration, and let~$(L,S,S')$ be a path in~$\quotientGraph{\coherentConfig}$ such that~$L$ is large, the interspace~$\interspace{L}{S}$ has the interspace pattern~$\ipsixMatchingComplementD$, and the interspace~$\interspace{S}{S'}$ has the interspace pattern~$\ipsixMatchingMatching$.
    If~$\{L,S\}$ is not dominating, then
    \begin{itemize}
        \item $\wldim{\coherentConfig} \leq 1 + \widetilde{\f}( \tau(\parameters(\coherentConfig)) - 1.05)$ or
        \item $\wldim{\coherentConfig} \leq 2 + \widetilde{\f}( \tau(\parameters(\coherentConfig)) - 2)$.
    \end{itemize}
\end{lemma}
\begin{proof}
    Assume that Theorems~\ref{local:L-S/thm} and~\ref{local:S-L-S/thm} are not applicable.
    For a union of fibers~$Z$ of~$\coherentConfig$ and~$V' \subseteq \vertices(\coherentConfig)$, set~$(a_Z,b_Z,c_Z) \coloneqq \parameters(\coherentConfig_{V'}[Z]) - \parameters(\coherentConfig[Z])$ and~$t_Z \coloneqq \tau(a_Z,b_Z,c_Z)$. Note that~$|L|$ is divisible by 4.

    Assume first~$\abs{L}\geq 16$.
    Since~$\interspace{S}{S'}$ has the interspace pattern~$\ipsixMatchingMatching$, there are three constituents in~$\inducedCC{S'}$ isomorphic to~$\disjointCliques{3}{2}$ and two constituents isomorphic to~$2\overrightarrow{C_3}$.
    Intuitively, two of them form a color-alternating $6$-cycle.
    Thus~$\interspace{L}{S'}$ can only have the interspace pattern~$\ipsixTriangle$ or the interspace pattern~$\ipsixMatchingMatching$.
    Since Theorems~\ref{local:L-S/thm} and~\ref{local:S-L-S/thm} are not applicable, we conclude~$L$ and~$S'$ are homogeneously connected.
    Hence~$S'$ is not dominating and by Lemma~\ref{critical:6-cc:restorable:large-neighborhood/lem} there is a fiber~$R$ such that~$\interspace{R}{S'}$ has the interspace pattern~$\ipsixTriangle$.
    Set~$V' = \{\ell,s'\}$ where~$\ell \in L$ and~$s' \in S'$, and examine~$\coherentConfig_{\ell,s'}$.
    In~$\coherentConfig_{\ell,s'}$, fibers~$S$ and~$S'$ split into singletons and thus~$t_{S\cup S'} \leq 0.6$.
    Fiber~$R$ splits in half.
    If~$R$ is small, then~$t_R \leq \tau(0,0,-4) \leq -0.2$.
    If~$R$ is large, then by Lemma~\ref{local:progress-in-large/lem} we obtain~$t_R \leq -0.4$.
    Since~$S$ becomes discrete in~$\coherentConfig_{\ell,s'}$, fiber~$L$ splits into the parts of~$\equivalenceClasses{L,S}$, which have size~$\frac{|L|}{4}$.
    If~$|L| \geq 32$, we have~$t_{L} \leq \tau(0,3,0) \leq -1.2$.
    If~$16 \leq |L| < 32$, we have~$t_L \leq \tau(-16,-1,16) \leq -1.2$.
    We conclude~$\tau(\parameters(\coherentConfig_{\ell,s'})) \leq \tau(\parameters(\coherentConfig)) - 2.0$.

    Assume that~$\abs{L} \in \{8,12\}$ and~$\colorDeg{L} = 1$.
    Thus by Lemma~\ref{critical:4cc-6cc/lem} the set~$\{S\}$ is dominating.
    We set~$V' \coloneqq \{s\}$ where~$s \in S$.
    Fiber~$S'$ splits entirely into tiny fibers while~$S$ splits into a size-$2$ fiber and one size-$4$ fiber, the latter of which we refer to by~$\overline{S}$.
    Fiber~$L$ splits into two fibers, both of which have size~$4$ or~$6$ and which we refer to by~$L_1$ and~$L_2$, respectively.
    After we restore criticality, we obtain a critical coherent configuration~$\coherentConfig'$ with~$\vertices(\coherentConfig') = L_1 \cup L_2 \cup \overline{S}$.

    We now show that~$\wldim{\coherentConfig'} \leq 2$ as follows.
    Observe that there is~$\overline{U} \in \coherentConfig'[L_1,\overline{S}]$ and~$\overline{U'} \in \coherentConfig'[L_2,\overline{S'}]$ such that~$(L_1 \disjointUnion \overline{S},\overline{U})$ and~$(L_2 \disjointUnion \overline{S'},\overline{U'})$ are isomorphic to~$\disjointCliques{2}{|L|/2,2}$ and~$\coherentConfig'[\overline{S}]$ contains three constituents isomorphic to~$\disjointCliques{2}{2}$.
    Further for every pair~$(\ell_1,\ell_2) \in L_1 \times L_2$ there is exactly one vertex~$s$ in~$\overline{S}$ such that~$\{s\} = \ell_1 \overline{U} \cap \ell_2 \overline{U'}$.
    Thus~$\overline{S}$ is restorable in~$\coherentConfig'$.
    By Lemma~\ref{interspace-pattern:partition-size/lem} all parts of~$\equivalenceClasses{L,S}$ have size~$2$ or~$3$.
    Therefore by coherence each part of~$\equivalenceClasses{L,S}$ induces a clique in~$\coherentConfig[L]$ and~$\disjointCliques{4}{|L|/4} \in \inducedCC{L}$.
    Furthermore, observe that either all remaining constituents in~$\inducedCC{L}$ are isomorphic to~$\clique{1,1,1,1}$ or there is one remaining constituent isomorphic to either~$\clique{2,2,2,2}$ or~$\clique{3,3,3,3}$ in~$\inducedCC{L}$ since otherwise~$L$ cannot be a fiber of~$\coherentConfig$.
    In the first case~$\matching{\frac{|L|}{2}} \in \coherentConfig'[L_1,L_2]$, and in the second case~$\coherentConfig'[L_1,L_2]$ is homogeneous.
    In both cases, $\wldim{\coherentConfig'} \leq 2$.

    Since~$\coherentConfig'$ has WL-dimension at most~$2$, we may replace it by a coherent configuration~$\coherentConfig''$ which has WL-dimension~$2$.
    We choose~$\coherentConfig''$ to be the homogeneous coherent configuration on~$6$ vertices which contains three constituents isomorphic to~$\cycle{6}$, $\disjointCliques{3}{2}$, and~$\disjointCliques{2}{3}$ respectively.
    Therefore~$(a_{L\cup S \cup S'},b_{L\cup S \cup S'},c_{L\cup S \cup S'}) \preceq (-8,-1,-6)$ and~$t_L \leq -1.1$.

    Assume that~$\abs{L} \in \{8,12\}$ and~$\colorDeg{L} > 1$.
    Let~$B$ be a fiber adjacent to~$L$ other than~$S$.
    Since Theorems~\ref{local:L-S/thm} and~\ref{local:S-L-S/thm} are not applicable, fiber~$L$ is adjacent to only one small fiber in~$\quotientGraph{\coherentConfig}$.
    Thus~$B$ is large.
    We set~$V' \coloneqq\{\ell\}$ where~$\ell \in L$, and consider~$\coherentConfig_{\ell}$.
    Fiber~$S$ splits into tiny fibers and~$t_S \leq 0.3$.
    Fiber~$L$ splits into two fibers of size~$\frac{|L|}{4}$ and~$\frac{3|L|}{4}$.
    Hence~$t_L \leq \tau(-8,-1,6) \leq -0.5$ if~$|L| = 8$ and~$t_L \leq \tau(-3,0,0) \leq -0.45$ if~$|L| = 12$.
    At least two vertices are split from fiber~$B$ and hence~$t_B \leq \tau(-2,0,0) \leq -0.3$.
    We conclude~ $\tau(\parameters(\coherentConfig_\ell)) \leq \tau(\parameters(\coherentConfig)) - 0.6 + \max\{-0.5,-0.45\}$.
\end{proof}

\begin{lemma}
\label{local:6cc:3neighbors/lem}
    Let~$\coherentConfig$ be a critical coherent configuration, and let~$S$ be a size~$6$ fiber of~$\coherentConfig$  such that~$\clique{6} \notin\inducedCC{S}$.
    If~$\colorDeg{S} \geq 3$ and~$\{S\}$ is non-dominating, then~$\wldim{\coherentConfig} \leq 1 + \widetilde{\f}( \tau(\parameters(\coherentConfig)) - 1  )$.
\end{lemma}
\begin{proof}
    We assume that Theorem~\ref{local:L-S/thm} as well as Lemma~\ref{local:K222-3D/lem} and Lemma~\ref{local:alternating-6cycle/lem} are not applicable.
    Let~$\{B_1,\dots,B_{\colorDeg{S}}\}$ be the neighborhood of~$S$ in~$\quotientGraph{\coherentConfig}$,~$i \in \{1,\dots,\colorDeg{S}\}$, and~$s \in S$.
    For a union of fibers~$Z$ of~$\coherentConfig$  and~$V' \subseteq \vertices(\coherentConfig)$, define~$(a_Z,b_Z,c_Z) \coloneqq \parameters(\coherentConfig_{V'}[Z]) - \parameters(\coherentConfig[Z])$ and~$t_Z \coloneqq \tau(a_Z,b_Z,c_Z)$.

    Assume that~$\abs{\ul(\inducedCC{S})} > 3$.
    Lemma~\ref{critical:6-cc:restorable:large-neighborhood/lem} applies. Thus~$\colorDeg{S} \geq 2$ and we can assume (after without loss of generality renaming the~$B_i$) that the  interspace~$\interspace{B_1}{S}$ has either the interspace pattern~$\ipsixMatchingAndCycle$ or the interspace pattern~$\ipsixMatchingMatching$ and the interspace~$\interspace{B_2}{S}$ has the interspace pattern~$\ipsixTriangle$.
    Since Theorem~\ref{local:L-S/thm} is not applicable, fiber~$B_1$ is small.

    Assume that~$\cycle{12} \in \interspace{B_1}{S}$.
    By Lemma~\ref{critical:6-cc:restorable:cycle/lem}, there is a partition~$\{\{B_1\},\mathcal{B}, \mathcal{Y}\}$ of the neighborhood of~$S$ in~$\quotientGraph{\coherentConfig}$ such that
    \begin{enumerate}
        \item for all~$B \in \mathcal{B}$ the interspace~$\interspace{B}{S}$ has the interspace pattern~$\ipsixTriangle$, and
        \item for all~$Y \in \mathcal{Y}$ the interspace~$\interspace{Y}{S}$ has either the interspace pattern~$\ipsixMatching$ or the interspace pattern~$\ipsixMatchingTwice$.
    \end{enumerate}
    Note that~$|\mathcal{B}| \geq 1$, and set~$V' \coloneqq \{s\}$.
    We separately consider~$\coherentConfig_s[B_i]$ for all~$i \in \{2,\dots,\colorDeg{S}\}$.
    \begin{itemize}
        \item
        If~$\interspace{B_i}{S}$ has  the interspace pattern~$\ipsixTriangle$, then~$B_i$ splits into equally sized fibers (or unions of fibers).
        If~$B_i$ is large, then by Lemma~\ref{local:progress-in-large/lem} we obtain~$t_{B_i} \leq -0.4$.
        If~$B_i$ is small, then~$|B_i| \geq 4$ and~$t_{B_i} \leq -0.2$.

        \item
        If~$\interspace{B_i}{S}$ has the interspace pattern~$\ipsixMatchingTwice$, then~$B_i$ splits into three fibers (or unions of fibers) of equal size.
        If~$B_i$ is large, then by Lemma~\ref{local:progress-in-large/lem} we obtain~$t_{B_i} \leq -0.8$.
        If~$B_i$ is small, then~$|B_i| = 6$ and~$t_{B_i} \leq -0.3$.

        \item
        If~$\interspace{B_i}{S}$ has the interspace pattern~$\ipsixMatching$, then~$B_i$ splits into fibers (or unions of fibers) of sizes~$\frac{\abs{B_i}}{3}$ and~$\frac{2\abs{B_i}}{3}$.
        If~$B_i$ is large, then by Lemma~\ref{local:progress-in-large/lem} we obtain~$t_{B_i} \leq -0.4$.
        If~$B_i$ is small, then~$|B_i| = 6$ and~$t_{B_i} \leq - 0.1$.
    \end{itemize}
    Additionally, fibers~$S$ and~$B_1$ split into tiny fibers and thus~$t_{S \cup B_1} \leq -0.6$.
    Furthermore, we have~$\abs{\mathcal{R}} = 1$,~$\abs{\mathcal{Y}} \geq 1$ and~$\abs{\mathcal{B}} \geq 1$.
    We use the reasoning and observations above to deal with the following three cases:

    If~$S$ is adjacent to at least one large fiber in~$\quotientGraph{\coherentConfig}$, then~$\tau(\parameters(\coherentConfig_s)) \leq \tau(\parameters(\coherentConfig))- 1.1$.

    Assume that all neighbors of~$S$ are small and there is at least one fiber~$B_3$ adjacent to~$S$ in~$\quotientGraph{\coherentConfig}$ such that interspace~$\interspace{B_3}{S}$ has the interspace pattern~$\ipsixMatchingTwice$.
    Thus~$B_3 \in \mathcal{Y}$.
    By the reasoning above, we have~$t_{B_3} \leq -0.3$, $t_S \leq -0.3$, $t_{B_1} \leq -0.3$, and~$t_{B_2} \leq -0.2$.
    Altogether we conclude~$\tau(\parameters(\coherentConfig_s)) \leq \tau(\parameters(\coherentConfig))- 1.1$.

    Now assume that all neighbors of~$S$ are small and there is no fiber~$F$ adjacent to~$S$ in~$\quotientGraph{\coherentConfig}$ such that~$\interspace{F}{S}$ has the interspace pattern~$\ipsixMatchingTwice$.
    Without loss of generality, assume that~$B_3$ is a neighbor of~$S$ such that~$\interspace{B_3}{S}$ has the interspace pattern~$\ipsixMatching$.
    By Lemma~\ref{6-cc:implied-interspace:DUC-DUC/lem}, we have~$\disjointCliques{3}{2,2} \in \interspace{B_1}{B_3}$.
    By Lemma~\ref{small-cc:interspace-implies-cc/lem}, we have~$\disjointCliques{3}{2} \in \inducedCC{B_3}$.
    Suppose that~$|\ul(\inducedCC{B_3})| = 3$.
    Note that there is exactly one constituent in~$\interspace{B_3}{S}$ isomorphic to~$\disjointCliques{3}{2,2}$, and so we have~$\clique{2,2,2} \in \inducedCC{B_3}$.
    Thus by Lemma~\ref{critical:small-cc:module/lem} there is a fiber~$Y$ such that~$\interspace{Y}{B_2}$ has the interspace pattern~$\ipsixMatchingComplementD$ or the interspace pattern~$\ipsixMatchingComplement$.
    This contradicts the assumption that Theorem~\ref{local:L-S/thm} and Lemma~\ref{local:K222-3D/lem} are not applicable.
    If~$|\ul(\inducedCC{B_3})| = 5$, then~$\interspace{S}{B_3}$ has the interspace pattern~$\ipsixMatchingMatching$.
    Thus~$\interspace{B_3}{S}$ has the interspace pattern~$\ipsixMatchingTwice$, which contradicts the assumption.
    Therefore~$|\ul(\inducedCC{B_3})| = 4$.
    Lemma~\ref{critical:6-cc:restorable:large-neighborhood/lem} applies to~$B_3$.
    Further, there is fiber~$Y$ such that~$\cycle{12} \in \interspace{Y}{B_3}$. (Recall that since Theorem~\ref{local:L-S/thm} is not applicable, fiber~$Y$ is small.)
    Since~$\coherentConfig$ is critical, Lemma~\ref{critical:adjacent-interspace-cycle/lem} implies~$Y \neq S$ and~$Y \neq B_1$, and Corollary~\ref{large-small-interspace:classification:uniqueness/cor} implies~$Y \neq B_2$.
    By Lemma~\ref{6-cc:implied-interspace:DUC-DUC/lem} we have~$\disjointCliques{3}{2,2} \in \interspace{Y}{S}$.
    Altogether, we have~$|\mathcal{Y}| \geq 2$,~$\abs{\mathcal{R}} = 1$, and~$\abs{\mathcal{B}} \geq 1$.
    From the calculation above, we conclude~$\tau(\parameters(\coherentConfig_s)) \leq \tau(\parameters(\coherentConfig))- 1$.

    Assume that~$\interspace{B_1}{S}$ has the interspace pattern~$\ipsixMatchingMatching$.
    Since~$\interspace{B_1}{B_2}$ is homogeneous,~$\{B_1\}$ is not dominating.
    By Lemma~\ref{critical:6-cc:restorable:DUC:deg1/lem}, we have~$\colorDeg{B_1} \geq 2$.
    If~$|\ul(\inducedCC{B_1})| = 3$, then there is fiber~$F$ adjacent to~$B_1$ in~$\quotientGraph{\coherentConfig}$ such that~$\interspace{R}{S}$ has the interspace pattern~$\ipsixMatchingComplement$ or~$\ipsixMatchingComplementD$.
    However Theorem~\ref{local:L-S/thm} and Lemma~\ref{local:alternating-6cycle/lem} are not applicable, and thus~$|\ul(\inducedCC{B_1})| >3$.
    Thus Lemma~\ref{critical:6-cc:restorable:large-neighborhood/lem} applies to~$B_1$.
    Hence there are~$R_1,R_2$ adjacent to~$B_1$ such that~$\interspace{R_1}{B_1}$ has either the interspace pattern~$\ipsixMatchingAndCycle$ or the interspace pattern~$\ipsixMatchingMatching$ and~$\interspace{R_2}{B_1}$ has the interspace pattern~$\ipsixTriangle$.
    Recall that, since Theorem~\ref{local:L-S/thm} is not applicable, fiber~$R_1$ is small.
    Furthermore~$R_1 \neq S$ and~$R_2 \neq S$ since~$\interspace{S}{B_1}$ has the interspace pattern~$\ipsixMatchingTwice$ (which is determined by the interspace pattern of~$\interspace{B_1}{S}$).
    We set~$V' \coloneqq \{b\}$ where~$b \in B_1$.
    By reasoning similar to the one above, we have~$t_S \leq 0.3$, $t_{R_1} \leq 0.3$, $t_{R_2} \leq 0.2$, and $t_{B_1} \leq 0.3$.
    We conclude~$\tau(\parameters(\coherentConfig_b)) \leq \tau(\parameters(\coherentConfig))- 1$.

    Assume that~$\abs{\ul(\inducedCC{S})} = 3$.
    Hence there is a constituent in~$\inducedCC{S}$ whose underlying graph is isomorphic to~$\disjointCliques{3}{2}$ or~$\disjointCliques{2}{3}$.
    By Lemma~\ref{critical:6-cc:restorable:DUC:deg1/lem} there is at least one fiber~$R$ adjacent to~$S$ such that~$\interspace{R}{S}$ has one of the following interspace patterns:~$\ipsixTriangleComplement$,~$\ipsixTriangleComplementTwice$, $\ipsixMatchingComplement$, or~$\ipsixMatchingComplementD$.
    Since Theorem~\ref{local:L-S/thm} and Lemma~\ref{local:K222-3D/lem} are not applicable, we have~$\disjointCliques{2}{3} \in \inducedCC{S}$ and~$\interspace{R}{S}$ has interspace patter~$\ipsixTriangleComplement$.
    Further, by the size of the equipartition provided by Lemma~\ref{interspace-pattern:partition-size/lem}, we conclude~$R$ is large.
    Thus let~$\{\mathcal{B},\mathcal{R}\}$ be a partition of the fibers adjacent to~$S$ in~$\quotientGraph{\coherentConfig}$ such that
    \begin{enumerate}
        \item for all~$R \in \mathcal{R}$ the interspace~$\interspace{R}{S}$ has the interspace pattern~$\ipsixTriangle$, and
        \item for all~$B \in \mathcal{B}$ the interspace~$\interspace{B}{S}$ has the interspace pattern~$\ipsixTriangleComplement$ or the interspace pattern~$\ipsixTriangleComplementTwice$.
    \end{enumerate}
    We set~$V' \coloneqq \{s\}$ where~$s \in S$, and consider separately~$\coherentConfig_s[B_i]$ for~$i \in \{1,\dots,\colorDeg{S}\}$.
    \begin{itemize}
        \item
        If~$B_i \in \mathcal{R}$, then~$B_i$ splits into two fibers (or unions of fibers) of equal size.
        If~$B_i$ is large, then by Lemma~\ref{local:progress-in-large/lem} we obtain~$t_{B_i} \leq -0.4$.
        If~$B_i$ is small, then~$|B_i| \geq 4$ and~$t_{B_i} \leq -0.2$.
        \item
        If~$B_i \in \mathcal{B}$, then fiber~$B_i$ splits into fibers (or unions of fibers) of sizes~$\frac{\abs{B_i}}{3}$ and~$\frac{2\abs{B_i}}{3}$.
        Since~$9$ divides~$|B_i|$, by Lemma~\ref{local:progress-in-large/lem} we have~$t_{B_i} \leq -0.4$.
    \end{itemize}
    Additionally, the fiber~$S$ splits entirely into tiny fibers and thus~$t_S \leq -0.3$.
    Due to Lemma~\ref{critical:6-cc:restorable:DUC:deg1/lem} we have~$\abs{\mathcal{B}} \geq 1$, and therefore we conclude~$\tau(\parameters(\coherentConfig_s)) \leq \tau(\parameters(\coherentConfig))- 1.1$.
\end{proof}


\section{The structure of reduced coherent configurations}
\label{structure-reduced-cc/sec}

In this section, we aim to show that, if no local reduction of the previous section is applicable to the coherent configuration~$\coherentConfig$, at least one of the following two cases applies:
in the first case, the WL-dimension of~$\coherentConfig$ is bounded, and, in the second one, the coherent configuration~$\coherentConfig$ has a very specific structure, which we describe by collecting 7 properties.
To this end, we require some additional definitions.
We call a small fiber~$S$~\emph{relevant} if~$|S| \in \{4,6\}$ and~$|\ul(\inducedCC{S})| > 2$, and~\emph{irrelevant} otherwise.
For a fiber~$R$ of~$\coherentConfig$, we define~$\colorDegRelevantSmall{R}$ to be number of relevant small fibers adjacent to~$R$ in~$\quotientGraph{\coherentConfig}$.
Further, recall that~$\colorDegLarge{F}$ (respectively~$\colorDegSmall{\coherentConfig}$) are the number of large (respectively small) fibers adjacent to~$F$ in~$\quotientGraph{\coherentConfig}$.

\begin{lemma}
    \label{irrelevant-small-fibers/lem}
    Let~$\coherentConfig$ be a critical coherent configuration such that for all relevant small fibers~$R$ the set~$\{R\}$ is non-dominating.
    If~$S$ is an irrelevant small fiber of~$\coherentConfig$,
    then~$\colorDegSmall{S} \leq 1$,
    the set of fibers adjacent to~$S$ in~$\quotientGraph{\coherentConfig}$ induces a clique in~$\quotientGraph{\coherentConfig}$,
    and~$S$ is not adjacent to a relevant small fiber in~$\quotientGraph{\coherentConfig}$.
\end{lemma}
\begin{proof}
    An irrelevant small fiber~$S$ either has odd size or~$|\ul(\inducedCC{S})| = 2$.
    By Theorem~\ref{interspace-divisor/thm} and Lemma~\ref{global-argument:k4-exclusion/lem} the interspace between any two neighbors of~$S$ in~$\quotientGraph{\coherentConfig}$ is non-homogeneous.
    Thus the set of neighbors of~$S$ in~$\quotientGraph{\coherentConfig}$ induces a clique in~$\quotientGraph{\coherentConfig}$.

    Assume that the size of~$S$ is odd.
    By Lemma~\ref{small-cc:interspace/lem}, any non-homogeneous interspace between~$S$ and another small fiber contains a constituent that is isomorphic to a matching,~$\cycle{2|S|}$, or~$\leviFano$.
    By Lemmas~\ref{critical:star/lem} and~\ref{critical:cycle/lem}, fiber~$S$ has size~$7$ and interspaces between~$S$ and another small fibers~$S'$ contain a constituent isomorphic to~$\leviFano$.
    Lemma~\ref{critical:7-cc:leviFano/lem} implies that~$\colorDegSmall{S} \leq 1$.
    Further note that~$S'$ has size~$7$ and therefore is irrelevant as well.

    Next assume~$S$ has even size but~$|\ul(\inducedCC{S})| = 2$.
    Thus there is only one constituent in~$\inducedCC{S}$, which is thus isomorphic to a clique.
    Considering Lemma~\ref{interspace-pattern:partition-size/lem}, interspaces between~$S$ and another small fiber~$R$ must have the interspace pattern~$\ipfourClique$ and~$|S| = 4$ since for all other possible patterns the size of the equipartition is already larger than~$7$.
    By Lemma~\ref{small-cc:interspace-implies-cc/lem}, fiber~$R$ is relevant and by Lemma~\ref{critical:4cc-6cc/lem} the set~$\{R\}$ is dominating.
    This violates the assumption, and hence~$\colorDegSmall{S} = 0$.
\end{proof}

\begin{lemma}
\label{dominating:wldim/lem}
    Let~$\coherentConfig$ be a critical coherent configuration and there is~$t \in \Nat$ such that all fibers of~$\coherentConfig$ have size at most~$t$.
    If there is a set of fibers~$\mathcal{R}$ which is dominating and there are at most~$u \in \Nat$ large fibers adjacent to a fiber of~$\mathcal{R}$, then~$\wldim{\coherentConfig} \leq 2 + \sum_{R \in \mathcal{R}} |R| + u \cdot t$.
\end{lemma}
\begin{proof}
    Let~$V$ be the set of all vertices contained in a fiber of~$\mathcal{R}$.
    In~$\coherentConfig_{V}$ every fiber of~$\mathcal{R}$ splits into singletons.
    Thus all small fibers adjacent to a fiber of~$\mathcal{R}$ in~$\quotientGraph{\coherentConfig}$ split into tiny fibers.
    By restoring criticality, we obtain a coherent configuration with at most~$u \cdot t$ vertices.
    Hence~$\wldim{\coherentConfig} \leq 2 + \abs{V} + u \cdot t$.
\end{proof}

\begin{definition}
\label{new:global-argument:assumption}
    For~$t\in \Nat$, we call a coherent configuration~$\coherentConfig$ \emph{$t$-reduced}, if it satisfies all of the following properties.
    \begin{enumerate}
        \item \label{new:global-argument:assumption:critical}
        $\coherentConfig$ is critical.
        \item \label{new:global-argument:assumption:limited-fiber-size}
        Every fiber~$R$ has at most~$t$ vertices, i.e.,~$|R|\leq t$.
        \item \label{new:global-argument:assumption:large-colorDeg}
        All large fibers~$L$ have at most two large neighbors, that is~$\colorDegLarge{L} \leq 2$.
        \item \label{new:global-argument:assumption:large-one-relevent}
        A large fiber has at most one relevant small neighbor, that is~$\colorDegRelevantSmall{L} \leq 1$.
        \item \label{new:global-argument:assumption:relevant-2-neighbors}
        If a relevant small fiber~$S$ has a large neighbor (i.e., if~$\colorDegLarge{S}\geq 1$), then its quotient degree is at most 2 (i.e., $\colorDeg{S}\leq 2$) and~$|\ul(\inducedCC{S})| = 3$.
        \item \label{new:global-argument:assumption:relevant-3-neighbors}
        For every~$s\in S$ in a relevant small fiber~$S$ neighboring with three small fibers (i.e.,~$\colorDegSmall{S}\geq 3$), we have that~$S$ is discrete in~$\coherentConfig_s$.
        \item \label{new:global-argument:assumption:dominating-relevant-small-fibers}
        Each relevant small fiber is non-dominating in the quotient graph~$\quotientGraph{\coherentConfig}$.
    \end{enumerate}
\end{definition}

\begin{lemma}\label{lem:properties:or:t:reduced}
    Let~$\coherentConfig$ be a critical coherent configuration whose largest fibers have size~$t$.
    Then
    \begin{itemize}
        \item $\wldim{\coherentConfig} \leq z + \widetilde{\f}( \tau(\parameters(\coherentConfig)) - z)$ for some positive integer~$z$,
        \item $\wldim{\coherentConfig} \leq 10 + 2\cdot t$, or
        \item $\coherentConfig$ is $t$-reduced (i.e., satisfies Definition~\ref{new:global-argument:assumption}).
    \end{itemize}
\end{lemma}
\begin{proof}
    Assume the first conclusion does not hold for~$\coherentConfig$.
    Throughout the proof, we will now repeatedly use that we can assume that the reduction lemmas of the previous sections with conclusions of the form $\wldim{\coherentConfig} \leq z + \widetilde{\f}( \tau(\parameters(\coherentConfig)) - z')$ with~$1\leq z\leq z'$ and~$z\in \Nat$ are not applicable.
    We will argue that~$\coherentConfig$ is~$t$-reduced or~$\wldim{\coherentConfig} \leq 10 + 2\cdot t$.

    We already assume that~$\coherentConfig$ is critical and thus Property~\ref{new:global-argument:assumption:critical} holds.
    Property~\ref{new:global-argument:assumption:limited-fiber-size} is given by the definition of~$t$.
    Lemma~\ref{lem:local-argument:3-large-neighbors} implies that~$\colorDegLarge{R} \leq 2$ for all fibers~$R$.
    Hence Property~\ref{new:global-argument:assumption:large-colorDeg} holds.

    Assume there is a relevant small fiber~$S$ which is dominating.
    Since Lemma~\ref{lem:local-argument:3-large-neighbors} is not applicable, we have~$\colorDegLarge{S} \leq 2$.
    By assumption, these large neighbors have size at most~$t$.
    Thus by Lemma~\ref{dominating:wldim/lem} we have~$\wldim{\coherentConfig} \leq 8 +2\cdot t$.

    For the rest of the proof, we may assume that there are no dominating relevant small fibers in~$\coherentConfig$.
    Hence,  Property~\ref{new:global-argument:assumption:dominating-relevant-small-fibers} holds.

    Since Theorem~\ref{local:L-S/thm} does not apply, all interspaces between large and relevant small fibers have one of the following interspace patterns: $\ipsixTriangle$, $\ipfourCycle$, $\ipsixMatching$, $\ipfourMatching$, or~$\ipsixMatchingComplementD$.
    Theorems~\ref{local:S-L-S/thm} and~\ref{local:S-L-S:rest/thm} imply that~$\colorDegRelevantSmall{L} \leq 1$ for all large fibers~$L$.
    So Property~\ref{new:global-argument:assumption:large-one-relevent} holds.

    Let~$S$ be a relevant size-$4$ fiber.
    \begin{itemize}
        \item
        If there is a small fiber~$S'$ adjacent to~$S$ in~$\quotientGraph{\coherentConfig}$ such that~$\cycle{8} \in \interspace{S}{S'}$, then by Lemma~\ref{critical:4cc:restorable:cycle/lem} the set~$\{S,S'\}$ is dominating.
        Since Lemma~\ref{local:4cc:3neighbors/lem} does not apply, there are either at most~$2$ large fibers, which have at most size~$t$, or at most~$4$ small fibers adjacent to~$S,S'$ in~$\quotientGraph{\coherentConfig}$.
        By Lemma~\ref{dominating:wldim/lem} we have~$\wldim{\coherentConfig} \leq 10+2t$.

        \item
        Assume $S$ has a large neighbor in~$\quotientGraph{\coherentConfig}$.
        Since Lemma~\ref{local:4cc:3neighbors/lem} is not applicable, we have $\colorDeg{S} \leq 2$.
        Since~$S$ is relevant we have~$|\ul(\inducedCC{S})| \in \{3,4\}$ by Lemma~\ref{small-cc:induced-cc/lem}.
        If~$|\ul(\inducedCC{S})| = 4$, then Lemma~\ref{critical:4-cc:restorable:DUC/lem} implies that~$\colorDeg{S} \geq 3$.
        Therefore~$|\ul(\inducedCC{S})| = 3$.
        This yields Property~\ref{new:global-argument:assumption:relevant-2-neighbors}.

        \item
        Assume that all neighbors of~$S$ in~$\quotientGraph{\coherentConfig}$ are small.
        If there is a small fiber~$S'$ adjacent to~$S$ in~$\quotientGraph{\coherentConfig}$ such that~$\cycle{8} \in \interspace{S}{S'}$, then by Lemma~\ref{critical:4cc:restorable:cycle/lem} the set~$\{S,S'\}$ is dominating.
        Thus~$\wldim{\coherentConfig} \leq 4$.
        If no interspace has the interspace pattern~$\ipfourCycle$, then Lemma~\ref{critical:4cc:restorable:2,C4/lem} implies~$|\ul(\inducedCC{S})| = 4$ since~$S$ relevant.
        Due to Lemma~\ref{critical:4-cc:restorable:DUC/lem} we have~$\colorDeg{S} \geq 3$.
        Further there are three constituents in~$S$ isomorphic to~$\disjointCliques{2}{2}$.
        Thus fiber~$S$ splits into singletons in~$\coherentConfig_s$ where~$s \in S$.
        This yields Property~\ref{new:global-argument:assumption:relevant-3-neighbors}.
    \end{itemize}
    Altogether for all relevant small fibers of size~$4$ either~$\wldim{\coherentConfig} < 10 + 2\cdot t$ or Property~\ref{new:global-argument:assumption:relevant-2-neighbors} and Property~\ref{new:global-argument:assumption:relevant-3-neighbors} holds.

    Now we deal with relevant small fibers of size~$6$.

    Recall that the following reductions are not applicable:
    Theorem~\ref{local:L-S/thm} and Lemmas~\ref{local:K222-3D/lem},~\ref{local:alternating-6cycle/lem}, and~\ref{local:6cc:3neighbors/lem}.
    Thus for all relevant size-$6$ fibers~$S$ that are not dominating we have~$\colorDeg{S} \leq 2$.

    Let~$S$ be relevant size~$6$-fiber.
    Towards a contradiction, suppose~$|\ul(\inducedCC{S})| = 3$ and all neighbors of~$S$ are small.
    Then by Lemma~\ref{small-cc:induced-cc/lem} either~$\disjointCliques{2}{3} \in \inducedCC{S}$ or~$\disjointCliques{3}{2} \in \inducedCC{S}$.
    So assume that there is a constituent~$G$ in~$\inducedCC{S}$ isomorphic to~$\disjointCliques{3}{2}$.
    Since all neighbors of~$S$ are small and~$|\ul(\inducedCC{S})| = 3$, by Lemma~\ref{small-cc:interspace-implies-cc/lem} all interspaces incident to~$S$ have the interspace pattern~$\ipsixMatching$ or~$\ipsixMatchingTwice$.
    This contradicts Lemma~\ref{critical:6-cc:restorable:DUC:deg1/lem}.
    If~$\disjointCliques{2}{3}\in\inducedCC{S}$, the reasoning is similar.

    So let~$S$ be a relevant size-$6$ fiber which is not dominating, has only small neighbors in~$\quotientGraph{\coherentConfig}$, $\colorDeg{S} \leq 2$, and~$|\ul(\inducedCC{S})| > 3$.
    Since~$\colorDeg{S} \leq 2$, by Lemma~\ref{critical:6-cc:restorable:cycle/lem} there is no interspace containing a constituent isomorphic to~$\cycle{12}$.
    Hence by Lemma~\ref{critical:6-cc:restorable:large-neighborhood/lem}, there is a small fiber~$S'$ adjacent to~$S$ such that~$\interspace{S'}{S}$ has the interspace pattern~$\ipsixMatchingMatching$.
    By the same reasoning as above, we have~$\colorDeg{S'} \leq 2$ and~$\{S'\}$ is not dominating.
    If~$|\ul(\inducedCC{S'})| > 3$, Lemma~\ref{critical:6-cc:restorable:large-neighborhood/lem} applies.
    Furthermore, the part~$\mathcal{Y}$ mentioned in that Lemma is also not empty since~$\interspace{S}{S'}$ has the interspace pattern~$\ipsixMatchingTwice$. (This is determined by the interspace pattern~$\ipsixMatchingAndCycle$.)
    We conclude~$\colorDeg{S'} \geq 3$.
    This is a contradiction since we have already proven that~$\colorDeg{S} < 3$ for all relevant size~$6$ fibers.
    Thus~$|\ul(\inducedCC{S'})| = 3$.
    Due to the interspace pattern of~$\interspace{S}{S'}$, there are two constituents in~$\inducedCC{S'}$ whose underlying graphs are isomorphic to~$\clique{2,2,2}$ and~$\disjointCliques{3}{2}$ respectively.
    If all interspaces have the interspace pattern~$\ipsixMatching$ or~$\ipsixMatchingTwice$, then we have a contradiction to Lemma~\ref{critical:6-cc:restorable:DUC:deg1/lem}.
    Thus there is at least one fiber~$R$ such that~$\interspace{R}{S}$ has the interspace pattern~$\ipsixMatchingComplementD$. (Recall since Theorem~\ref{local:L-S/thm} is not applicable we can rule out pattern~$\ipsixMatchingComplement$.)
    However then all preconditions of Lemma~\ref{local:alternating-6cycle/lem}, which is not applicable, are satisfied.

    Altogether for all relevant small fibers of size~$6$ either Property~\ref{new:global-argument:assumption:relevant-2-neighbors} and Property~\ref{new:global-argument:assumption:relevant-3-neighbors} holds or~$\wldim{\coherentConfig} < 8 + 2\cdot t$.
\end{proof}


\section{A global argument}
\label{global-argument/sec}

In the previous section, we examined critical coherent configurations~$\coherentConfig$ that are~$t$-reduced (Definition~\ref{new:global-argument:assumption}) and collected some of their structural properties.
This section's goal is to bound the WL-dimension of such coherent configuration.
In the following, we briefly describe the intuition of this section's goal and method:
in the coherent configuration~$\coherentConfig$ one finds many induced subconfigurations~$\coherentConfig'$ each of which has one of two shapes:
in the first case,~$\coherentConfig'$ mainly consists of paths or cycles of large fibers with (possibly) some irrelevant small fibers attached.
To obtain an upper bound, we use our work of Section~\ref{sec:limit:fiber:sizes}.
In the second case, all fibers of~$\coherentConfig'$ are small, and in Section~\ref{wldim-small/sec} we establish an upper bound.
Thus, if we split~$\coherentConfig$ into these induced subconfigurations, we can deal with each subconfigurations on its own.
However, there are still some relevant small fibers that connect theses induced subconfigurations via paths.
Due to the structural properties, we use a set of individualizations to separate~$\coherentConfig$ into these induced subconfigurations on a global scale.
Furthermore, we charge the cost of individualization to the large fibers.
Overall, this enables us to bound the WL-dimension of such critical coherent configuration by the potential function introduced in Section~\ref{sec:potential:func}.

Before we consider the effects of the mentioned individualizations on a global scale, we examine the possible paths between the subconfigurations.
For each such path, we show that there are individualizations such that after restoring coherence and criticality which this path does not exists anymore.

\begin{lemma}
\label{global-argument:LSL/lem}
    Let~$\coherentConfig$ be a critical coherent configuration that is $t$-reduced (Defini-{\linebreak}tion~\ref{new:global-argument:assumption}), and let~$(L,S,L')$ be an induced path in~$\quotientGraph{\coherentConfig}$ such that~$L, L'$ are large and $S$ is small and relevant.
    There are~$s_{1},\dots,s_{|S|/2 -1} \in S$ such that either the vertex sets~$L \cup S$ and~$L'$ or the vertex sets~$L$ and~$L' \cup S$ are homogeneously connected in~$(\coherentConfig[L \cup S \cup L'])_{s_1,\dots,s_{|S|/2 -1}}$.
    Furthermore~$L$ or~$L'$ has at least size~$9$ if~$|S| = 6$.
\end{lemma}
\begin{proof}
    Since~$\coherentConfig$ is~$t$-reduced, by Condition~\ref{new:global-argument:assumption:relevant-2-neighbors},
    we have~$\abs{\ul(\inducedCC{S})} = 3$ and~$S$ is the only relevant small fiber adjacent to~$L$ (respectively~$L'$) in~$\quotientGraph{\coherentConfig}$.
    Recall that~$\interspace{L}{L'}$ is homogeneous and that by Lemma~\ref{critical:small-cc:module/lem} fiber~$S$ is not a union of modules.
    Thus up to symmetry of~$L$ and~$L'$, we only need to consider the following cases:
    \begin{enumerate}
        \item The interspace~$\interspace{L}{S}$ has the interspace pattern~$\ipsixTriangle$ and~$\interspace{L'}{S}$ has the interspace pattern~$\ipsixTriangleComplement$ or~$\ipsixTriangleComplementTwice$.
        \item The interspace~$\interspace{L}{S}$ has the interspace pattern~$\ipsixMatching$  and~$\interspace{L'}{S}$ has the interspace pattern~$\ipsixMatchingComplement$ or~$\ipsixMatchingComplementD$.
        \item The interspace~$\interspace{L}{S}$ has the interspace pattern~$\ipfourMatching$ and~$\interspace{L'}{S}$ has the interspace pattern~$\ipfourCycle$.
    \end{enumerate}
    Choose~$s_{1},\dots,s_{|S|/2 -1} \in S$ so that none of the vertices are adjacent in~$(S,\arcs^1(\interspace{L}{S}))$.
    In $\coherentConfig_{s_1,\dots,s_{x-1}}$, the fiber~$S$ is split into the connected components of~$(S,\arcs^1(\interspace{L}{S}))$ and~$L$ splits into~$\equivalenceClasses{L,S}$.
    Due to the interspace pattern of~$\interspace{L}{S}$, the claim follows.

    Assume~$|S| = 6$.
    If~$\interspace{L'}{S}$ has the interspace pattern~$\ipsixTriangleComplement$ or~$\ipsixTriangleComplementTwice$, then Lemma~\ref{interspace-pattern:partition-size/lem} implies that~$9$ divides~$|L'|$.
    If~$\interspace{L}{S}$ has the interspace pattern~$\ipsixMatching$, then Lemma~\ref{interspace-pattern:partition-size/lem} implies that~$3$ divides~$|L|$.
    In both cases~$L$ or~$L'$ has size at least~$9$.
\end{proof}

\begin{lemma}
\label{global-argument:LSSL/lem}
    Let~$\coherentConfig$ be a critical coherent configuration that is $t$-reduced (Defini-{\linebreak}tion~\ref{new:global-argument:assumption}), and let~$(L_0 ,S_0 ,S_1, L_1)$ be an induced path in~$\quotientGraph{\coherentConfig}$ such that~$L_0, L_1$ are large and~$S_0, S_1$ are small and relevant.
    There are~$s_{1},\dots,s_{|S_0|/2 -1} \in S_0$ such that vertex sets~$L_0 \cup S_0$ and~$L_1 \cup S_1$ are homogeneously connected in~$(\coherentConfig[L_0 \cup S_0 \cup S_1 \cup L_1])_{s_1,\dots,s_{|S_0|/2-1}}$.
\end{lemma}
\begin{proof}
    Since~$\coherentConfig$ is~$t$-reduced, we have~$|\ul(\inducedCC{S_0})| = |\ul(\inducedCC{S_1})| = 3$ and~$S_0$ (respectively~$S_1$) is the only relevant small fiber adjacent to~$L_0$ (respectively~$L_1$) in~$\quotientGraph{\coherentConfig}$.
    Since~$\coherentConfig$ is critical, we have~$\interspaceFourSix \notin \interspace{S_0}{S_1}$.
    Further by Lemma~\ref{critical:4cc:restorable:cycle/lem} we have~$\cycle{8} \notin \interspace{S_}{S_1}$, and by Lemma~\ref{critical:6-cc:restorable:cycle/lem} we have~$\cycle{12} \notin \interspace{S_0}{S_1}$.
    Thus by Lemma~\ref{small-cc:interspace/lem} we only need to consider the following cases:
    \begin{multicols}{2}
        \begin{enumerate}
            \item $\disjointCliques{2}{2,2} \in \interspace{S_0}{S_1}$
            \item $\disjointCliques{2}{3,3} \in \interspace{S_0}{S_1}$
            \item $\disjointCliques{3}{2,2} \in \interspace{S_0}{S_1}$
            \item $\disjointCliques{2}{2,3} \in \interspace{S_0}{S_1}$
        \end{enumerate}
    \end{multicols}
    Choose vertices~$s_{1},\dots,s_{|S_0|/2 -1} \in S_0$ so that none of the vertices are adjacent in~$(S_0,\arcs^1(\interspace{L_0}{S_0}))$.
    In~$\coherentConfig_{s_1,\dots,s_{x-1}}$, the fiber~$S_0$ is split into the connected components of~$(S_0,\arcs^1(\interspace{L_0}{S_0}))$ and~$S_1$ splits into~$\equivalenceClasses{S_1,S_0}$.
    Due to the interspace pattern of~$\interspace{S_1}{S_0}$, the claim follows.
\end{proof}

\begin{lemma}
\label{lem:the:newest:global-argument}
    Let~$\coherentConfig$ be a~$t$-reduced critical coherent configuration with the parame\-ters~$Par(\coherentConfig)=(n_\ell, k_\ell,n_s)$.
    There is a set~$M$ of size~$q\leq k_\ell$ such that each connected component of~$\quotientGraph{\coherentConfig_M}$
    \begin{itemize}
        \item induces a configuration of WL-dimension at most~$3t$ or
        \item has only small fibers and order at most~$n_s - r_1 \cdot 3-r_2\cdot 4$, where~$r_1\cdot 8.5 +r_2\cdot 8\leq n_\ell$.
    \end{itemize}
\end{lemma}
\begin{proof}
    Our goal is to disconnect all components of the large quotient graph by individualizing vertices in small fibers that connect them.
    Due to Lemma~\ref{crictial:quotientGraph-connected/lem} the quotient graph~$\quotientGraph{\coherentConfig}$ is connected.

    Consider the subgraph of quotient graph~$\quotientGraphSmall{\coherentConfig}$ induced by the set of all small fibers.
    By Lemma~\ref{irrelevant-small-fibers/lem} either all fibers of a connected component of~$\quotientGraphSmall{\coherentConfig}$ are relevant or the connected component has at most two fibers.
    For each of its components~$C$ we will construct a vertex set~$M_C$ such that after individualizing and taking the coherent closure we disconnect the neighboring large fibers. The final vertex set~$M$ will be the union of all vertex sets~$M_C$. We will charge the cost of individualizing the set~$M_C$ partially to the vertices in the large fibers neighboring~$C$ and partially to the small fibers that disappear.
    Since large fibers have at most one small neighboring fiber, there will be no double charging.

    This graph~$\quotientGraphSmall{\coherentConfig}$ of small fibers has four types of connected components~$C$:
    \begin{itemize}
        \item (\emph{Type 1}).
        The first type of component~$C$ contains at least one irrelevant fiber (small fiber of size~$5$, size~$7$, or that induce a complete graph).
        By Lemma~\ref{irrelevant-small-fibers/lem} component~$C$ has at most~$2$ fibers.
        Further, if there are neighbors of~$C$, then these neighbors are all large and form a clique while~$C$ contains only irrelevant small fibers.

        \item (\emph{Type 2}).
        The second type are components that only consist of a relevant single small fiber~$S$. All neighbors of~$S$ are large. We can assume that~$S$ has at least two neighbors~$L_1$ and~$L_2$, otherwise we can set~$M_C$ to be the empty set.
        We set~$M_C$ to contain the~$(|S|/2 -1)$-many vertices that are to be individualized of Lemma~\ref{global-argument:LSL/lem}.
        Thus in~$(\coherentConfig[S\cup L_1\cup L_2])_{M_C}$ the partition into~$S\cup L_1$ and~$L_2$ or the partition into~$S\cup L_1$ and~$L_2$ is homogeneously connected.

        We charge the cost of the individualization to~$L_1$,~$L_2$ and the set~$S$:
        If~$|S|=4$, we charge 1 individualization to at least~$8+8=16$ vertices in large fibers and at least~$4$ vertices in small fibers.
        If~$|S|=6$, at least one of the fibers~$L_1$ or~$L_2$ has size at least 9 by Lemma~\ref{global-argument:LSL/lem}. We charge 2 individualizations to at least~$9+8=17$ vertices in large fibers and at least~$6$ vertices in small fibers.
        In either case the number of vertices in large fibers that are charged is at least~$8.5$ per individualization and at least~$3$ vertices in small fibers per individualization.

        \item (\emph{Type 3}).
        The third type of component consists of relevant small fibers that all have quotient degree 2.
        As before, all neighboring fibers of the component are large. If the component has at least two neighbors~$L_1$ and~$L_2$, by Lemma~\ref{global-argument:LSSL/lem}  we can individualize two vertices~$v_1,v_2\in S$ so that in~$(\coherentConfig[S_1\cup S_2\cup L_1\cup L_2])_v$
        the sets~$S_1\cup L_1$ and~$S_2\cup L_2$ are homogeneously connected. We set~$M_C=\{v_1,v_2\}$. We charge the vertices of~$L_1$,~$L_2$, and the all vertices in the small fibers.

        Overall we charge 2 individualizations to at least~$8+8=16$ vertices in large fibers and at least~$8$ vertices in small fibers, so at least~$8$ vertices in large and~$4$ vertices in small fibers per individualization.

        \item (\emph{Type 4}).
        The fourth type of component is a component~$C$ that contains at least 4 relevant small fibers.
        Since the configuration is~$t$-reduced,
        the component is comprised as follows. Due to Property~\ref{new:global-argument:assumption:relevant-3-neighbors}, there are two kinds of fibers. Those that have quotient degree 2 with a neighboring large and a neighboring relevant small fiber. Let us call these \emph{boundary fibers}.
        And relevant small fibers of degree 3 that have exactly 3 relevant small neighbors. We call these \emph{inner fibers}.

        Consider the quotient graph~$\quotientGraph{\coherentConfig}$.
        Let~$T$ be the set of relevant small fibers~$S$ of~$C$ for which there is a path that starts in~$S$ ends in a large fiber and all of whose internal vertices are relevant small fibers of quotient degree 2.  Property~\ref{new:global-argument:assumption:large-one-relevent} implies that~$|T|$ is bounded by the number~$\ell$ of large fibers that have a neighbor in~$C$.

        Form~$M_C$ by picking one vertex from each fiber in~$T$. Since every fiber~$S$ in~$T$ has quotient degree~$3$, by Property~\ref{new:global-argument:assumption:relevant-3-neighbors}, the set~$S$ is discrete in~$\coherentConfig_{M_C}$.

        We charge the individualization in~$S$ to a large fiber (there could be several) that is reachable via a path with internal degree~$2$ vertices and to the boundary fiber that is the penultimate vertex of that path. We conclude that per individualization at least~$8$ vertices in large fibers and at least~$4$ vertices in relevant small fibers are charged.
    \end{itemize}

    Overall, for all four types, per individualization either at  least~$8$ vertices in large fibers are charged and at least~$4$ vertices in small fibers or at least~$8.5$ vertices in large fibers and~$3$ vertices in small fibers.

    We argue that~$M\coloneqq \bigcup  \{M_C\mid C\text{ is a component of the small graph}\}$ satisfies the properties required by the lemma. Consider a component~$D$ of the quotient graph~$\coherentConfig_M$ that is not a singleton.

    \begin{itemize}
        \item
        Suppose~$D$ contains some vertex from a large fiber~$L$ of~$\coherentConfig$.
        The connected component~$C_L$ of~$\quotientGraphLarge{\coherentConfig}$ containing~$L$ has treewidth at most~$2$.
        Now consider the small fibers of~$\coherentConfig$ containing a vertex of~$D$. The small fibers in components of Type 1 are attached to a clique of~$C_L$ and have size at most 2.
        Being a clique sum, they in particular increase the treewidth to at most~$3$.

        The other small fibers (of Types 2--4) induce connected subgraphs of~$\quotientGraph{\coherentConfig}$ of maximum degree 2 which are attached to a single large fiber. Thus the component~$D$ has treewidth at most 3. The induced configuration thus has WL-dimension at most~$3\cdot t$, since every
        fiber has at most~$t$ points.

        It follows overall that the component~$D$ induces a coherent configuration of WL-dimension at most~$3\cdot t$ (Lemma~\ref{lem:bd:tw:and:fibre:size:bd:WL}).

        \item
        Suppose now that~$D$ does not contain a vertex contained in large fiber of~$\coherentConfig$. Recall that there were four types of components of the small graph, where the first three have degree at most~$2$ and thus treewidth at most~$2$ (Lemma~\ref{lem:max:degree:2:means:tw:3}), so they have WL-dimension bounded by~$2 \cdot t$ (Lemma~\ref{lem:bd:tw:and:fibre:size:bd:WL}). Suppose~$D$ is contained in the fourth type of component.
        Let~$r_1$ be the number of vertices of~$M_C$ that were added as part of the~$|S|=6$ case of components of Type 2. Let~$r_2$ be the number of remaining vertices in~$M_C$.
        When adding one vertex, respectively two vertices, to~$M_C$, vertices in large fibers are charged. Since every vertex is charged only once, we have~$r_1\cdot 8.5 +r_2\cdot 8\leq n_\ell$.
        Also, adding these vertices charges~$3$ or $4$ vertices of small fibers, respectively.
        These vertices are not part of the component~$D$, as they are either attached to a large fiber or become part of a component with only degree 2 vertices.

        Overall the number~$n_s$ of vertices contained in small components decreases by~$r_1\cdot 3 + r_2\cdot 4$.\qedhere
    \end{itemize}
\end{proof}

\begin{corollary}
\label{cor:of:main:global}
    Let~$\coherentConfig$ be a~$t$-reduced critical coherent configuration. Suppose~$Par(\coherentConfig)=(n_\ell, k_\ell,n_s)$ are the parameters of~$\coherentConfig$.
    Then
    \[
        \wldim{\coherentConfig}\leq \frac{2}{20} n_\ell + \frac{1}{20}n_s + \mathcal{O}(t) + o(n_s) \leq \tau(n_\ell, k_\ell,n_s) + \mathcal{O}(t) +  o(n_s).
    \]
\end{corollary}
\begin{proof}
    By the previous theorem we have~$\wldim{\coherentConfig}\leq r_1+r_2+ \widetilde{\f}(\tau(0,0,n_s-r_1\cdot 3 - r_2\cdot 4)) +3t$, where~$r_1\cdot 8.5 +r_2\cdot 8\leq n_\ell$ for non-negative integers~$r_1,r_2$.
    Using our bound for graphs with only small fibers (Theorem~\ref{small-cc:wldim/thm}) we obtain~$\wldim{\coherentConfig}
    \leq r_1+r_2+ \frac{1}{20} (n_s-  r_1\cdot 3 - r_2\cdot 4) + o(n_s) + 3t$.

    We maximize the function~$ r_1+r_2 + \frac{1}{20} (n_s-  r_1\cdot 3 - r_2\cdot 4) = \frac{2}{20} (r_1\cdot 8.5  + r_2\cdot 8) + \frac{1}{20} n_s$ under the condition that~$r_1\cdot 8.5 +r_2\cdot 8\leq n_\ell$.
    We obtain a value of~$\frac{2}{20} n_\ell + \frac{1}{20} n_s$. Since~$k_\ell\leq n_\ell/8$ we can bound this by~$\tau(n_\ell, k_\ell,n_s)$.
\end{proof}


\section{Proof of the main theorem}
\label{sec:proof:of:main:thm}

Finally, we are able to prove this chapter's main theorem.
To this end, we combine the local reduction with the previous section's argument.

\begin{theorem}
    Let~$\coherentConfig$ be a coherent configuration in which every fiber has size at most~$t$ with parameters $\parameters(\coherentConfig)=(n_\ell, k_\ell,n_s)$. Then~$\wldim{\coherentConfig}\leq \tau(n_\ell, k_\ell,n_s) + \mathcal{O}(t) +  o(n_s)$.
\end{theorem}
\begin{proof}
    We show that~$\widetilde{\f}(\tau(n_\ell, k_\ell,n_s))\leq \tau(n_\ell, k_\ell,n_s)+o(n_s)$ by performing an induction on~$20\cdot \tau(n_\ell, k_\ell,n_s)$, which is a function into the integers.
    By induction we may assume that~$\coherentConfig$ is critical.

    For the base case~$(0,0,0)$ there is nothing to show. Thus suppose~$\coherentConfig$ has at least one vertex.

    We consider the three options of Lemma~\ref{lem:properties:or:t:reduced}.
    \begin{itemize}
        \item
        If~$\wldim{\coherentConfig} \leq 10 + 2\cdot t$, then~$\wldim{\coherentConfig} \leq  \mathcal{O}(t) + o(n_s)$.

        \item
        If~$\coherentConfig$ is~$t$-reduced, the statement follows by Corollary~\ref{cor:of:main:global}.

        \item
        The last option is that~$\wldim{\coherentConfig}\leq x+\widetilde{\f}( \tau(Par(\coherentConfig))-\hat{x})$ for positive numbers~$x,\hat{x}$ with~$\hat{x}\geq x$. The statement follows by induction hypothesis.\qedhere
    \end{itemize}
\end{proof}

\begin{theorem}
\label{main-theorem/thm}
    Let~$\coherentConfig$ be a coherent configuration on~$n$ vertices.
    Then~$\wldim{\coherentConfig}\leq 3/20 \cdot n + o(n)$.
\end{theorem}
\begin{proof}
    We choose~$t$ depending on~$n$ such that $t(n) \in o(n)$ and $t(n) \in \omega(1)$.
    Suppose that~$\coherentConfig$ has parameters~$Par(\coherentConfig)=(n_\ell, k_\ell,n_s)$.
    By Theorem~\ref{lem:bound-on-cc-size} there is a refinement~$\coherentConfig'$ of~$\coherentConfig$ such that~$\wldim{\coherentConfig}\leq \wldim{\coherentConfig'} + o(n)$ and every fiber in~$\coherentConfig'$ has size at most~$t$.
    By the previous theorem, we have that
    \[
        \wldim{\coherentConfig'} \leq
        \tau(n_\ell, k_\ell,n_s) +  \mathcal{O}(t) + o(n_s)\leq
        \frac{3}{20}(n_\ell+n_s) +  \mathcal{O}(t) +  o(n_s)
        \leq \frac{3}{20} n + o(n).\qedhere
    \]
\end{proof}


\section{Lower bound}
\label{lower-bound/sec}

In this section, we assemble several known results to obtain an improved lower bound for the maximum WL-dimension of graphs in terms of their order.

\begin{theorem}
    A random cubic graph asymptotically almost surely has a treewidth of at {\linebreak}least~$ 0.042011151 \cdot n -1$.
\end{theorem}
\begin{proof}
    By~\cite{DBLP:journals/siamdm/KolesnikW14} a random cubic graph asymptotically almost surely has vertex expansion (i.e., vertex-isoperimetric number) at least~$\alpha\geq 0.144208556$ (this is~$A_3(1/2)$ in~\cite{DBLP:journals/siamdm/KolesnikW14}).
    This implies that a random cubic graph asymptotically almost surely has treewidth at least~$\frac{\alpha}{3(1+\alpha)} n -1 \geq  0.042011151 \cdot  n -1$~\cite[Corollary 7]{DBLP:journals/siamdm/DvorakN16}.
\end{proof}

It is well known that cubic graphs of high treewidth yield graphs with high WL-dimension via the CFI-construction.
Specifically, we have the following relation.

\begin{theorem}[Consequence of {\cite[Theorem 3]{DBLP:conf/csl/DawarR07}}]
    The CFI-graph~$\cfi{G}$ with base graph~$G$ satisfies~$\wldim{\cfi{G}}\geq \treewidth(G)$.
\end{theorem}

We remark that there are two versions of the CFI-constructions used in the literature. One, where for cubic graphs each vertex is replaced by a gadget of order 10 and one, where for cubic graphs each vertex is replaced by a gadget of order 4. (See~\cite{DBLP:conf/icalp/Furer01,DBLP:conf/esa/NeuenS17,tuprints24244} for more information.)
These versions are very similar and the theorem, as well as many other theorems, hold for either of them.
The difference between the constructions is that the former produces CFI-graphs of order~$10 |G|$, while the latter produces graphs of order~$4|G|$. In the terminology of our current paper, the first version produces coherent configurations with fibers of size 2, so a non-critical configuration. Removal of the tiny fibers yields the other construction.

The two theorems combine as follows.

\begin{corollary}
    The maximum Weisfeiler-Leman dimension for graphs of order~$n$ is at least~$0.0105027 \cdot n - o(n)$.
\end{corollary}

In the light of our discussions on configurations with small fiber size we remark that, due to the nature of the CFI-construction, the statement is also true for graphs of color class size 4.
    \bibliography{bib/bound-wldim}

\newcommand{\etalchar}[1]{$^{#1}$}
\begin{thebibliography}{MLM{\etalchar{+}}23}

\bibitem[AKRV17]{DBLP:journals/cc/ArvindKRV17}
Vikraman Arvind, Johannes K{\"o}bler, Gaurav Rattan, and Oleg Verbitsky.
\newblock Graph isomorphism, color refinement, and compactness.
\newblock {\em computational complexity}, 26(3):627--685, 2017.
\newblock \href {https://doi.org/10.1007/s00037-016-0147-6}
  {\path{doi:10.1007/s00037-016-0147-6}}.

\bibitem[AM13]{DBLP:journals/siamcomp/AtseriasM13}
Albert Atserias and Elitza~N. Maneva.
\newblock Sherali-adams relaxations and indistinguishability in counting
  logics.
\newblock {\em {SIAM} Journal on Computing}, 42(1):112--137, 2013.
\newblock \href {https://doi.org/10.1137/120867834}
  {\path{doi:10.1137/120867834}}.

\bibitem[AMR{\etalchar{+}}19]{DBLP:journals/jct/AtseriasMRSSV19}
Albert Atserias, Laura Mancinska, David~E. Roberson, Robert S{\'{a}}mal, Simone
  Severini, and Antonios Varvitsiotis.
\newblock Quantum and non-signalling graph isomorphisms.
\newblock {\em Journal of Combinatorial Theory, Series B}, 136:289--328, 2019.
\newblock \href {https://doi.org/10.1016/J.JCTB.2018.11.002}
  {\path{doi:10.1016/J.JCTB.2018.11.002}}.

\bibitem[Bab79]{Ba79b}
L{\'{a}}szl{\'{o}} Babai.
\newblock Lectures on graph isomorphism.
\newblock Mimeographed lecture notes, Department of Computer Science,
  University of Toronto, October 1979.
\newblock Lecture notes.

\bibitem[Bab81]{DBLP:conf/fct/Babai81}
L{\'{a}}szl{\'{o}} Babai.
\newblock Moderately exponential bound for graph isomorphism.
\newblock In {\em Fundamentals of Computation Theory, FCT'81, Proceedings of
  the 1981 International FCT-Conference, Szeged, Hungary, August 24-28, 1981},
  volume 117 of {\em Lecture Notes in Computer Science}, pages 34--50.
  Springer, 1981.
\newblock \href {https://doi.org/10.1007/3-540-10854-8\_4}
  {\path{doi:10.1007/3-540-10854-8\_4}}.

\bibitem[Bab16]{DBLP:conf/stoc/Babai16}
L{\'{a}}szl{\'{o}} Babai.
\newblock Graph isomorphism in quasipolynomial time [extended abstract].
\newblock In {\em Proceedings of the 48th Annual {ACM} {SIGACT} Symposium on
  Theory of Computing, {STOC} 2016, Cambridge, MA, USA, June 18-21, 2016},
  pages 684--697. {ACM}, 2016.
\newblock \href {https://doi.org/10.1145/2897518.2897542}
  {\path{doi:10.1145/2897518.2897542}}.

\bibitem[CFI92]{DBLP:journals/combinatorica/CaiFI92}
Jin{-}yi Cai, Martin F{\"{u}}rer, and Neil Immerman.
\newblock An optimal lower bound on the number of variables for graph
  identification.
\newblock {\em Combinatorica}, 12(4):389--410, 1992.
\newblock \href {https://doi.org/10.1007/BF01305232}
  {\path{doi:10.1007/BF01305232}}.

\bibitem[CP19]{CC}
Gang Chen and Ilia Ponomarenko.
\newblock {\em Lectures on Coherent Configurations}.
\newblock Central China Normal University Press, 2019.

\bibitem[DGR18]{DBLP:conf/icalp/DellGR18}
Holger Dell, Martin Grohe, and Gaurav Rattan.
\newblock Lov{\'{a}}sz meets {Weisfeiler} and {Leman}.
\newblock In {\em 45th International Colloquium on Automata, Languages, and
  Programming, {ICALP} 2018, July 9-13, 2018, Prague, Czech Republic}, volume
  107 of {\em LIPIcs}, pages 40:1--40:14. Schloss Dagstuhl - Leibniz-Zentrum
  f{\"{u}}r Informatik, 2018.
\newblock \href {https://doi.org/10.4230/LIPICS.ICALP.2018.40}
  {\path{doi:10.4230/LIPICS.ICALP.2018.40}}.

\bibitem[DN16]{DBLP:journals/siamdm/DvorakN16}
Zdenek Dvor{\'{a}}k and Sergey Norin.
\newblock Strongly sublinear separators and polynomial expansion.
\newblock {\em {SIAM} Journal on Discrete Mathematics}, 30(2):1095--1101, 2016.
\newblock \href {https://doi.org/10.1137/15M1017569}
  {\path{doi:10.1137/15M1017569}}.

\bibitem[DR07]{DBLP:conf/csl/DawarR07}
Anuj Dawar and David Richerby.
\newblock The power of counting logics on restricted classes of finite
  structures.
\newblock In {\em Computer Science Logic, 21st International Workshop, {CSL}
  2007, 16th Annual Conference of the EACSL, Lausanne, Switzerland, September
  11-15, 2007, Proceedings}, volume 4646 of {\em Lecture Notes in Computer
  Science}, pages 84--98. Springer, 2007.
\newblock \href {https://doi.org/10.1007/978-3-540-74915-8\_10}
  {\path{doi:10.1007/978-3-540-74915-8\_10}}.

\bibitem[Dvo10]{DBLP:journals/jgt/Dvorak10}
Zdenek Dvor{\'{a}}k.
\newblock On recognizing graphs by numbers of homomorphisms.
\newblock {\em Journal of Graph Theory}, 64(4):330--342, 2010.
\newblock \href {https://doi.org/10.1002/JGT.20461}
  {\path{doi:10.1002/JGT.20461}}.

\bibitem[EPT00]{Ponomarenko2000interval}
Sergei Evdokimov, Ilia Ponomarenko, and Gottfried Tinhofer.
\newblock Forestal algebras and algebraic forests (on a new class of weakly
  compact graphs).
\newblock {\em Discrete Mathematics}, 225(1):149--172, 2000.
\newblock FPSAC'98.
\newblock \href {https://doi.org/10.1016/S0012-365X(00)00152-7}
  {\path{doi:10.1016/S0012-365X(00)00152-7}}.

\bibitem[FH06]{pathwidthCubicGraphs}
Fedor~V. Fomin and Kjartan H{\o}ie.
\newblock Pathwidth of cubic graphs and exact algorithms.
\newblock {\em Information Processing Letters}, 97(5):191 -- 196, 2006.
\newblock \href {https://doi.org/10.1016/j.ipl.2005.10.012}
  {\path{doi:10.1016/j.ipl.2005.10.012}}.

\bibitem[FKV21]{DBLP:journals/siamdm/FuhlbruckKV21}
Frank Fuhlbr{\"{u}}ck, Johannes K{\"{o}}bler, and Oleg Verbitsky.
\newblock Identifiability of graphs with small color classes by the
  {Weisfeiler-Leman} algorithm.
\newblock {\em {SIAM} Journal on Discrete Mathematics}, 35(3):1792--1853, 2021.
\newblock \href {https://doi.org/10.1137/20M1327550}
  {\path{doi:10.1137/20M1327550}}.

\bibitem[F{\"{u}}r01]{DBLP:conf/icalp/Furer01}
Martin F{\"{u}}rer.
\newblock Weisfeiler-lehman refinement requires at least a linear number of
  iterations.
\newblock In {\em Automata, Languages and Programming, 28th International
  Colloquium, {ICALP} 2001, Crete, Greece, July 8-12, 2001, Proceedings},
  volume 2076 of {\em Lecture Notes in Computer Science}, pages 322--333.
  Springer, 2001.
\newblock \href {https://doi.org/10.1007/3-540-48224-5\_27}
  {\path{doi:10.1007/3-540-48224-5\_27}}.

\bibitem[GGP24]{DBLP:journals/corr/abs-2305-15861}
Jin Guo, Alexander~L Gavrilyuk, and Ilia Ponomarenko.
\newblock On the {Weisfeiler-Leman} dimension of permutation graphs.
\newblock {\em {SIAM} Journal on Discrete Mathematics}, 38(2):1915--1929, 2024.
\newblock \href {https://doi.org/10.1137/23M1575019}
  {\path{doi:10.1137/23M1575019}}.

\bibitem[GN23]{DBLP:journals/tocl/GroheN23}
Martin Grohe and Daniel Neuen.
\newblock Canonisation and definability for graphs of bounded rank width.
\newblock {\em {ACM} Transactions on Computational Logic}, 24(1):6:1--6:31,
  2023.
\newblock \href {https://doi.org/10.1145/3568025} {\path{doi:10.1145/3568025}}.

\bibitem[GNP23]{DBLP:journals/gc/GavrilyukNP23}
Alexander~L. Gavrilyuk, Roman Nedela, and Ilia Ponomarenko.
\newblock The {Weisfeiler-Leman} dimension of distance-hereditary graphs.
\newblock {\em Graphs and Combinatorics}, 39(4):84, 2023.
\newblock \href {https://doi.org/10.1007/S00373-023-02683-3}
  {\path{doi:10.1007/S00373-023-02683-3}}.

\bibitem[GO15]{DBLP:journals/jsyml/GroheO15}
Martin Grohe and Martin Otto.
\newblock Pebble games and linear equations.
\newblock {\em The Journal of Symbolic Logic}, 80(3):797--844, 2015.
\newblock \href {https://doi.org/10.1017/JSL.2015.28}
  {\path{doi:10.1017/JSL.2015.28}}.

\bibitem[Gol83]{DBLP:journals/dam/Goldberg83}
Mark~K. Goldberg.
\newblock A nonfactorial algorithm for testing isomorphism of two graphs.
\newblock {\em Discrete Applied Mathematics}, 6(3):229--236, 1983.
\newblock \href {https://doi.org/10.1016/0166-218X(83)90078-1}
  {\path{doi:10.1016/0166-218X(83)90078-1}}.

\bibitem[Gro17]{DBLP:books/cu/G2017}
Martin Grohe.
\newblock {\em Descriptive Complexity, Canonisation, and Definable Graph
  Structure Theory}, volume~47 of {\em Lecture Notes in Logic}.
\newblock Cambridge University Press, 2017.
\newblock \href {https://doi.org/10.1017/9781139028868}
  {\path{doi:10.1017/9781139028868}}.

\bibitem[GSW21]{DBLP:conf/soda/GroheSW21}
Martin Grohe, Pascal Schweitzer, and Daniel Wiebking.
\newblock {Deep Weisfeiler Leman}.
\newblock In {\em Proceedings of the 2021 {ACM-SIAM} Symposium on Discrete
  Algorithms, {SODA} 2021, Virtual Conference, January 10 - 13, 2021}, pages
  2600--2614. {SIAM}, 2021.
\newblock \href {https://doi.org/10.1137/1.9781611976465.154}
  {\path{doi:10.1137/1.9781611976465.154}}.

\bibitem[HM00]{MiyamotoHanaki2000}
Akihide Hanaki and Izumi Miyamoto.
\newblock Classification of primitive association schemes of order up to 22.
\newblock {\em Kyushu Journal of Mathematics}, 54(1):81--86, 2000.
\newblock \href {https://doi.org/10.2206/kyushujm.54.81}
  {\path{doi:10.2206/kyushujm.54.81}}.

\bibitem[HM03]{DBLP:journals/dm/HanakiM03}
Akihide Hanaki and Izumi Miyamoto.
\newblock Classification of association schemes of small order.
\newblock {\em Discrete Mathematics}, 264(1-3):75--80, 2003.
\newblock \href {https://doi.org/10.1016/S0012-365X(02)00551-4}
  {\path{doi:10.1016/S0012-365X(02)00551-4}}.

\bibitem[IL90]{MR1060782}
Neil Immerman and Eric Lander.
\newblock {\em Describing Graphs: A First-Order Approach to Graph
  Canonization}, pages 59--81.
\newblock Springer New York, New York, NY, 1990.
\newblock \href {https://doi.org/10.1007/978-1-4612-4478-3_5}
  {\path{doi:10.1007/978-1-4612-4478-3_5}}.

\bibitem[KN24]{DBLP:journals/corr/abs-2402-03274}
Sandra Kiefer and Daniel Neuen.
\newblock Bounding the {Weisfeiler-Leman} dimension via a depth analysis of
  {I/R}-trees.
\newblock In {\em Proceedings of the 39th Annual {ACM/IEEE} Symposium on Logic
  in Computer Science, {LICS} 2024, Tallinn, Estonia, July 8-11, 2024}, pages
  50:1--50:14. {ACM}, 2024.
\newblock \href {https://doi.org/10.1145/3661814.3662122}
  {\path{doi:10.1145/3661814.3662122}}.

\bibitem[KPS19]{DBLP:journals/jacm/KieferPS19}
Sandra Kiefer, Ilia Ponomarenko, and Pascal Schweitzer.
\newblock The {Weisfeiler-Leman} dimension of planar graphs is at most 3.
\newblock {\em Journal of the {ACM}}, 66(6):44:1--44:31, 2019.
\newblock \href {https://doi.org/10.1145/3333003} {\path{doi:10.1145/3333003}}.

\bibitem[KSS22]{DBLP:journals/tocl/KieferSS22}
Sandra Kiefer, Pascal Schweitzer, and Erkal Selman.
\newblock Graphs identified by logics with counting.
\newblock {\em {ACM} Transactions on Computational Logic}, 23(1):1:1--1:31,
  2022.
\newblock \href {https://doi.org/10.1145/3417515} {\path{doi:10.1145/3417515}}.

\bibitem[KW14]{DBLP:journals/siamdm/KolesnikW14}
Brett Kolesnik and Nick Wormald.
\newblock Lower bounds for the isoperimetric numbers of random regular graphs.
\newblock {\em {SIAM} Journal on Discrete Mathematics}, 28(1):553--575, 2014.
\newblock \href {https://doi.org/10.1137/120891265}
  {\path{doi:10.1137/120891265}}.

\bibitem[Lic23]{tuprints24244}
Moritz Lichter.
\newblock {\em Continuing the Quest for a Logic Capturing Polynomial Time -
  Potential, Limitations, and Interplay of Current Approaches}.
\newblock PhD thesis, Technische Universit{\"a}t Darmstadt, Darmstadt, July
  2023.
\newblock \href {https://doi.org/10.26083/tuprints-00024244}
  {\path{doi:10.26083/tuprints-00024244}}.

\bibitem[LRS25]{DBLP:conf/csl/LichterRS25}
Moritz Lichter, Simon Ra{\ss}mann, and Pascal Schweitzer.
\newblock Computational complexity of the {Weisfeiler-Leman} dimension.
\newblock In {\em 33rd {EACSL} Annual Conference on Computer Science Logic,
  {CSL} 2025, February 10-14, 2025, Amsterdam, Netherlands}, volume 326 of {\em
  LIPIcs}, pages 13:1--13:22. Schloss Dagstuhl - Leibniz-Zentrum f{\"{u}}r
  Informatik, 2025.
\newblock \href {https://doi.org/10.4230/LIPICS.CSL.2025.13}
  {\path{doi:10.4230/LIPICS.CSL.2025.13}}.

\bibitem[LS24]{DBLP:journals/jacm/LichterS24}
Moritz Lichter and Pascal Schweitzer.
\newblock Choiceless polynomial time with witnessed symmetric choice.
\newblock {\em Journal of the {ACM}}, 71(2):7:1--7:70, 2024.
\newblock \href {https://doi.org/10.1145/3648104} {\path{doi:10.1145/3648104}}.

\bibitem[MLM{\etalchar{+}}23]{Morrisetal2023}
Christopher Morris, Yaron Lipman, Haggai Maron, Bastian Rieck, Nils~M. Kriege,
  Martin Grohe, Matthias Fey, and Karsten~M. Borgwardt.
\newblock {Weisfeiler and Leman} go machine learning: The story so far.
\newblock {\em Journal of Machine Learning Research}, 24:1--59, 2023.
\newblock URL: \url{http://jmlr.org/papers/v24/22-0240.html}.

\bibitem[NS17]{DBLP:conf/esa/NeuenS17}
Daniel Neuen and Pascal Schweitzer.
\newblock Benchmark graphs for practical graph isomorphism.
\newblock In {\em 25th Annual European Symposium on Algorithms, {ESA} 2017,
  September 4-6, 2017, Vienna, Austria}, volume~87 of {\em LIPIcs}, pages
  60:1--60:14. Schloss Dagstuhl - Leibniz-Zentrum f{\"{u}}r Informatik, 2017.
\newblock \href {https://doi.org/10.4230/LIPICS.ESA.2017.60}
  {\path{doi:10.4230/LIPICS.ESA.2017.60}}.

\bibitem[OS14]{DBLP:conf/swat/OtachiS14}
Yota Otachi and Pascal Schweitzer.
\newblock Reduction techniques for graph isomorphism in the context of width
  parameters.
\newblock In {\em Algorithm Theory - {SWAT} 2014 - 14th Scandinavian Symposium
  and Workshops, Copenhagen, Denmark, July 2-4, 2014. Proceedings}, volume 8503
  of {\em Lecture Notes in Computer Science}, pages 368--379. Springer, 2014.
\newblock \href {https://doi.org/10.1007/978-3-319-08404-6\_32}
  {\path{doi:10.1007/978-3-319-08404-6\_32}}.

\bibitem[PV09]{DBLP:conf/asl/PikhurkoV09}
Oleg Pikhurko and Oleg Verbitsky.
\newblock Logical complexity of graphs: {A} survey.
\newblock In {\em Model Theoretic Methods in Finite Combinatorics - {AMS-ASL}
  Joint Special Session, Washington, DC, USA, January 5-8, 2009}, volume 558 of
  {\em Contemporary Mathematics}, pages 129--180. American Mathematical
  Society, 2009.

\bibitem[PVV06]{DBLP:journals/dam/PikhurkoVV06}
Oleg Pikhurko, Helmut Veith, and Oleg Verbitsky.
\newblock The first order definability of graphs: Upper bounds for quantifier
  depth.
\newblock {\em Discrete Applied Mathematics}, 154(17):2511--2529, 2006.
\newblock \href {https://doi.org/10.1016/J.DAM.2006.03.002}
  {\path{doi:10.1016/J.DAM.2006.03.002}}.

\bibitem[Sch17]{DBLP:journals/mst/Schweitzer17}
Pascal Schweitzer.
\newblock Towards an isomorphism dichotomy for hereditary graph classes.
\newblock {\em Theory of Computing Systems}, 61(4):1084--1127, 2017.
\newblock \href {https://doi.org/10.1007/S00224-017-9775-8}
  {\path{doi:10.1007/S00224-017-9775-8}}.

\bibitem[Wei76]{MR0543783}
Boris Weisfeiler.
\newblock {\em On construction and identification of graphs}.
\newblock Lecture Notes in Mathematics, Vol. 558. Springer-Verlag, Berlin-New
  York, 1976.

\bibitem[WP24]{wu2024weisfeiler}
Yulai Wu and Ilia Ponomarenko.
\newblock On the {Weisfeiler-Leman} dimension of circulant graphs, 2024.
\newblock \href {https://doi.org/10.48550/arXiv.2406.15822}
  {\path{doi:10.48550/arXiv.2406.15822}}.

\end{thebibliography}

\end{document}